\documentclass[oneside,english]{amsart}
\usepackage[T1]{fontenc}
\usepackage[latin9]{inputenc}
\usepackage{geometry}
\usepackage{verbatim}
\usepackage{amsthm}
\usepackage{amstext}
\usepackage{amssymb}
\usepackage{graphicx}
\usepackage{esint}
\usepackage{url}
\usepackage{yhmath}
\usepackage[all]{xy}
\usepackage{color}

\usepackage{shuffle}

\makeatletter
\numberwithin{equation}{section}
\numberwithin{figure}{section}
\theoremstyle{plain}
\newtheorem*{thm*}{\protect\theoremname}
\theoremstyle{plain}
\newtheorem{thm}{\protect\theoremname}[section]
\theoremstyle{plain}
\newtheorem{lem}[thm]{\protect\lemmaname}
\theoremstyle{remark}
\newtheorem{rem}[thm]{\protect\remarkname}
\theoremstyle{plain}

\theoremstyle{plain}
\newtheorem*{prop*}{\protect\propositionname}
\theoremstyle{plain}
\newtheorem{prop}[thm]{\protect\propositionname}
\theoremstyle{plain}
\newtheorem*{cor*}{\protect\corollaryname}
\theoremstyle{plain}
\newtheorem{cor}[thm]{\protect\corollaryname}
\theoremstyle{plain}

\usepackage{mathrsfs}
\usepackage{subfigure}

\AtBeginDocument{

}

\makeatother

\usepackage{babel}
\providecommand{\corollaryname}{Corollary}
\providecommand{\lemmaname}{Lemma}
\providecommand{\propositionname}{Proposition}
\providecommand{\remarkname}{Remark}
\providecommand{\theoremname}{Theorem}
\providecommand{\conjecturename}{Conjecture}
\providecommand{\notename}{Note}

\begin{document}
\global\long\def\SLE{\mathrm{SLE}}
\global\long\def\SLEk{\mathrm{SLE}_{\kappa}}

\global\long\def\CLE{\mathrm{CLE}}
\global\long\def\CLEk{\mathrm{CLE}_{\kappa}}

\global\long\def\SLEkappa#1{\mathrm{SLE}_{#1}}
\global\long\def\SLEkapparho#1#2{\mathrm{SLE}_{#1}(#2)}
\global\long\def\SLEmeasure{\mathsf{P}}
\global\long\def\chordal{\oslash}
\global\long\def\driving{D}

\global\long\def\PR{\mathbb{P}}
\global\long\def\EX{\mathbb{E}}

\global\long\def\sF{\mathcal{F}}
\global\long\def\sZ{\mathcal{Z}}
\global\long\def\sD{\mathcal{D}}
\global\long\def\sC{\mathcal{C}}
\global\long\def\sL{\mathcal{L}}
\global\long\def\sA{\mathcal{A}}
\global\long\def\sR{\mathcal{R}}
\global\long\def\sS{\mathcal{S}}
\global\long\def\sP{\mathcal{P}}
\global\long\def\sM{\mathcal{M}}

\global\long\def\bR{\mathbb{R}}
\global\long\def\bRpos{\mathbb{R}_{> 0}}
\global\long\def\bRnn{\mathbb{R}_{\geq 0}}
\global\long\def\bZ{\mathbb{Z}}
\global\long\def\bN{\mathbb{N}}
\global\long\def\bZpos{\mathbb{Z}_{> 0}}
\global\long\def\bZnn{\mathbb{Z}_{\geq 0}}
\global\long\def\bQ{\mathbb{Q}}
\global\long\def\bC{\mathbb{C}}

\global\long\def\Rsphere{\overline{\bC}}
\global\long\def\bD{\mathbb{D}}
\global\long\def\bH{\mathbb{H}}
\global\long\def\re{\Re\mathfrak{e}}
\global\long\def\im{\Im\mathfrak{m}}
\global\long\def\arg{\mathrm{arg}}
\global\long\def\ii{\mathfrak{i}}
\global\long\def\domain{\Lambda}
\global\long\def\bdrypt{\xi}
\global\long\def\bdryptb{\eta}
\global\long\def\bdry{\partial}
\global\long\def\cl#1{\overline{#1}}
\global\long\def\Mob{\mu}
\global\long\def\confmap{\phi}
\global\long\def\zbar{\bar{z}}

\global\long\def\OO{\mathcal{O}}
\global\long\def\oo{\mathit{o}}

\global\long\def\ud{\mathrm{d}}
\global\long\def\der#1{\frac{\ud}{\ud#1}}
\global\long\def\pder#1{\frac{\partial}{\partial#1}}
\global\long\def\pdder#1{\frac{\partial^{2}}{\partial#1^{2}}}
\global\long\def\pddder#1{\frac{\partial^{3}}{\partial#1^{3}}}

\global\long\def\set#1{\left\{  #1\right\}  }
\global\long\def\setcond#1#2{\left\{  #1\;\big|\;#2\right\}  }

\global\long\def\sl{\mathfrak{sl}}
\global\long\def\sltwo{\mathfrak{sl}_2}
\global\long\def\Uqsltwo{\mathcal{U}_{q}(\mathfrak{sl}_{2})}
\global\long\def\Hcp{\Delta}
\global\long\def\qnum#1{\left[#1\right] }
\global\long\def\qfact#1{\left[#1\right]! }
\global\long\def\qbin#1#2{\left[\begin{array}{c}
	#1\\
	#2 
	\end{array}\right]}
	
\global\long\def\Wd{\mathsf{M}}
\global\long\def\HWsp{\mathsf{H}}

\global\long\def\Wbas{e}
\global\long\def\Tbas{\tau}
\global\long\def\MTbas{\theta}
\global\long\def\Sbas{\mathsf{s}}
\global\long\def\TRbas{\mathsf{t}}
\global\long\def\multiplicity{D}
\global\long\def\Wbastwodim{\Wbas^{(2)}}

\global\long\def\Puregeomtwodim{v}
\global\long\def\Puregeom{\mathfrak{v}}
\global\long\def\Bdryvec{\Puregeom}

\global\long\def\Singletbastwodim{\mathsf{w}}
\global\long\def\Singletbas{\mathfrak{w}}

\global\long\def\Coblobastwodim{\mathsf{u}}
\global\long\def\Coblobas{\mathfrak{u}}

\global\long\def\Projection{\mathfrak{p}}
\global\long\def\Projectionhat{\widehat{\Projection}}
\global\long\def\Embedding{\mathfrak{I}}
\global\long\def\projdmn{\delta}
\global\long\def\ImgofEmbedding{\mathsf{J}}

\global\long\def\TL{\mathrm{TL}}
\global\long\def\fugacity{\nu}

\global\long\def\Rmatrix{\sR}
\global\long\def\RmatrixOperator{\mathscr{R}}


\global\long\def\LP{\mathrm{LP}}
\global\long\def\PP{\mathrm{PP}}
\global\long\def\nested{\boldsymbol{\underline{\Cap}}}
\global\long\def\unnested{\boldsymbol{\underline{\cap\cap}}}
\global\long\def\walks{\mathcal{W}}
\global\long\def\Catalan{\mathrm{C}}
\global\long\def\constantfromdiagram#1#2#3{\mathcal{C}({#1;#2,#3})}
\global\long\def\Fusionconstant#1#2#3{\mathfrak{C_{singl}}({#1;#2,#3})}
\global\long\def\Rpluscomb{\mathcal{R}_+}
\global\long\def\Rminuscomb{\mathcal{R}_-}
\global\long\def\Rpmcomb{\mathcal{R}_\pm}
\global\long\def\Scomb{\mathcal{S}}
\global\long\def\Embeddingcomb{\mathcal{I}}

\global\long\def\linkpatt{\omega}
\global\long\def\partition{\lambda}
\global\long\def\defendpt{u}

\newcommand{\defects}{%
  \raisebox{-.5ex}{%
    \scalebox{1.5}{%
      \rotatebox[origin=c]{270}{$\exists$}%
    }%
  }%
}

\global\long\def\defpatt{\shuffle}

\global\long\def\linkInEquation#1#2{\underset{\;#1\;\;#2}{
\vcenter{\hbox{\includegraphics[scale=0.5]{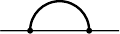}}}}}
\global\long\def\linksmallerInEquation#1#2{\hspace{-1mm}
\underset{\scriptscriptstyle{\,#1\,#2}}{
\vcenter{\hbox{\includegraphics[scale=0.2]{link-0.pdf}}}}}
\global\long\def\defectInEquation#1{\underset{#1}{
\vcenter{\hbox{\includegraphics[scale=0.5]{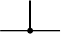}}}}}

\global\long\def\link#1#2{\raisebox{.5ex}{\hspace{-1mm}
\scalebox{.75}{$\linkInEquation{\boldsymbol{#1}}{\boldsymbol{#2}}$}}}

\global\long\def\defect#1{\raisebox{.5ex}{\hspace{-1mm}
\scalebox{.75}{$\defectInEquation{\boldsymbol{#1}}$}}}

\global\long\def\sciOp{\text{\Rightscissors}}
\global\long\def\tieOp{\wp}
\global\long\def\removeLink{/}

\global\long\def\Mmat{\mathscr{M}}
\global\long\def\Minv{\mathscr{M}^{-1}}
\global\long\def\genMmat{\mathfrak{M}}
\global\long\def\genMinv{\mathfrak{M}^{-1}}
\global\long\def\KWleq{{\leftharpoondown}}
\global\long\def\ExcK{\mathsf{K}}
\global\long\def\ExcKH{\mathcal{K}}
\global\long\def\FominDet{\mathbf{\Delta}}
\global\long\def\LPdet#1#2{\FominDet_{#1}^{#2}}
\global\long\def\ConfBlockFun{\mathcal{U}}
\global\long\def\ConfBlockVec{\mathsf{u}}
\global\long\def\Gr{\mathcal{G}}
\global\long\def\Vert{\mathcal{V}}
\global\long\def\Edg{\mathcal{E}}

\global\long\def\tree{\mathcal{T}}
\global\long\def\branch{\gamma}

\global\long\def\LE{\mathrm{LE}}
\global\long\def\RW{\mathcal{X}}
\global\long\def\LERW{\mathcal{L}}

\global\long\def\SymmGrp{\mathfrak{S}}
\global\long\def\sgn{\mathrm{sgn}}

\global\long\def\PPP{\mathrm{PPP}}
\global\long\def\UST{\mathcal{T}}

\global\long\def\MobF{\mathfrak{m}}

\newcommand{\walkfromto}[3]{#1 \text{ from } #2 \text{ to } #3}
\newcommand{\pathfromto}[2]{#1 \rightsquigarrow #2}

\newcommand{\edgeof}[2]{{\langle #1 , #2 \rangle}}

\global\long\def\Hom{\mathrm{Hom}}
\global\long\def\End{\mathrm{End}}
\global\long\def\Aut{\mathrm{Aut}}
\global\long\def\Rad{\mathrm{Rad}}
\global\long\def\Ext{\mathrm{Ext}}
\global\long\def\Mat{\mathrm{Mat}}
\global\long\def\dmn{\mathrm{dim}}
\global\long\def\spn{\mathrm{span}}
\global\long\def\tens{\otimes}
\global\long\def\unitmat{\mathbb{I}}
\global\long\def\id{\mathrm{id}}
\global\long\def\isom{\cong}
\global\long\def\Kern{\mathrm{Ker}}
\global\long\def\SymmGrp{\mathfrak{S}}
\global\long\def\BrGrp{\mathfrak{Br}}
\global\long\def\PureBrGrp{\mathfrak{PBr}}
\global\long\def\Comm{\mathrm{Comm}}
\global\long\def\algebra{\mathfrak{A}}

\global\long\def\Witt{\mathfrak{Witt}}
\global\long\def\Vir{\mathfrak{Vir}}
\global\long\def\Hei{\mathfrak{H}}

\global\long\def\chamber{\mathfrak{X}}
\global\long\def\extendedChamber{\mathfrak{W}}
\global\long\def\FWint{\varphi}
\global\long\def\SurfFW{\mathfrak{L}^{\Supset}}
\global\long\def\PartF{\sZ}
\global\long\def\Sol{\sR}
\global\long\def\ConvSet{\sC}
\global\long\def\FKdual{\mathscr{L}}
\global\long\def\Quantumdual{\psi}

\global\long\def\Ampl{\zeta}
\global\long\def\Corr{\chi}
\global\long\def\Orders{\mathrm{VO}}

\global\long\def\braid{\sigma}
\global\long\def\Pf{\mathrm{Pf}}
\global\long\def\sgn{\mathrm{sgn}}
\global\long\def\dist{\mathrm{dist}}
\global\long\def\const{\mathrm{const.}}
\global\long\def\eps{\varepsilon}
\global\long\def\half{\frac{1}{2}}


\global\long\def\multii{\varsigma}

\global\long\def\LimitOp{\mathscr{L}}
\global\long\def\Projector{\mathscr{P}}
\global\long\def\Projectorhat{\widehat{\mathscr{P}}}

\global\long\def\SpinChain{\mathsf{V}}
\global\long\def\FusedSpinChain{\Projection\mathsf{V}}
\global\long\def\FusedSpinChainHat{\Projectionhat\mathsf{V}}

\global\long\def\FusedHWsp{\Projection\HWsp}

\global\long\def\NormalizationConstant{C}

\global\long\def\BasisF{\mathscr{F}}


\author{E.~Peltola}

\

\vspace{2.5cm}

\begin{center}
\LARGE \bf \scshape {Basis for solutions of the Benoit \& Saint-Aubin PDEs with particular asymptotics properties
}
\end{center}

\vspace{0.75cm}

\begin{center}
{\large \scshape Eveliina Peltola}\\
{\footnotesize{\tt eveliina.peltola@unige.ch}}\\
{\small{Section de Math\'{e}matiques, Universit\'{e} de Gen\`{e}ve,}}\\
{\small{2--4 rue du Li\`{e}vre, C.P. 64, 1211 Gen\`{e}ve 4, Switzerland}}
\end{center}

\vspace{0.75cm}

\begin{center}
\begin{minipage}{0.85\textwidth} \footnotesize
{\scshape Abstract.}
Applying the quantum group method developed 
in~\cite{Kytola-Peltola:Conformally_covariant_boundary_correlation_functions_with_quantum_group}, 
we construct solutions to the Benoit~\& Saint-Aubin partial differential 
equations with boundary conditions given by specific recursive asymptotics properties. Our results generalize solutions constructed 
in~\cite{Kytola-Peltola:Pure_partition_functions_of_multiple_SLEs, PW:Global_multiple_SLEs_and_pure_partition_functions},
known as
the 
pure partition functions of multiple Schramm-Loewner evolutions.
The generalization is reminiscent of fusion in conformal field theory, and our solutions 
can be thought of as partition functions of systems of random curves, where many curves may emerge from the same point. 
\end{minipage}
\end{center}

\vspace{0.75cm}
\setcounter{tocdepth}{2}

\bigskip{}
\section{\label{sec: intro}Introduction}

Conformal field theories (CFT) are expected to describe scaling limits 
of critical models of statistical mechanics. In particular, scaling limits of 
correlations in discrete critical systems should be CFT correlation functions.
Many correlation functions of interest satisfy linear homogeneous partial 
differential equations (PDEs), which in CFT arise from the presence of singular 
vectors in representations of the Virasoro 
algebra~\cite{BPZ:Infinite_conformal_symmetry_in_2D_QFT, Cardy:Conformal_invariance_and_surface_critical_behavior, 
Feigin-Fuchs:Verma_modules_over_Virasoro_book, DMS:CFT}.

Such PDEs of second order frequently appear also in the theory of Schramm-Loewner evolutions ($\SLE$).
In this probabilistic context, they arise from stochastic differentials of certain local martingales.
Solutions to systems of these second order PDEs are known as partition functions for multiple 
$\SLE$s~\cite{BBK:Multiple_SLEs_and_statistical_mechanics_martingales,
Dubedat:Commutation_relations_for_SLE, 
Kozdron-Lawler:Configurational_measure_on_mutually_avoiding_SLEs,
Kytola-Peltola:Pure_partition_functions_of_multiple_SLEs,
PW:Global_multiple_SLEs_and_pure_partition_functions}.
On the other hand, the higher order PDEs of CFT seem not to have a direct probabilistic interpretation, 
but can in some cases be understood in terms of scaling limits, as 
in~\cite{Gamsa-Cardy:The_scaling_limit_of_two_cluster_boundaries_in_critical_lattice_models, 
KKP:Conformal_blocks_pure_partition_functions_and_KW_binary_relation},
$\SLE$ observables, as in~\cite{BJV:Some_remarks_on_SLE_bubbles_and_Schramms_2point_observable,
LV:Coulomb_gas_for_commuting_SLEs, LV:Coulomb_gas_integrals_PART2},
or generalizations of multiple $\SLE$ 
measures~\cite{Friedrich-Werner:Conformal_restriction_highest_weight_representations_and_SLE,
Kontsevich:CFT_SLE_and_phase_boundaries,
Friedrich-Kalkkinen:On_CFT_and_SLE,
Kontsevich-Suhov:On_Malliavin_measures_SLE_and_CFT,
Dubedat:SLE_and_Virasoro_representations_fusion}. 
See also~\cite{BPZ:Infinite_conformal_symmetry_of_critical_fluctuations_in_2D, 
Cardy:Critical_percolation_in_finite_geometries,  
Watts:A_crossing_probability_for_critical_percolation_in_two_dimensions,
Bauer-Bernard:Conformal_field_theories_of_SLEs,
Bauer-Bernard:SLE_CFT_and_zigzag_probabilities, 
Dubedat:Euler_integrals_for_commuting_SLEs,
Dubedat:Excursion_decomposition_for_SLE_and_Watts_crossing_formula,
Sheffield-Wilson:Schramms_proof_of_Watts_formula,
Flores-Kleban:Solution_space_for_system_of_null-state_PDE4,
FSK-Multiple_SLE_connectivity_weights_for_rectangles_hexagons_and_octagons,
JJK:SLE_boundary_visits, PW:Global_multiple_SLEs_and_pure_partition_functions} for further examples.

In this article, we consider solutions to the class of Benoit \& Saint-Aubin type PDE 
systems~\cite{BSA:Degenerate_CFTs_and_explicit_expressions_for_some_null_vectors,
BDIZ:Covariant_differential_equations_and_singular_vectors_in_Virasoro_representations},
corresponding to singular vectors with conformal weights of type $h_{1,s}$ in the Kac table. 
We assume throughout that the parameter $\kappa$ in the central charge 
$c = c(\kappa) = \frac{1}{2\kappa}(6-\kappa)(3\kappa-8)$ is non-rational. 
We construct solutions which satisfy particular boundary conditions given in 
terms of specified asymptotic behavior, recursive in the number of variables. 
In the articles~\cite{JJK:SLE_boundary_visits, Kytola-Peltola:Pure_partition_functions_of_multiple_SLEs,
KKP:Conformal_blocks_pure_partition_functions_and_KW_binary_relation,
LV:Coulomb_gas_for_commuting_SLEs, LV:Coulomb_gas_integrals_PART2, PW:Global_multiple_SLEs_and_pure_partition_functions},
examples of such functions were constructed for applications 
concerning Schramm-Loewner evolutions (see also~\cite{KKP:Conformal_blocks} for solutions relevant in CFT). 
In these applications, the choice of boundary conditions is motivated 
by properties of the random curves, which in CFT language means
specific fusion channels for the fields;
see also~\cite{Cardy:Boundary_conditions_fusion_rules_and_Verlinde_formula,
Cardy:Critical_percolation_in_finite_geometries, 
Watts:A_crossing_probability_for_critical_percolation_in_two_dimensions,
Bauer-Bernard:Conformal_field_theories_of_SLEs,
Bauer-Bernard:SLE_CFT_and_zigzag_probabilities, 
BBK:Multiple_SLEs_and_statistical_mechanics_martingales,
Gamsa-Cardy:The_scaling_limit_of_two_cluster_boundaries_in_critical_lattice_models,
BJV:Some_remarks_on_SLE_bubbles_and_Schramms_2point_observable,
Dubedat:SLE_and_Virasoro_representations_fusion}.

To construct our solutions with the particular boundary conditions, we apply 
the quantum group method developed in~\cite{Kytola-Peltola:Conformally_covariant_boundary_correlation_functions_with_quantum_group}.
We consider spaces of highest weight vectors in tensor product representations 
of the quantum group $\Uqsltwo$ in the generic, semisimple 
case\footnote{By the generic case we mean that the deformation parameter $q$ is not a root of unity. 
The parameter $\kappa$ and the deformation parameter of the quantum group $\Uqsltwo$ are related by
 $q = e^{\ii \pi 4 / \kappa}$.}.
We construct particular bases for these vector spaces, specified by projections to subrepresentations.
Then, via the ``spin chain~--~Coulomb gas correspondence'' 
of~\cite{Kytola-Peltola:Conformally_covariant_boundary_correlation_functions_with_quantum_group},
we obtain the sought solutions to the Benoit~\& Saint-Aubin PDE systems.

Our results provide a generalization of the pure partition functions of multiple 
$\SLE$s~\cite{Kytola-Peltola:Pure_partition_functions_of_multiple_SLEs, 
PW:Global_multiple_SLEs_and_pure_partition_functions}.
They are solutions to a system of second order Benoit~\& Saint-Aubin  PDEs, 
and their recursive boundary conditions are related to
multiple $\SLE$ processes having deterministic connectivities of 
the random curves.
For statistical physics models, the pure partition functions give formulas for crossing probabilities, 
see~\cite{Cardy:Critical_percolation_in_finite_geometries,
Bauer-Bernard:Conformal_field_theories_of_SLEs,
Kontsevich:CFT_SLE_and_phase_boundaries,
BBK:Multiple_SLEs_and_statistical_mechanics_martingales, 
FSK-Multiple_SLE_connectivity_weights_for_rectangles_hexagons_and_octagons,
Izyurov:Smirnovs_observable_for_free_boundary_conditions_interfaces_and_crossing_probabilities, 
KKP:Conformal_blocks_pure_partition_functions_and_KW_binary_relation,
PW:Global_multiple_SLEs_and_pure_partition_functions}. 
Analogously, our solutions can be thought of as partition functions for systems of random curves, 
where packets of curves grow from boundary points of a simply connected domain.
In statistical physics, this corresponds to boundary arm events.
Thus, probabilities of such events should be given by our more general partition functions.
In CFT point of view, this kind of 
events should arise from insertions of boundary changing operators 
with Kac conformal weights of type $h_{1,s}$ at the starting points of the curves. 
In this sense, our generalization of the pure partition functions of multiple $\SLE$s is reminiscent of fusion in CFT.

J.~Dub\'edat studied related questions in his articles~\cite{Dubedat:SLE_and_Virasoro_representations_localization,
Dubedat:SLE_and_Virasoro_representations_fusion}, with emphasis on a priori regularity
of the partition functions, as well as on the construction of a very general framework for the relationship 
of random $\SLE$ type curves and representations of the Virasoro algebra. 
His work is based on the approach of Virasoro uniformization
initiated by M.~Kontsevich~\cite{Kontsevich:Virasoro_and_Teichmuller_spaces,
Kontsevich:CFT_SLE_and_phase_boundaries,
Friedrich-Kalkkinen:On_CFT_and_SLE,
Kontsevich-Suhov:On_Malliavin_measures_SLE_and_CFT},
and hypoellipticity~\cite{HormanderHypoelliptic,
Bony:Maximum_principle_Harnack_inequality_and_uniqueness_of_Cauchy_problem} 
and stochastic flow arguments.

\subsection{Description of the main results}

We now give an overview of the main results of the present article, in a slightly informal manner. The detailed formulations are given later, as referred to below.

\subsubsection{\textbf{Solutions with particular asymptotics}}

The main result of this article is the construction of translation invariant, 
homogeneous solutions 
$\BasisF \colon \set{ (x_1,\ldots,x_p) \in \bR^p \;|\; x_1 < \cdots < x_p} \to \bC$
to the Benoit~\& Saint-Aubin  PDE systems, determined by 
boundary conditions which concern the asymptotic behavior of the functions. 
The PDE system contains $p$ linear, homogeneous PDEs of orders 
$d_1,d_2,\ldots,d_p$,
\begin{align}\label{eq: BSA differential equations intro}
\tag{PDE}
\sD^{(j)}_{d_j} \BasisF(x_1,\ldots,x_p) = 0 , \qquad \text{ for all } j \in \{1,2, \ldots,p \},
\end{align}
and the partial differential operators $\sD^{(j)}_{d_j}$ of interest are given 
in Equation~\eqref{eq: BSA differential operator} in
Section~\ref{sec: basis functions}.

The particular asymptotics properties of the solutions 
are recursive in the number of variables. The collection $(\BasisF_\linkpatt)$ 
of solutions satisfying these properties is indexed by planar link patterns 
\begin{align*} 
\linkpatt = \Big\{\linkInEquation{a_1}{b_1},\ldots,\linkInEquation{a_\ell}{b_\ell}\Big\} \bigcup
\Big\{\defectInEquation{c_1},\ldots,\defectInEquation{c_{s}}\Big\}, 
\end{align*} 
which are defined as collections of $\ell$ links $\link{a\;}{\;b\,}$
and $s$ defects $\defect{c}$ in the upper half-plane, with endpoints
$a_1,\ldots,a_\ell,b_1,\ldots,b_\ell,c_1,\ldots,c_s$
on the real axis --- see Section~\ref{subsec: connectivities} for 
details\footnote{The parameters $d_j$ are the degrees of the partial differential 
operators in~\eqref{eq: BSA differential equations intro}. They are related to
the link patterns $\linkpatt$ in such a way that the total number of lines 
in $\linkpatt$ attached to each index $j$ equals $s_j = d_j-1$.}. 
The set of links in $\linkpatt$ is a multiset, and
for a link $\link{a\;}{\;b\,}$, we denote by 
$\ell_{a,b} = \ell_{a,b}(\linkpatt)$ its multiplicity in $\linkpatt$.

The homogeneity degree of the solution $\BasisF_\linkpatt$ is related to the number $s$ of defects in $\linkpatt$, as explained 
in Section~\ref{sec: basis functions}.
The asymptotic boundary conditions for $\BasisF_\linkpatt$
are given in terms of removing links from the link pattern $\linkpatt$. 
Removal of $m$ links
$\link{a\;}{\;b\,}$ from $\linkpatt$ results in a planar link pattern with $(\ell - m)$ links, denoted by 
$\hat{\linkpatt} = \linkpatt\removeLink (m\times\link{a\;}{\;b\,})$,
as illustrated in Figure~\ref{fig: removing links} and explained in Section~\ref{subsub: link removals}.

\begin{thm*}[Theorem~\ref{thm: asymptotic properties of general basis vectors}]
There exists a collection $(\BasisF_\linkpatt)$ 
of translation invariant, homogeneous solutions
to the Benoit $\&$ Saint-Aubin PDE system~\eqref{eq: BSA differential equations intro}
such that each function $\BasisF_\linkpatt$ has the asymptotic behavior
\begin{align*}
\BasisF_\linkpatt(x_1,\ldots,x_{p}) 
\; \sim \; 
\; & C_j \times (x_{j+1}-x_j)^{\Delta_j} \times
\BasisF_{\hat{\linkpatt}}(x_1,\ldots,x_{j-1},\xi,x_{j+2}\ldots,x_{p}),
\end{align*}
as $x_j,x_{j+1}\to\xi$, 
for any $j \in \{1,2,\ldots,p-1 \}$ and $\xi \in (x_{j-1},x_{j+2})$,
where $\hat{\linkpatt} = \linkpatt\removeLink (\ell_{j,j+1}\times\link{j}{j+1})$
denotes the link pattern obtained from $\linkpatt$ by removing all the links $\link{j}{j+1}$,
and the constant 
$C_j = C(\ell_{j,j+1};d_j,d_{j+1})$ 
and exponent 
$\Delta_j = \Delta(\ell_{j,j+1};d_j,d_{j+1})$ 
are explicitly given in Section~\ref{sec: basis functions}.
\end{thm*}

We prove in~\cite{Flores-Peltola:Solution_space_of_BSA_PDEs} that
the solutions $(\BasisF_\linkpatt)$ are in fact linearly independent,
and hence, indeed form a basis of a solution space for the Benoit $\&$ Saint-Aubin PDE system~\eqref{eq: BSA differential equations intro}.
A special case is already established in Proposition~\ref{prop: FK dual elements} of the present article.

Examples of solutions with asymptotics as above were considered 
in~\cite{JJK:SLE_boundary_visits, Kytola-Peltola:Pure_partition_functions_of_multiple_SLEs, 
KKP:Conformal_blocks_pure_partition_functions_and_KW_binary_relation,
PW:Global_multiple_SLEs_and_pure_partition_functions} 
with applications to $\SLE$s: the \emph{multiple $\SLE$ pure partition functions} 
$\PartF_\alpha (x_1,\ldots,x_{2N}) \propto \BasisF_\alpha(x_1,\ldots,x_{2N})$,
where $\alpha$ is a planar pair partition 
(thought of as a link pattern with $\ell = N$ links and $s = 0$ defects), 
and the \emph{chordal $\SLE$ boundary visit probability amplitudes} 
$\Ampl_\linkpatt (x;y_1,y_2,\ldots,y_n) = \BasisF_\linkpatt (x;y_1,y_2,\ldots,y_n)$,
where the link pattern $\linkpatt$ encodes the order of visits of the $\SLE$ curve started from $x$ to 
the boundary points $y_1,y_2,\ldots,y_n$.
We discuss the pure partition functions $\PartF_\alpha$ briefly in Section~\ref{sec: multiple SLEs},
but refer to the literature for details about the boundary visit amplitudes $\Ampl_\linkpatt$; 
see~\cite{JJK:SLE_boundary_visits, Kytola-Peltola:Pure_partition_functions_of_multiple_SLEs}.

\begin{figure}
\includegraphics[scale=.75]{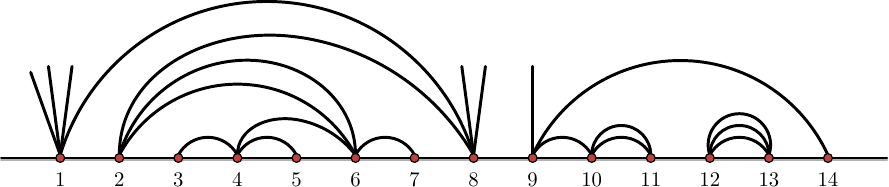}
\caption{\label{fig: example connectivity}
Example of a planar $(36,15)$-link pattern of $p=14$ points.}
\end{figure}

In Section~\ref{sec: basis functions}, we also prove a further property
of the functions $\BasisF_\linkpatt$, concerning limits when taking several 
variables together simultaneously. 
In terms of the link pattern $\linkpatt$, this means removing several links simultaneously. 
Such asymptotics pertains to general boundary behavior of the solutions.

\begin{prop*}[Proposition~\ref{prop: strong limits}]
For any $1 \leq j < k \leq p$ and $\xi \in (x_{j-1},x_{k+1})$, we have
\begin{align*}
\lim_{x_j, x_{j+1}, \ldots, x_k \to \xi}
\frac{\BasisF_\linkpatt(x_1,\ldots,x_p)}{\BasisF_\tau(x_j,\ldots,x_k)}
= \; & \BasisF_{\linkpatt \removeLink \tau}(x_1,\ldots,x_{j-1},\xi,x_{k+1},\ldots,x_p) ,
\end{align*}
where $\tau$ denotes the sub-link pattern of $\linkpatt$ between the indices $j,k$
and $\linkpatt \removeLink \tau$ denotes the link pattern obtained from $\linkpatt$ by removing $\tau$, 
as detailed in Section~\ref{subsec: further limit properties}.
\end{prop*}

\subsubsection{\textbf{Cyclic permutation symmetry}}

Solutions of the Benoit~\& Saint-Aubin  PDEs enjoying 
M\"obius covariance play a special role in conformal field theory. In particular, physical correlation 
functions should transform covariantly under all M\"obius maps, 
by conformal invariance of the theory. In applications to the theory of 
$\SLE$s, observables such as the multiple $\SLE$ (pure) partition functions 
also have this property.
More generally, in Theorem~\ref{thm: asymptotic properties of general basis vectors} 
we show also that the solutions $\BasisF_\linkpatt$ corresponding to link 
patterns $\linkpatt$ with zero defects are M\"obius covariant.
These functions also behave nicely in the limits 
$x_1 \to -\infty$ and $x_p \to +\infty$, in the following sense.

\begin{prop*}[Proposition~\ref{prop: cyclicity of limits}]
For any link pattern $\linkpatt$ with no defects $(s = 0)$, we have
\begin{align*}
\lim_{y \to + \infty} \Big( y^{2h} \times 
\BasisF_{\linkpatt}(x_1,\ldots,x_{p-1},y) \Big)
= \lim_{y \to - \infty} \Big( |y|^{2h} \times 
\BasisF_{\Scomb(\linkpatt)}(y,x_2,\ldots,x_{p}) \Big),
\end{align*} 
where $h = h_{1,d_{p}}$ is a Kac weight associated to the point $x_p$,
and $\Scomb(\linkpatt)$ is a planar link pattern obtained from $\linkpatt$ by 
a cyclic permutation, as defined in Equation~\eqref{eq: combinatorial Smap}
in Section~\ref{sec: cyclic permutation symmetry}.
\end{prop*}

\subsubsection{\textbf{Application to multiple $\SLE$ pure partition functions}}

\begin{figure}
\includegraphics[scale=.75]{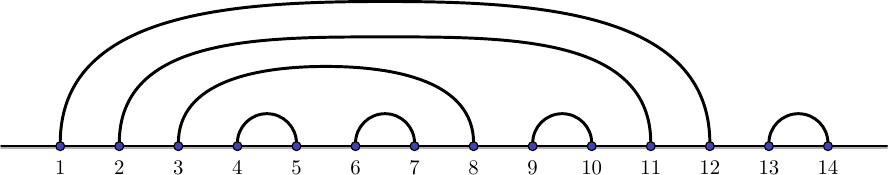}
\caption{\label{fig: example planar pair partition}
Example of a planar pair partition of $2N = 14$ points.}
\end{figure}

As a special case of the cyclic permutation symmetry, it follows that the 
multiple $\SLE$ pure partition functions $\PartF_\alpha(x_1,\ldots,x_{2N})$
satisfy a cascade property when $x_1 \to -\infty$ and $x_{2N} \to +\infty$.
In terms of the probability measures of the random curves, 
we have a natural cascade property concerning 
the removal of one curve, see~\cite{Kozdron-Lawler:Configurational_measure_on_mutually_avoiding_SLEs,
Kytola-Peltola:Pure_partition_functions_of_multiple_SLEs, 
PW:Global_multiple_SLEs_and_pure_partition_functions}.

\begin{cor*}[Corollary~\ref{cor: removing link around infinity}]
For any planar pair partition $\alpha$, we have
\begin{align*}
\lim_{\substack{x_{1}\to -\infty, \\ x_{2N}\to +\infty}}
|x_{2N}-x_{1}|^{2h_{1,2}} \times \PartF_{\alpha}(x_{1},\ldots,x_{2N})
=\; & \begin{cases}
0\quad & \text{if }\linkInEquation{1}{2N}\notin\alpha\\
\PartF_{\hat{\alpha}}(x_{2},\ldots,x_{2N-1}) 
& \text{if } \linkInEquation{1}{2N}\in\alpha,
\end{cases}
\end{align*}
where $\hat{\alpha} = \alpha \removeLink \link{1}{2N}$, and
$h_{1,2} = \frac{6-\kappa}{2\kappa}$,
and $\kappa \in (0,8) \setminus \bQ$ is the parameter of the $\SLEk$.
\end{cor*}

\subsection{Organization}

Our most important results are given in Section~\ref{sec: basis functions} 
(Theorem~\ref{thm: asymptotic properties of general basis vectors} \& Proposition~\ref{prop: strong limits}):
construction of the solutions $\BasisF_\linkpatt$ to the Benoit~\& Saint-Aubin  PDEs, 
Section~\ref{sec: multiple SLEs} (Corollary~\ref{cor: cyclicity of limits for pure geometries}):
application to the multiple $\SLE$ pure partition functions,
and Section~\ref{sec: basis vectors} (Theorem~\ref{thm: highest weight vector space basis vectors}): construction of
certain basis vectors $\Puregeom_\linkpatt$ in tensor product representations of the quantum group $\Uqsltwo$,
which serve as building blocks for the solutions $\BasisF_\linkpatt$ with the ``spin chain~--~Coulomb gas correspondence''
of~\cite{Kytola-Peltola:Conformally_covariant_boundary_correlation_functions_with_quantum_group}.

Sections~\ref{sec: preliminaries}~--~\ref{sec: cyclic permutation symmetry}
concern the representation theory of the quantum group $\Uqsltwo$ and the construction and properties of the vectors $\Puregeom_\linkpatt$. 
Sections~\ref{sec: basis functions}~--~\ref{sec: multiple SLEs} treat 
the solutions $\BasisF_\linkpatt$ themselves.
In Section~\ref{sec: basis functions}, we also very briefly discuss 
the quantum group method of the article~\cite{Kytola-Peltola:Conformally_covariant_boundary_correlation_functions_with_quantum_group}. 
Appendices~\ref{app: q-combinatorics} and \ref{app: auxiliary calculations}
contain auxiliary calculations which are needed in the proofs in 
Section~\ref{sec: basis vectors}. Appendix~\ref{app: dual elements}
constitutes some additional tools needed to prove the general limiting behavior 
of the basis functions $\BasisF_\linkpatt$.

\subsection{Acknowledgments}

During this work, the author was supported by Vilho, Yrj\"o and Kalle V\"ais\"al\"a Foundation and affiliated with the University of Helsinki.
She wishes to especially thank Steven Flores and Kalle Kyt\"ol\"a for many inspiring discussions and ideas. 
She has also enjoyed stimulating and helpful discussions 
with Michel Bauer, Dmitry Chelkak, Julien Dub\'edat, Bertrand Duplantier, 
Philippe Di Francesco, Cl\'ement Hongler, Konstantin Izyurov, Jesper Jacobsen, 
Fredrik Johansson-Viklund, Rinat Kedem, Antti Kemppainen, Jonatan Lenells, 
Jason Miller, Wei Qian, Hubert Saleur, and Hao Wu.
She thanks Roland Friedrich for pointing out important references.

\bigskip{}
\section{\label{sec: preliminaries}Preliminaries: The quantum group $\Uqsltwo$ and some combinatorics}

In this section, we discuss preliminaries concerning the quantum group 
$\Uqsltwo$ and its representations, as well as the set of planar link patterns 
$\linkpatt$ and some combinatorial results. We also introduce notations which 
will be used and referred to throughout this article.

Fix a parameter $q\in\bC\setminus\set{0}$, and assume that 
$q^m\neq1$, for all $m\in\bZ\setminus\set{0}$, i.e., that $q$
is not a root of unity.
Let $m\in\bZ$, and $n,k\in\bN$, with $0\leq k\leq n$. 
We define the $q$-integers as
\begin{align*}
\qnum{m} = \; & \frac{q^{m}-q^{-m}}{q-q^{-1}} 
    =  q^{m-1}+q^{m-3}+\cdots+q^{3-m}+q^{1-m} 
\end{align*}
and the $q$-factorials and $q$-binomial coefficients as
\begin{align*}
\qfact n = \; & \prod_{m=1}^{n}\qnum m \qquad \qquad \text{and} \qquad \qquad
\qbin nk = \frac{\qfact n}{\qfact k \qfact{n-k}}. 
\end{align*}

\subsection{\label{subsec: quantum group}The quantum group}

We define the quantum group $\Uqsltwo$ as the associative unital algebra over the field $\bC$ of complex numbers, 
with generators $K,K^{-1},E,F$ and relations
\begin{align*}
 KK^{-1} = &\; 1 = K^{-1}K,\qquad KE = q^{2}EK,\qquad KF = q^{-2}FK,\\
 EF-FE = &\; \frac{1}{q-q^{-1}}\left(K-K^{-1}\right).\nonumber 
\end{align*}
The algebra homomorphism
\begin{align*}
\Delta \;\colon\; & \Uqsltwo\rightarrow\Uqsltwo\tens\Uqsltwo,
\end{align*}
defined by its values on the generators,
\begin{align}\label{eq: coproduct}
\Hcp(E) = E\tens K + 1\tens E, \qquad \Hcp(K) = K\tens K,\qquad 
\Hcp(F) = F\tens1 + K^{-1}\tens F,
\end{align}
gives a coproduct on $\Uqsltwo$, and it determines the unique Hopf algebra 
structure on the quantum group. Furthermore, using the the coproduct $\Delta$, 
tensor products of representations of $\Uqsltwo$ can be equipped with 
a representation structure as follows. If $M'$ and $M''$ are two
representations, and we have
\begin{align*}
\Hcp(X) = \; & \sum_{i}X_{i}'\tens X_{i}''\,\in\,\Uqsltwo\tens\Uqsltwo,
\end{align*}
we define the action of $X \in \Uqsltwo$ on the tensor product $M' \tens M''$ 
by linear extension of the formula
\begin{align*}
X.(v'\tens v'') = \; & \sum_{i}(X_{i}'.v')\tens(X_{i}''.v'')\,\in\, M'\tens M'' ,
\end{align*}
for any $v'\in M'$, $v''\in M''$.
We similarly define tensor product representations with 
$n$ tensor components using the $(n-1)$-fold coproduct
\begin{align*}
\Hcp^{(n)} \colon\Uqsltwo\rightarrow\Big(\Uqsltwo\Big)^{\tens n} , \qquad
\Hcp^{(n)} =\; & (\Hcp\tens\id^{\tens(n-2)})\circ(\Hcp\tens\id^{\tens(n-3)})\circ\cdots
\circ(\Hcp\tens\id)\circ\Hcp ,
\end{align*}
and by the coassociativity property 
$(\id\tens\Hcp)\circ\Hcp=(\Hcp\tens\id)\circ\Hcp$ of the coproduct,
there is no need to specify the order in which the tensor products are formed.
The multiple coproducts of the generators have the following formulas
(see e.g. \cite[Lemma~2.2]{Kytola-Peltola:Conformally_covariant_boundary_correlation_functions_with_quantum_group}):
\begin{align} \label{eq: multiple coproducts}
\begin{split}
\Hcp^{(n)}(K)=\; & K^{\tens n} , \\
\Hcp^{(n)}(E)=\; & \sum_{i=1}^{n}1^{\tens(i-1)}\tens E\tens K^{\tens(n-i)} , \\
\Hcp^{(n)}(F)=\; & \sum_{i=1}^{n}(K^{-1})^{\tens(i-1)}\tens F\tens1^{\tens(n-i)}.
\end{split}
\end{align}

\subsection{\label{subsec: representations}Representations of the quantum group}

The quantum group $\Uqsltwo$ has irreducible representations of any dimension $d \in \bZpos$.
Given $d$, we always denote $s = d-1$.
A $d$-dimensional representation $\Wd_d$ of highest weight $q^s$ 
is obtained by suitably $q$-deforming the irreducible representation
of dimension $d$ and highest weight $s$ of the semisimple Lie algebra $\sl_{2}(\bC)$:
$\Wd_d$ has a basis 
$\Wbas_{0}^{(d)},\Wbas_{1}^{(d)},\ldots,\Wbas_{s}^{(d)}$
with action
\begin{align*}
K.\Wbas_{l}^{(d)}=\; & q^{s-2l}\,\Wbas_{l}^{(d)} , \\
F.\Wbas_{l}^{(d)}=\; & \begin{cases}
\Wbas_{l+1}^{(d)} & \text{if }l\neq s \\
0 & \text{if }l=s ,
\end{cases} \\
E.\Wbas_{l}^{(d)}=\; & \begin{cases}
\qnum l\qnum{s-l+1}\,\Wbas_{l-1}^{(d)} & \text{if }l\neq0\\
0 & \text{if }l=0
\end{cases}
\end{align*}
of the generators $E,F,K$.
For simplicity, we usually drop the superscript notation from the basis vectors, writing $\Wbas_{l}^{(d)}=\Wbas_{l}$.
It is well-known that $\Wd_d$ are irreducible,
see, e.g.,~\cite[Lemma~2.3]{Kytola-Peltola:Conformally_covariant_boundary_correlation_functions_with_quantum_group}.

When $q$ is generic (not a root of unity), the representation theory of $\Uqsltwo$ is semisimple, 
and in particular, tensor products of representations decompose
into direct sums of irreducible subrepresentations. 
We will make use of the following quantum Clebsch-Gordan decomposition.

\begin{lem}[{see, e.g., \cite[Lemma~2.4]{Kytola-Peltola:Conformally_covariant_boundary_correlation_functions_with_quantum_group}}]
\label{lem: tensor product representations of quantum sl2}
Let $d_1, d_2 \in \bZpos$ and $m\in \{0,1,\ldots,\min(s_{1},s_{2}) \}$,
where we denote $d=d_{1}+d_{2}-1-2m$ and 
$s_1 = d_1-1$, $s_2 = d_2-1$. 
In the representation $\Wd_{d_{2}}\tens\Wd_{d_{1}}$, the vector
\begin{align} 
\Tbas_{0}^{(d;d_{1},d_{2})} = \; & \sum_{l_{1},l_{2} = 0}^{\min (s_1, s_2)}
\delta_{l_{1}+l_{2},m}\times(-1)^{l_{1}}\frac{\qfact{s_{1}-l_{1}}\,\qfact{s_{2}-l_{2}}}{\qfact{l_{1}}\qfact{s_{1}}\qfact{l_{2}}\qfact{s_{2}}}\,\frac{q^{l_{1}(s_{1}-l_{1}+1)}}{(q-q^{-1})^{m}}
\times(\Wbas_{l_{2}}\tens\Wbas_{l_{1}})
\label{eq: tensor product hwv}, 
\end{align}
is a highest weight vector of a subrepresentation isomorphic to $\Wd_{d}$,
that is, we have
\begin{align*}
E.\Tbas_{0}^{(d;d_{1},d_{2})}=0\qquad\text{and}\qquad K.\Tbas_{0}^{(d;d_{1},d_{2})} = q^{s} \, \Tbas_{0}^{(d;d_{1},d_{2})}.
\end{align*}
The space $\Wd_{d_{2}}\tens\Wd_{d_{1}}$ has the following decomposition
according to the $d$-dimensional subrepresentations:
\begin{align}\label{eq: decomposition of tensor product}
\Wd_{d_{2}}\tens\Wd_{d_{1}}\isom\; & \Wd_{d_{1}+d_{2}-1}\oplus\Wd_{d_{1}+d_{2}-3}\oplus\cdots\oplus\Wd_{|d_{1}-d_{2}|+3}\oplus\Wd_{|d_{1}-d_{2}|+1} .
\end{align}
\end{lem}

For each $d$, the subrepresentation
$\Wd_{d} \subset \Wd_{d_{2}}\tens\Wd_{d_{1}}$
is generated by the highest weight vector $\Tbas_{0}^{(d;d_{1},d_{2})}$,
and we denote by $\Tbas_{l}^{(d;d_{1},d_{2})} = F^l . \Tbas_{0}^{(d;d_{1},d_{2})}$ the corresponding basis, with $l \in \{0,1,\ldots,s\}$.

\subsection{\label{subsec: tensor products}Tensor products of representations}

In this article, we consider tensor products
\begin{align}\label{eq: order of tensorands}
\bigotimes_{i=1}^{p}\Wd_{d_{i}} =\; & \Wd_{d_{p}}\tens\Wd_{d_{p-1}}\tens\cdots\tens\Wd_{d_{2}}\tens\Wd_{d_{1}}
\end{align}
of irreducible representations of the quantum group $\Uqsltwo$, and we use 
the convention of
\cite{Kytola-Peltola:Pure_partition_functions_of_multiple_SLEs,
Kytola-Peltola:Conformally_covariant_boundary_correlation_functions_with_quantum_group}
for the order of tensorands, as explicitly written on the right hand
side. We occasionally abbreviate the tensor product as above,
in which case the reverse order of tensorands is implicit.
By the coassociativity property of the coproduct, repeated application of the 
decomposition~\eqref{eq: decomposition of tensor product} gives the direct sum 
formula
\begin{align}\label{eq: decomposition of general tensor product}
\Wd_{d_p}\tens\cdots\tens\Wd_{d_1} \isom \bigoplus_d m_d \, \Wd_d,
\end{align}
where the subrepresentations isomorphic to $\Wd_d$ now have multiplicities 
$m_d = m_d(d_1,\ldots,d_p)$.

Throughout this article, it is convenient to denote by
\begin{align}\label{eq: s in terms of d}
s = d - 1, \qquad
s_i = d_i - 1 , \quad \text{ for all } i \in \{1,2,\ldots,p\} , \qquad 
\text{ and } \qquad \multii = (s_1,s_2,\ldots,s_p) \in \bZnn^p.
\end{align}
For $d = s+1$, 
the $m_d$ copies of $\Wd_d$ are generated by highest weight vectors 
of weight $q^{s}$. We denote by
\begin{align}\label{eq: HW space}
\HWsp_\multii^{(s)}
= \; & \Big\{v\in\bigotimes_{i=1}^{p}\Wd_{d_{i}} \; \big| \;
    E.v=0 , \; K.v = q^{s}\,v \Big\}
\end{align}
the $m_d$-dimensional subspace consisting of such vectors.
The dimensions $m_d$ satisfy a recursion equation, 
given in Lemma~\ref{lem: cardinalities are equal},
and they can be calculated by counting certain type of planar link patterns.

\subsection{\label{subsec: projections}Projections to subrepresentations}

Fix $j \in \{1,2,\ldots,p-1 \}$. We decompose the $j$:th and $(j+1)$:st
tensor components in~\eqref{eq: order of tensorands} according to
the quantum Clebsch-Gordan formula~\eqref{eq: decomposition of tensor product},
and denote 
the embedding of the $d$-dimensional component 
in the $j$:th and $(j+1)$:st tensor positions by
\begin{align} \label{eq: embedding of a subrepresentation} 
\begin{split}
& \iota_{j}^{(d)} = \iota_{j,j+1}^{(d;d_j,d_{j+1})} \;\colon\;  \Big(\bigotimes_{i=j+2}^p\Wd_{d_{i}}\Big)\tens\Wd_{d}\tens\Big(\bigotimes_{i=1}^{j-1}\Wd_{d_{i}}\Big)
    \;\to\; \bigotimes_{i=1}^p\Wd_{d_{i}} \\
& \iota_{j}^{(d)} 
\left(\Wbas_{l_{p}}\tens\cdots\tens\Wbas_{l_{j+2}}
\tens\Wbas_{l}\tens\Wbas_{l_{j-1}}\tens\cdots\tens\Wbas_{l_{1}}\right)
= \Wbas_{l_{p}}\tens\cdots\tens\Wbas_{l_{j+2}}
\tens\Tbas_{l}^{(d;d_{j},d_{j+1})}\tens\Wbas_{l_{j-1}}\tens\cdots
\tens\Wbas_{l_{1}} . 
\end{split}
\end{align}
Via the embedding~\eqref{eq: embedding of a subrepresentation},
we identify the shorter tensor product
as a subrepresentation of~\eqref{eq: order of tensorands},
\begin{align}\label{eq: subrepresentation identification}
\Big(\bigotimes_{i=j+2}^p\Wd_{d_{i}}\Big)\tens\Wd_{d}\tens\Big(\bigotimes_{i=1}^{j-1}\Wd_{d_{i}}\Big)\,\subset\,\; & \bigotimes_{i=1}^p\Wd_{d_{i}},
\end{align}
and denote by
\begin{align}\label{eq: subrepresentation projection}
\pi_{j}^{(d)} = \pi_{j,j+1}^{(d;d_j,d_{j+1})} \;\colon\; & \bigotimes_{i=1}^p\Wd_{d_{i}}\to\bigotimes_{i=1}^p\Wd_{d_{i}}
\end{align}
the projection to this subrepresentation
--- by definition, 
a vector $v\in\bigotimes_{i=1}^p\Wd_{d_{i}}$ lies in 
the subrepresentation~\eqref{eq: subrepresentation identification}
if and only if we have $\pi_{j}^{(d)}(v) = v$.
We further let
\begin{align*}
\hat{\pi}_{j}^{(d)} = \hat{\pi}_{j,j+1}^{(d;d_j,d_{j+1})} \;\colon\; & \bigotimes_{i=1}^p\Wd_{d_{i}}\to\Big(\bigotimes_{i=j+2}^p\Wd_{d_{i}}\Big)
\tens\Wd_{d}\tens\Big(\bigotimes_{i=1}^{j-1}\Wd_{d_{i}}\Big)
\end{align*}
denote the projection~\eqref{eq: subrepresentation projection} 
combined with the 
identification~\eqref{eq: subrepresentation identification}
--- the identity
$\pi_{j}^{(d)} = \iota_{j}^{(d)} \circ \hat{\pi}_{j}^{(d)}$ then holds.

More generally, for $m \in \{0,1,\ldots,\min(d_{j},d_{j+1})-1 \}$
and $d=d_{j}+d_{j+1}-1-2m$, we define the map
\begin{align} \label{eq: generalized projection}
\begin{split}
& \widetilde{\pi}_{j}^{(d)} = 
\widetilde{\pi}_{j,j+1}^{(d;d_j,d_{j+1})} \;\colon \;
 \bigotimes_{i=1}^p\Wd_{d_{i}}    \;\to\; 
     \Big(\bigotimes_{i=j+2}^p\Wd_{d_{i}}\Big)\tens
 \Wd_{d_{j+1}-m}\tens\Wd_{d_{j}-m}
 \tens\Big(\bigotimes_{i=1}^{j-1}\Wd_{d_{i}}\Big) \\
& \widetilde{\pi}_{j,j+1}^{(d;d_j,d_{j+1})} :=
\iota_{j,j+1}^{(d;d_j-m,d_{j+1}-m)} \circ\hat{\pi}_{j,j+1}^{(d;d_j,d_{j+1})} , 
\end{split}
\end{align}
whose image is a subrepresentation of 
type~\eqref{eq: subrepresentation identification}, 
having the subrepresentation of $\Wd_{d_{j+1}-m}\tens\Wd_{d_{j}-m}$
isomorphic to $\Wd_d$ in the $j$:th and $(j+1)$:st tensor positions.

The trivial 
representation $\Wd_{1}$ 
is the neutral element for tensor products of representations.
We always identify it with the complex numbers $\bC$, 
via $\Tbas_{0}^{(1;d_j,d_{j+1})}\mapsto1$, and omit it from the tensor products.
The image of the projection $\hat{\pi}_{j}^{(1)}$
thus lies in the shorter tensor product
$\Big(\bigotimes_{i=j+2}^p\Wd_{d_{i}}\Big)\tens
\Big(\bigotimes_{i=1}^{j-1}\Wd_{d_{i}}\Big)$, and for 
$m=\min(d_{j},d_{j+1})-1$, the embedding
$\iota_{j,j+1}^{(d;d_j-m,d_{j+1}-m)}$ reduces to the identity map.

\subsection{\label{subsec: connectivities}Planar link patterns}

Tensor product representations of type~\eqref{eq: decomposition of general tensor product}
have bases indexed by planar link patterns, where each highest weight 
vector corresponds to a link pattern, and the other basis elements are obtained by action of the generator $F$. 
For example, a relatively well-known fact 
is that the tensor power $\Wd_2^{\tens n}$ of two-dimensional irreducibles has such a basis;
see e.g.~\cite{Jimbo:q_analog_of_UqglN_Hecke_algebra_and_YBE,
Lusztig:Canonical_bases_in_tensor_products, 
Martin:On_Schur-Weyl_duality_An_Hecke_algebras_and_quantum_slN_on_CN_tensor_nplus1,
Frenkel-Khovanov:Canonical_bases_in_tensor_products_and_graphical_calculus_for_Uqsl2,
PSA:Idempotents_of_the_TL_module_C2_tensor_n_in_terms_of_elements_of_Uqsltwo,
Ridout-Saint-Aubin:Standard_modules_induction_and_structure_of_TL}.
In this case, for each $s \in \bZnn$ the space~\eqref{eq: HW space}
of highest weight vectors in $\Wd_2^{\tens n}$ admits a natural
diagrammatic action of the Temperley-Lieb algebra, 
known as the link-state representation.
For $s = 0$, $n$ is even, and the link states are indexed by planar pair partitions 
of $n/2$ points, 
see Figure~\ref{fig: example planar pair partition}. 
For $s > 0$, there are also additional lines called defects,
see Figure~\ref{fig: example planar pair partition with defects}.

In the present article, we consider general link patterns, which are useful
in calculations concerning general tensor product representations of 
type~\eqref{eq: decomposition of general tensor product}. The planar pair partitions then arise as 
a special case. A word of warning is in order here: the bases of the tensor 
product representations of type~\eqref{eq: decomposition of general tensor product} which we 
construct in Section~\ref{sec: basis vectors} \emph{do not} carry 
the ``usual'' link-state action of diagram algebras such as
the Temperley-Lieb algebra, even in the special case of $\Wd_2^{\tens n}$. 
In fact, the basis we construct in the present article is the dual basis of 
the ``canonical basis''~\cite{Lusztig:Canonical_bases_in_tensor_products, 
Frenkel-Khovanov:Canonical_bases_in_tensor_products_and_graphical_calculus_for_Uqsl2}\footnote{
General link patterns do not span representations of 
the Temperley-Lieb algebra, but they admit a natural action of a generalized diagram algebra, discussed 
in~\cite{Flores-Peltola:WJTL_algebra, Flores-Peltola:Colored_braid_representations_and_QSW}.
}.
However, we will not pursue this direction here --- our interests lie in 
constructing solutions to the Benoit~\& Saint-Aubin  PDE systems with given 
asymptotics properties, using the quantum group rather as a tool.

Denote the upper half-plane by $\bH=\set{z\in\bC\;|\;\im(z)>0}$.
Fix a multiindex $\multii = (s_1,\ldots,s_p) \in \bZpos^p$, 
and let $\ell \in \bZnn$ be an integer such that $2\ell \leq n$, 
where we denote 
\begin{align}\label{eq: definition of n and s}
n = |\multii| := \sum_{i=1}^p s_i \qquad \text{ and } \qquad s = n - 2\ell \,\in\, \bZnn.
\end{align}
We define planar ($n,\ell$)-link patterns of $p$ points
with index valences $\multii = (s_1,\ldots,s_p)$ as collections
\begin{align*} 
\linkpatt = \Big\{\linkInEquation{a_1}{b_1},\ldots,\linkInEquation{a_\ell}{b_\ell}\Big\}
\bigcup
\Big\{\defectInEquation{c_1},\ldots,\defectInEquation{c_{s}}\Big\}
\end{align*} 
of 
\begin{itemize}
\item $\ell$ links of type $\link{a\;}{\;b\,}$ in $\bH$, which connect
a pair $a < b$ of indices $a,b \in \{1,2,\ldots,p \}$, and
\item $s = n - 2\ell$ defects of type $\defect{c}$ in $\bH$, attached 
to an index $c \in \{1,2,\ldots,p \}$,
\end{itemize}
such that
\begin{itemize}
\item for any $i \in \{1,2, \ldots,p\}$, the index $i$ is an endpoint of 
exactly $s_i$ links and defects, 
\item all the defects lie in the unbounded component of the complement 
of the set of links in $\bH$, and
\item none of the links and defects intersect in $\bH$, 
but only at their common endpoints in $\bN \subset \bR$.
\end{itemize}
Figure~\ref{fig: example connectivity} shows an example of 
a planar link pattern. We denote by $\LP^{(s)}_\multii$ 
the set of planar ($n,\ell$)-link patterns of $p$ points with index valences
$\multii = (s_1,\ldots,s_p)$, having $s = n - 2\ell$ defects.
We usually omit the word ``planar'' when we speak of link patterns.

Because the planar pair partitions play a special role in this article,
we denote the set of them by
\begin{align*}
\PP_N := \LP^{(0)}_{(1,1,\ldots,1,1)} , \quad \text{ for } N\in\bZpos, \qquad
\text{ and }\qquad
\PP_0 := \LP^{(0)}_{()}=\set{\emptyset} , \quad \text{ for }N=0.
\end{align*}
We also set $\PP = \bigsqcup_{N\geq0}\PP_N$.

More generally, if $\multii = (1,1,\ldots,1,1) \in \bZ^{n}$ 
for $n = 2N + s$, we denote by
$\PP_N^{(s)} := \LP^{(s)}_{(1,1,\ldots,1,1)}$
the set of planar $(n,N)$-link patterns each of which consists of a planar pair partition 
of $2N$ points and $s$ defects --- see 
Figure~\ref{fig: example planar pair partition with defects} for an example. 
The set of planar pair partitions 
then corresponds to $\PP_N = \PP_N^{(0)}$.

\begin{figure}
\includegraphics[scale=.75]{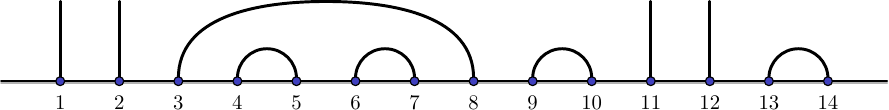}
\caption{\label{fig: example planar pair partition with defects}
Example of a planar pair partition with defects.}
\end{figure}

Next, we consider the tensor product representation~\eqref{eq: order of tensorands},
with dimensions $d_1,\ldots,d_p \geq 2$ related to the multiindex
$\multii$ as in \eqref{eq: s in terms of d}.
The dimension of the subspace $\HWsp_\multii^{(s)}$
of highest weight vectors of weight $q^{s}$ can be calculated 
by counting the planar link patterns in $\LP_\multii^{(s)}$.

\begin{lem}\label{lem: cardinalities are equal}
For each $s \in \bZnn$, we have 
$\dmn \HWsp_\multii^{(s)} = \# \LP_\multii^{(s)}$.
\end{lem}
\begin{proof}
Fix $s \in \bZnn$. 
The claim follows from the fact that both sides of the asserted equation,
\begin{align*}
\multiplicity_\multii^{(s)} := 
\dmn \HWsp_\multii^{(s)} \qquad\text{and}\qquad
N_\multii^{(s)} := \# \LP_\multii^{(s)} ,
\end{align*}
satisfy the same recursion with the same initial condition.
If $p=1$, then obviously
$\multiplicity_{(s_1)}^{(s)} = \delta_{s,s_1} =
N_{(s_1)}^{(s)}$.
For general $\multii = (s_1,\ldots,s_p) \in \bZpos^p$,
consider first the dimension $\multiplicity_\multii^{(s)}$
of $\HWsp_\multii^{(s)}$. 
Using the notations~\eqref{eq: s in terms of d}, the direct sum 
decomposition~\eqref{eq: decomposition of general tensor product}
of $p$ irreducibles, with $m_d = \multiplicity_\multii^{(s)}$,
can be written recursively as
\begin{align}\label{eq: recursion decomposition}
\bigoplus_d \multiplicity_\multii^{(s)} \; \Wd_d = 
\Wd_{d_p}\tens (\Wd_{d_{p-1}}\tens\cdots\tens\Wd_{d_1}) = 
\Wd_{d_p}\tens \Big( \bigoplus_{\hat{d}} \multiplicity_{\hat{\multii}}^{(\hat{s})} \; \Wd_{\hat{d}} \Big) =
\bigoplus_{\hat{d}} \multiplicity_{\hat{\multii}}^{(\hat{s})}
\left( \Wd_{d_p}\tens\Wd_{\hat{d}} \right),
\end{align}
where $\hat{s}  = \hat{d} - 1$ and $\hat{\multii} = (s_1,\ldots,s_{p-1})$,
by the coassociativity property of the coproduct of $\Uqsltwo$.
Using the explicit decomposition~\eqref{eq: decomposition of tensor product}
of the tensor product of two irreducibles, we obtain the recursion
\begin{align*}
\multiplicity_\multii^{(s)} = 
\sum_{k \geq 0}
\multiplicity_{\hat{\multii}}^{(s-s_p+2k)},
\end{align*}
where the numbers $\multiplicity_{\hat{\multii}}^{(s-s_p+2k)}$
are zero when $k$ is large enough (and for small $k$ in some cases).

Consider then the number $N_\multii^{(s)}$ of link patterns with
$s$ defects. We classify the link patterns 
$\linkpatt \in \LP_\multii^{(s)}$ according to the number $k$ of 
links having the endpoint~$p$ (so there are $s_p-k$ defects having 
the endpoint~$p$). Imagine cutting the point $p$ off from the link pattern
$\linkpatt$. Then, the remaining points $1,2,\ldots,p-1$ will have 
$\hat{s} = s-(s_p-k)+k = s-s_p+2k$ defects in total
--- see Figure~\ref{fig: recursion} --- namely, the $s-(s_p-k)$ defects 
inherited from $\linkpatt$ and in addition $k$ defects attached to 
the $k$ links which had the endpoint $p$. 
This gives the same recursion as above:
\begin{align*}
N_\multii^{(s)} = 
\sum_{k \geq 0} N_{\hat{\multii}}^{(s-s_p+2k)}.
\end{align*}
It follows that $\multiplicity_\multii^{(s)} = N_\multii^{(s)}$, as claimed.
\end{proof}

\begin{figure}
\includegraphics[scale=.75]{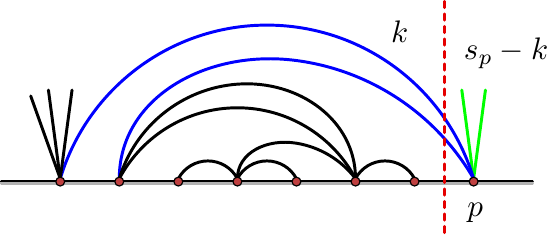}
\caption{\label{fig: recursion}
Illustration of the recursion used in the proof of 
Lemma~\ref{lem: cardinalities are equal}. When cutting the point $p$ off from 
the link pattern, the blue links become defects.}
\end{figure}

\subsection{\label{subsec: combinatorial maps}Combinatorial maps}

\begin{figure}
\includegraphics[scale=.75]{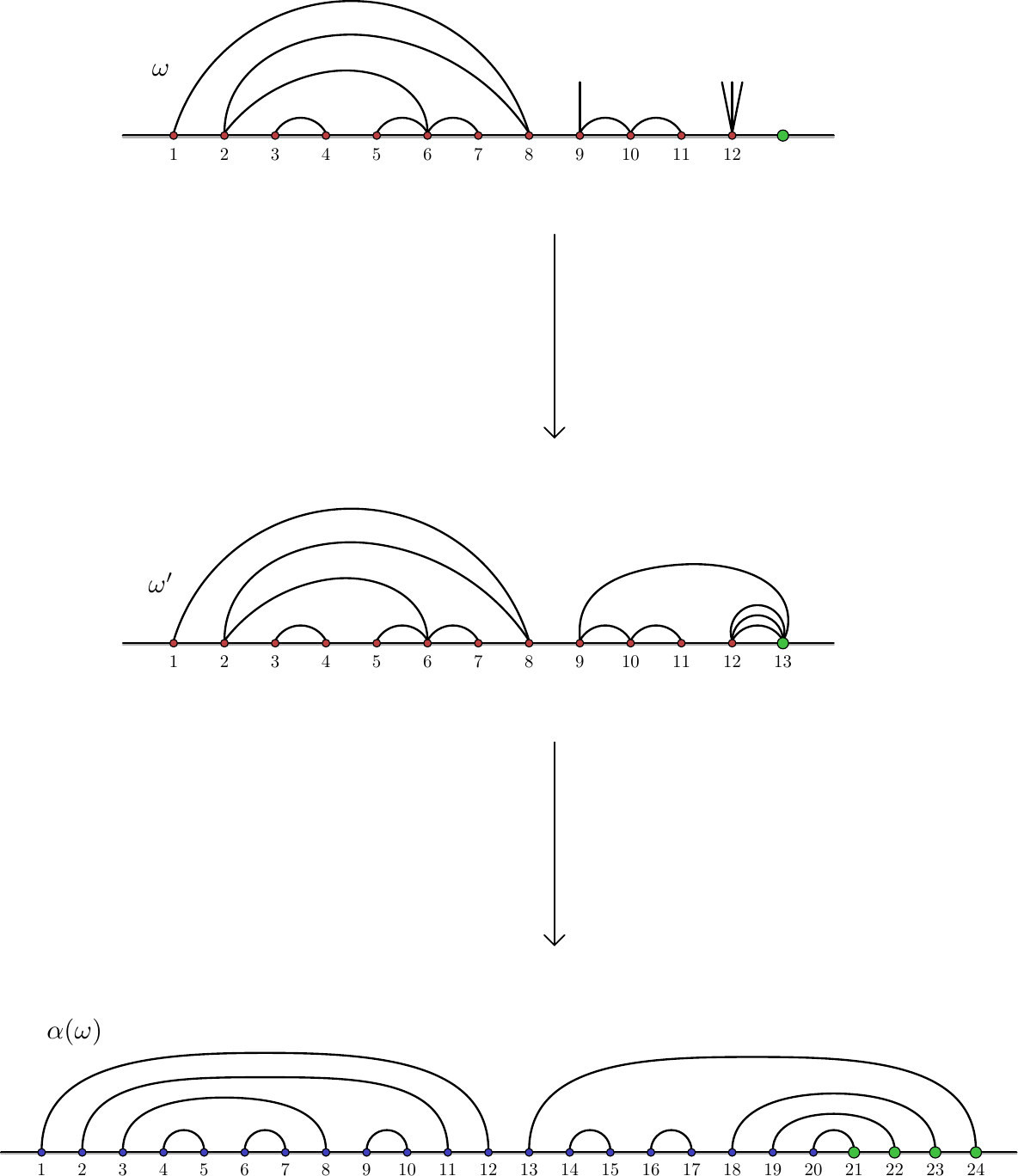}
\bigskip
\bigskip

\caption{\label{fig: open up map}
Illustration of the map $\varphi \colon \linkpatt\mapsto\alpha(\linkpatt)$, 
defined as the composition $\varphi = \Embeddingcomb \circ \Rpluscomb^{-1}$.
The middle figure illustrates the image of $\linkpatt$ under the first map,
that is, $\linkpatt' = \Rpluscomb^{-1}(\linkpatt)$, where the defects of 
$\linkpatt$ are attached to an additional index $p+1$ on the right of all the 
other indices. The lowest figure depicts the planar pair partition 
$\alpha(\linkpatt)$, which is obtained from $\linkpatt'$ by ``opening up'' 
all the points, that is, splitting each index $i$ to $s_i$ new indices and 
taking the lines of $i$ along with the points.}
\end{figure}

We next define a natural map, which associates to each planar link pattern
$\linkpatt \in \LP^{(s)}_{\multii}$
a planar pair partition $\alpha=\alpha(\linkpatt) \in \PP_N$, such that 
\begin{align}\label{eq: particular choice of N}
N=\frac{1}{2}\left(\sum_{i=1}^p s_i + s\right) = 
\frac{n+s}{2} = \ell+s.
\end{align}
This map, denoted by
\begin{align}\label{eq: link pattern to pair partition}
\varphi \; \colon \; \LP^{(s)}_{\multii} \rightarrow \PP_N,\qquad
\linkpatt \mapsto \alpha(\linkpatt),
\end{align}
is defined as a composition $\varphi = \Embeddingcomb \circ \Rpluscomb^{-1}$ 
of the two combinatorial maps
\begin{align*}
\Rpluscomb^{-1} = (\Rpluscomb^{(s)})^{-1} \; & \colon \;
\LP^{(s)}_{\multii} \rightarrow \LP^{(0)}_{(\multii,s)},
\qquad \text{ and } \qquad
\Embeddingcomb = \Embeddingcomb_{(\multii,s)} \; \colon \;
\LP^{(0)}_{(\multii,s)} \rightarrow \PP_N,
\end{align*}
which we define shortly --- see also Figure~\ref{fig: open up map} for an illustration.

We first define $\Rpluscomb^{-1}$ and its inverse map 
$\Rpluscomb \colon \LP^{(0)}_{(\multii,s)} \to \LP^{(s)}_{\multii}$. 
Consider a link pattern
\begin{align*}
\linkpatt=\Big\{\linkInEquation{a_1}{b_1},\ldots,\linkInEquation{a_\ell}{b_\ell}\Big\}
\bigcup
\Big\{\defectInEquation{c_1},\ldots,\defectInEquation{c_{s}}\Big\} \,
\in \, \LP^{(s)}_{\multii}.
\end{align*}
Introduce an additional index $p+1$, of valence $s_{p+1}=s$,
and connect the defects of $\linkpatt$ to it, to form 
\begin{align*}
\linkpatt' = \Big\{\linkInEquation{a_1}{b_1},\ldots,\linkInEquation{a_\ell}{b_\ell},\linkInEquation{c_{1}}{p+1},\linkInEquation{c_{2}}{p+1},\ldots,\linkInEquation{c_{s}}{p+1}\Big\} \,
\in \, \LP^{(0)}_{(\multii,s)},
\end{align*} 
a link pattern of $p+1$ points having index valences 
$(\multii,s) := (s_1,\ldots,s_p,s)$ and zero defects. Set
\[ \Rpluscomb^{-1}(\linkpatt) := \linkpatt' .\] 
This defines the map 
$\Rpluscomb^{-1} \colon \LP^{(s)}_{\multii} \rightarrow
\LP^{(0)}_{(\multii,s)}$. It has an obvious inverse map 
$\Rpluscomb = \Rpluscomb^{(s)} \; \colon \;
\LP^{(0)}_{(\multii,s)} \rightarrow \LP^{(s)}_{\multii}$
obtained by removing the last index $p+1$ of valence $s$ so that the links attached to it become defects. 
We similarly define 
$\Rminuscomb = \Rminuscomb^{(s)} \colon
\LP^{(0)}_{(s,\multii)} \rightarrow \LP^{(s)}_{\multii}$
and its inverse map
$\Rminuscomb^{-1} = (\Rminuscomb^{(s)})^{-1} \colon
\LP^{(s)}_{\multii} \rightarrow \LP^{(0)}_{(s,\multii)}$
by removing (resp. adding) the index $1$ and relabeling 
the other indices from left to right by $1,2,\ldots,p$ 
(resp. $2,3,\ldots,p+1$).

To define the map $\Embeddingcomb = \Embeddingcomb_{(\multii,s)}$, 
split each index $i \in \{1,2,\ldots,p+1 \}$ 
of $\linkpatt'$ to $s_i$ distinct indices, with $s_{p+1} = s$,
and attach the $s_i$ links ending at $i$ in $\linkpatt'$ to these new $s_i$ indices, 
so that each of them has valence one (see Figure~\ref{fig: open up map}).
This results in a diagram with $2N = \sum_{i=1}^p s_i+s$ 
indices, each of which has valence one. 
Label these indices from left to right by $1,2,\ldots,2N$, to obtain
the planar pair partition
\begin{align*}
\Embeddingcomb(\linkpatt') = \alpha(\linkpatt) \, \in \, 
\LP^{(0)}_{(1,1,\ldots,1,1)} = \PP_N.
\end{align*}
This finally defines the map 
$\Embeddingcomb \colon \LP^{(0)}_{(\multii,s)} \rightarrow \PP_N$ 
and thus the map $\varphi = \Embeddingcomb \circ \Rpluscomb^{-1}$.

\subsection{\label{subsec: removing links}Properties of the link patterns}

To finish, we introduce some notation concerning the recursive structure 
of the set of planar link patterns, to be used throughout this article.

\begin{figure}
\includegraphics[scale=.75]{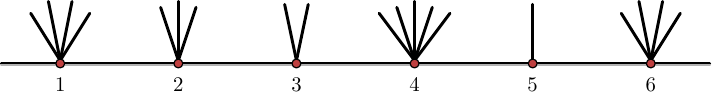}
\caption{\label{fig: defect link pattern}
Example of a link pattern {$\defpatt_\partition$} for a partition
$\partition=(4,3,2,5,1,4)$.}
\end{figure}

\subsubsection{\textbf{Defects and partitions}}
Integer partitions $\partition=(s_{1},\ldots,s_{|\partition|})$ 
of $s$ correspond naturally to endpoints of defects in 
planar link patterns. A partition $\partition$ of $s$
determines a unique $(s,0)$-link pattern denoted by
$\defpatt_\partition \in \LP^{(s)}_{\partition}$, 
which consists of $s$ defects with endpoints 
$i=1,2\ldots,|\partition|$, having valences $s_i=s_{i}$ specified by 
$\partition$, as in Figure~\ref{fig: defect link pattern}.
We also include the $(0,0)$-link pattern $\defpatt_{()}=\emptyset$ for $s=0$.

Conversely, let $\multii \in \bZpos^p$ and consider an $(n,\ell)$-link pattern with $s = n-2\ell$ defects, with notations~\eqref{eq: definition of n and s},
\begin{align*}
\linkpatt=\Big\{\linkInEquation{a_1}{b_1},\ldots,\linkInEquation{a_\ell}{b_\ell}\Big\} 
\bigcup 
\Big\{\defectInEquation{c_1},\ldots,\defectInEquation{c_{s}}\Big\}
\, \in \, \LP^{(s)}_{\multii}.
\end{align*}
When $s \geq 1$, the set 
$\big\{\defect{c_1},\ldots,\defect{c_{s}}\big\}$ of
defects in $\linkpatt$ defines naturally a partition of $s$ as follows: 
if $\set{\defendpt_1,\ldots,\defendpt_{t}} \subset \set{1,\ldots,p}$, for
$\defendpt_1<\ldots<\defendpt_{t}$,
denote the distinct endpoints of the defects in $\linkpatt$ with multiplicities 
given by the number $r_i(\linkpatt)=\#\set{k\;|\;c_k = u_i}\geq 1$ of 
defects ending at the index $\defendpt_i$, then we have
$s = \sum_{i=1}^t r_i(\linkpatt)$, and 
the numbers $r_i(\linkpatt)$ thus define a partition 
of $s$ into $t$ positive integers,
\[\partition(\linkpatt)=\big(r_1(\linkpatt),\ldots,r_t(\linkpatt)\big). \]

We denote the set of all planar link patterns 
with a fixed number $s\in\bZnn$ of defects by 
\begin{align*}
\LP^{(s)} = \bigsqcup_{p \in \bN, \; \multii \in \bZpos^p} \LP^{(s)}_{\multii},
\end{align*}
and, for a fixed partition 
$\partition=(s_{1},\ldots,s_{|\partition|})$ of $s$,
we denote by 
\[\LP^{(s)}(\partition)=\set{\linkpatt\in\LP^{(s)}\;|\;\partition(\linkpatt)=\partition} \qquad \text{ and } \qquad
\LP^{(0)}() = \LP^{(0)} \]
the set of all planar link patterns whose defects 
$\big\{\defect{c_1},\ldots,\defect{c_{s}}\big\}$
define the partition $\partition(\linkpatt)=\partition$.

\subsubsection{\textbf{Removing links}}
\label{subsub: link removals}
Also the links in the $(n,\ell)$-link pattern
\begin{align*}
\linkpatt=\Big\{\linkInEquation{a_1}{b_1},\ldots,\linkInEquation{a_\ell}{b_\ell}\Big\}
\bigcup
\Big\{\defectInEquation{c_1},\ldots,\defectInEquation{c_{s}}\Big\}
\, \in \, \LP^{(s)}_{\multii}
\end{align*} 
appear with multiplicity. For two indices $a<b$, we denote by 
$\ell_{a,b} = \ell_{a,b}(\linkpatt)$ the multiplicity of the link 
$\link{a\;}{b\;}$ in $\linkpatt$, that is, we have 
$\ell_{a,b}=\#\set{i\;|\;a_i=a,b_i=b}$ and $\ell = \sum_{k,l} \ell_{k,l}$.
In particular, the links of $\linkpatt$ can be regarded as a multiset 
of $k \leq \ell = \sum_{a,b} \ell_{a,b}$ elements,
\begin{align*}
\mathfrak{L}(\linkpatt) = 
\Big\{\ell_{a_1,b_1} \times \linkInEquation{a_1}{b_1},\; \ldots, \;\ell_{a_k,b_k} \times \linkInEquation{a_k}{b_k}\Big\} .
\end{align*}
Removing one link from an $(n,\ell)$-link pattern determines 
an $(n-2,\ell-1)$-link pattern. If the removed link had an endpoint with 
valence one, then the endpoint must be removed as well, and the remaining indices
must be relabeled so as to form the endpoints of the smaller link pattern,
as illustrated in Figure~\ref{fig: removing links}.
We denote the removal of a link $\link{a\;}{\;b\,}$ from a link pattern 
$\linkpatt$ by $\linkpatt\removeLink\link{a\;}{\;b\,}$.

More generally, if the link $\link{a\;}{\;b\,}$ appears in 
$\linkpatt \in \LP^{(s)}_{\multii}$ with 
multiplicity $\ell_{a,b}$, we can remove $m\leq \ell_{a,b}$ links
from $\linkpatt$. The removal of $m$ links $\link{a\;}{\;b\,}$ from $\linkpatt$ 
is then denoted by $\linkpatt\removeLink (m\times\link{a\;}{\;b\,})$.
In this case, if $s_j=m$ or $s_{j+1}=m$, we have to also remove the index $j$ 
or $j+1$, respectively (or both), and relabel the indices of the remaining 
links and defects, as also illustrated in Figure~\ref{fig: removing links}.

\begin{figure}
\includegraphics[scale=.75]{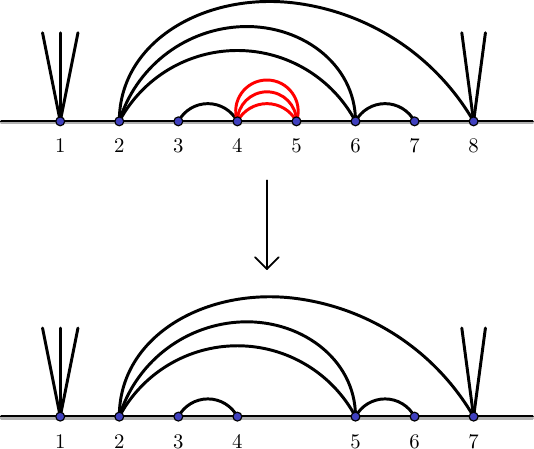} 
\qquad \qquad \qquad
\includegraphics[scale=.75]{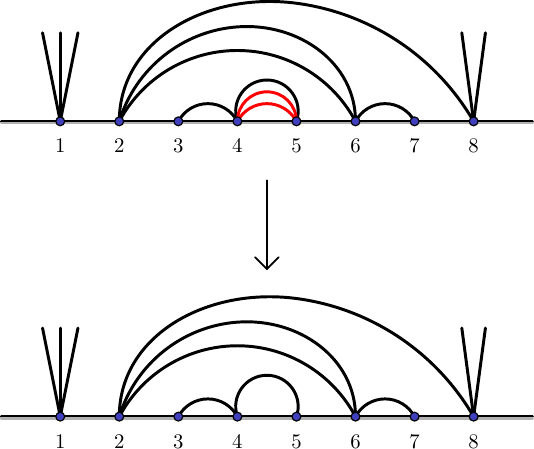}
\caption{\label{fig: removing links}
Removal of links between the indices $a = 4$ and $b = 5$. In the left figure,
all three links are removed, so the index $5$ has to be removed as well
(because it becomes empty, i.e., its valence becomes equal to zero), 
and the remaining indices are labeled accordingly.
On the right, only two links are removed, so the indices remain the same.}
\end{figure}

\bigskip{}
\section{\label{sec: basis vectors}Basis vectors in quantum group representations}

In this section, we construct a basis for each highest weight vector space
$\HWsp_\multii^{(s)}$ (defined in Equation~\eqref{eq: HW space})
whose vectors are uniquely characterized by certain recursive properties, 
concerning projections onto subrepresentations. 
These basis vectors are crucial in our construction of the basis for solutions to 
the Benoit~\& Saint-Aubin  PDEs in Section~\ref{sec: basis functions}.
The defining properties of the basis vectors correspond to the asymptotic boundary conditions 
for the basis functions, as explained in Section~\ref{sec: basis functions}.

In view of Lemma~\ref{lem: cardinalities are equal},
it is natural to index the basis vectors $\Puregeom_\linkpatt$ by link patterns $\linkpatt$.
Specifically, we consider the following system of equations for vectors
$\Puregeom_\linkpatt\in\bigotimes_{i=1}^{p}\Wd_{d_{i}}$, with $\linkpatt \in \LP^{(s)}_{\multii}$:
\begin{align}
&K.\Puregeom_\linkpatt = q^{s}\,\Puregeom_\linkpatt 
\label{eq: cartan eigenvalue}\\
&E.\Puregeom_\linkpatt=0 
\label{eq: highest weight vector}\\
&\widetilde{\pi}_{j}^{(\projdmn)}(\Puregeom_\linkpatt) =
\label{eq: projection conditions}
\begin{cases} 
\frac{1}{\constantfromdiagram{m}{s_j}{s_{j+1}}} \times \Puregeom_{\hat{\linkpatt}} & \mbox{if there are at least $m$ links } \linkInEquation{j}{j+1} \text{ in }\linkpatt\\ 0 & \mbox{otherwise,} \end{cases}\\
&\text{for all } j \in \{1, 2, \ldots,p-1 \} , \; m \in \{1,2,\ldots,\min(s_{j},s_{j+1}) \},
\text{ and }\projdmn=s_{j}+s_{j+1}+1-2m, \nonumber
\end{align}
where $\hat{\linkpatt} = \linkpatt\removeLink (m\times\link{j}{j+1})$, and
the constants in~\eqref{eq: projection conditions} are non-zero 
and explicit:
\begin{align}\label{eq: constant from diagram}
\constantfromdiagram{m}{s_j}{s_{j+1}} =
\frac{\qfact{s_j-m}\qfact{s_{j+1}-m}\qfact{s_j+s_{j+1}-m+1}}{\qnum{2}^m\qfact{s_j}\qfact{s_{j+1}}\qfact{s_j+s_{j+1}-2m+1}}
= \frac{\qbin{s_j+s_{j+1}-m+1}{m}}{\qnum{2}^m\qfact{m}\qbin{s_j}{m}\qbin{s_{j+1}}{m}} ,
\end{align}
and where we use, by default, the notations~\eqref{eq: s in terms of d}
for the parameters $s$, $s_j$, $s_{j+1}$, and $d = s + 1$.

Equations~\eqref{eq: cartan eigenvalue}~--~\eqref{eq: highest weight vector}
state that each $\Puregeom_\linkpatt$ belongs to the highest weight vector space
$\HWsp_\multii^{(s)}$. Equations~\eqref{eq: projection conditions}
concern projections of $\Puregeom_\linkpatt$ to subrepresentations,
corresponding to removing links from the link pattern~$\linkpatt$. 

\begin{thm}\label{thm: highest weight vector space basis vectors}
\
\begin{description}
\item[(a)]
For each integer $s \geq 0$, there exists a unique collection $\left(\Puregeom_{\linkpatt}\right)_{\linkpatt\in\LP^{(s)}}$ 
of vectors 
such that the system of 
equations~\eqref{eq: cartan eigenvalue}~--~\eqref{eq: projection conditions} holds for all $\linkpatt\in\LP^{(s)}$, we have
$\Puregeom_\emptyset = 1$, and
\begin{align}\label{eq: normalization of basis vectors}
\Puregeom_{\defpatt_\partition} =
\frac{1}{(q-q^{-1})^{s}}\frac{\qnum{2}^{s}}{\qfact{s+1}}
\times
(\Wbas_0^{(s_{|\partition|}+1)}\tens\cdots\tens\Wbas_0^{(s_1+1)})
\,\in\, \HWsp_\partition^{(s)} ,
\end{align} 
for any partition $\partition=(s_{1},\ldots,s_{|\partition|})$ of $s\geq 1$.
\item[(b)]
For fixed $\multii \in \bZpos^p$, the collection
$\left(\Puregeom_{\linkpatt}\right)_{\linkpatt\in\LP^{(s)}_{\multii}}$
is a basis of the vector space $\HWsp_\multii^{(s)}$.
In particular, 
\[\big\{ F^l.\Puregeom_{\linkpatt}\;|\;\linkpatt\in\LP^{(s)}_{\multii},
\; l \in \{0,1,\ldots,s \}\big\}\]
is a basis of the subrepresentation 
$m_d\,\Wd_d\subset\bigotimes_{i=1}^p\Wd_{d_{i}}$, with $d = s + 1$
and $m_d = \# \LP^{(s)}_{\multii}$.
\end{description}
\end{thm}

A special case of the above problem was considered
in~\cite[Theorem~3.1]{Kytola-Peltola:Pure_partition_functions_of_multiple_SLEs}
where a particular basis for the trivial subrepresentation 
$\HWsp_{(1,1,\ldots,1,1)}^{(0)} \subset \Wd_2^{\tens 2N}$
was constructed. We state the result below in Theorem~\ref{thm: existence of multiple SLE vectors}.
In this case, all valences in 
$\multii = (1,1,\ldots,1,1)$ are equal to one: $s_i = 1$, for all $i$.
The solution to this special case is crucial in the proof of the general case in Section~\ref{subsec: construction}.

\begin{rem}\label{rem: one dimensionality of partition space}
\emph{ 
Let $\partition = (s_{1},\ldots,s_{|\partition|})$
be a partition of $s\geq 1$. 
Then, the space $\HWsp_\partition^{(s)}$ is one-dimensional: 
by Lemma~\ref{lem: cardinalities are equal}, we have 
$\dmn \HWsp_\partition^{(s)} = \# \LP^{(s)}_{\partition} 
= \# \set{\defpatt_\partition} = 1$.
In the the tensor product~\eqref{eq: order of tensorands}, the vector 
$\Puregeom_{\defpatt_\partition} \in \HWsp_\partition^{(s)}$
generates the highest dimensional subrepresentation isomorphic to $\Wd_{s+1}$ 
with multiplicity one. It is sometimes convenient to identify the space 
$\HWsp_\partition^{(s)}$ with $\bC$, via the map
$\Puregeom_{\defpatt_\partition} \mapsto 1 \in \bC$.
}
\end{rem}

The somewhat lengthy proof of Theorem~\ref{thm: highest weight vector space basis vectors} 
is distributed in the next subsections. The results obtained in 
Sections~\ref{subsec: powers of two-dimensionals}~--~\ref{subsec: uniqueness}
are put together in Section~\ref{subsec: proof}, which constitutes a summary of the proof.

We begin with introducing needed results concerning tensor products of 
two-dimensional irreducible representations of $\Uqsltwo$.
Throughout, we use the notations from~\eqref{eq: s in terms of d}~and~\eqref{eq: definition of n and s}.

\subsection{\label{subsec: powers of two-dimensionals}Tensor powers of two-dimensional irreducibles}

The tensor power $\Wd_2^{\tens s}$ of two-dimensional irreducible representations of $\Uqsltwo$
contains a unique subrepresentation of highest dimension, generated by 
the highest weight vector (a special case of the vectors 
in Remark~\ref{rem: one dimensionality of partition space})
\[ \MTbas_{0}^{(s)} := \Wbastwodim_{0} \tens \cdots \tens \Wbastwodim_{0} 
\, \in \, \Wd_2^{\tens s} .\]
This subrepresentation is isomorphic to $\Wd_d$, with $s = d-1$. 
For its basis, we use the notation 
\begin{align*}
\MTbas_{l}^{(s)} := F^l.\MTbas_{0}^{(s)}, \quad \text{ for } 
l \in \{0,1, \ldots s \},
\end{align*}
with convention $\MTbas_{l}^{(s)}=0$ 
when $l < 0$ or $l > s$. Using this basis, we define the projections 
\begin{align}\label{eq: projections}
\Projection = \Projection^{(s)} \colon & \; \Wd_2^{\tens s} \to \Wd_2^{\tens s}
\qquad \text{ and } \qquad
\Projectionhat^{(s)} \colon \Wd_2^{\tens s} \to \Wd_d
\end{align}
as follows. The map $\Projection^{(s)}$ is the projection onto the subrepresentation 
isomorphic to $\Wd_d$, so that we have 
\[ \Projection^{(s)}(\MTbas_{l}^{(s)}) = \MTbas_{l}^{(s)} ,
\quad\text{ for } l \in \{0, 1, \ldots, s \} \qquad \text{ and } \qquad
\Projection^{(s)}(v) = 0 , \quad \text{ for }
v \notin \mathrm{span} \set{\MTbas_{0}^{(s)},\ldots,\MTbas_{s}^{(s)}} \isom \Wd_d. \]
The map $\Projectionhat^{(s)}$ is defined as a composition of $\Projection^{(s)}$ with 
the identification $\MTbas_{l}^{(s)}\mapsto\Wbas_l^{(s)}$
of its image and $\Wd_d$, so that we have
$\Embedding^{(s)} \circ \Projectionhat^{(s)} = \Projection^{(s)}$, where 
$\Embedding^{(s)}$ is the embedding
\begin{align*}
\Embedding^{(s)} \colon \Wd_{d} \hookrightarrow \Wd_2^{\tens s} ,\qquad
\Embedding^{(s)}(\Wbas_l^{(d)}):=\MTbas_{l}^{(s)} ,
\quad\text{ for } l \in \{0, 1, \ldots, s \} .
\end{align*}

Vectors of $\Wd_{d}\subset\Wd_2^{\tens s}$ can be
characterized in terms of projections to subrepresentations 
in two consecutive tensorands. This property is used repeatedly in the proof of 
Theorem~\ref{thm: highest weight vector space basis vectors}.

\begin{lem}[{see, e.g., \cite[Lemma~2.4 \& Corollary~2.5]{Kytola-Peltola:Pure_partition_functions_of_multiple_SLEs}}]
\label{lem: all projections vanish}
For any $v\in\Wd_{2}^{\tens s}$, $s = d - 1 \in \bZpos$,
the following two conditions are equivalent.
\begin{align*}
\textnormal{\bf (a):} \;\; \text{$\hat{\pi}_{j}^{(1)}(v) = 0$, for all $j \in \{1, 2, \ldots, s - 1\}$,} 
\qquad \qquad \qquad
\textnormal{\bf (b):} \;\; \text{$v\in\Wd_{d} \subset \Wd_{2}^{\tens s}$}
\end{align*}
In particular, if we have $E.v=0$, $K.v=v$, and 
$\hat{\pi}_{j}^{(1)}(v) = 0$, for all 
$j \in \{1, 2, \ldots, s - 1\}$, then $v=0.$  
\end{lem}

Consider now the tensor product~\eqref{eq: order of tensorands} with $\multii = (1,1,\ldots,1,1) \in \bZpos^n$. 
By the decomposition~\eqref{eq: decomposition of tensor product}, we can write this tensor product in the form
\begin{align*}
\Wd_{2}^{\tens n} \isom \bigoplus_d m_d^{(n)} \, \Wd_d ,
\end{align*}
where, by Lemma~\ref{lem: cardinalities are equal}, the multiplicities are
$m_d^{(n)} = \# \PP_{N}^{(s)}$, with $N = \frac{1}{2}(n-s)$
and $s = d-1$. These numbers can be calculated explicitly
(see e.g.~\cite[Lemma~2.2]{Kytola-Peltola:Pure_partition_functions_of_multiple_SLEs}): 
we have
\begin{align*}
m_d^{(n)} = \# \PP_{N}^{(s)} 
=\; & \begin{cases}
\frac{2d}{n+d+1}\binom{n}{\frac{n+d-1}{2}} 
= \frac{s + 1}{N+s+1}\binom{2N+s}{N+s} \quad & \text{if } n+s \in 2\,\bZnn \text{ and } 0\leq s \leq n\\
0 & \text{otherwise} .
\end{cases}
\end{align*}
In particular, when $n = 2N$ (i.e., $s = 0$), 
the dimension of the trivial subrepresentation
\[ \HWsp_{2N}^{(0)} := \Big\{ v \in \Wd_{2}^{\tens 2N} \; \big| \;
    E.v=0 , \; K.v = v \Big\} \subset \Wd_{2}^{\tens 2N} \]
is the Catalan number $m_1^{(2N)} = \Catalan_N = \frac{1}{N+1} \binom{2N}{N}$.
For convenience, we also denote by 
$\HWsp_{n}^{(s)} = \HWsp_{(1,1,\ldots,1,1)}^{(s)}$
the $m_d^{(n)}$-dimensional spaces of highest weight vectors in 
$\Wd_{2}^{\tens n}$.

\subsection{\label{subsec: pure geometries}The special case $\multii = (1,1,\ldots,1,1)$}

In the proof of Theorem~\ref{thm: highest weight vector space basis vectors}, 
we make use of results of the article~\cite{Kytola-Peltola:Pure_partition_functions_of_multiple_SLEs}
concerning a particular basis of the trivial subrepresentation 
$\HWsp_{2N}^{(0)} \subset \Wd_{2}^{\tens 2N}$.
Then, the basis vectors $\Puregeomtwodim_\alpha$
are indexed by planar pair partitions $\alpha\in\PP_N$ of $2N$ points. 
They are uniquely characterized by the projection 
properties~\eqref{eq: multiple sle projection conditions} given below --- 
a special case of~\eqref{eq: projection conditions}.

Now, we consider the following linear system of equations for vectors
$\Puregeomtwodim_\alpha\in\Wd_{2}^{\tens 2N}$, with
$\alpha\in\PP_{N}$:
\begin{align}
&K.\Puregeomtwodim_\alpha = \Puregeomtwodim_\alpha
\label{eq: multiple sle cartan eigenvalue}\\
&E.\Puregeomtwodim_\alpha = 0
\label{eq: multiple sle highest weight vector}\\
&\hat{\pi}^{(1)}_{j}(\Puregeomtwodim_\alpha) =
\label{eq: multiple sle projection conditions}
\begin{cases} 0 & \mbox{if } \linkInEquation{j}{j+1}\notin\alpha\\ \Puregeomtwodim_{\hat{\alpha}} 
& \mbox{if } \linkInEquation{j}{j+1}\in\alpha , \end{cases}
\qquad\text{for all } j \in \{1,2,\ldots,2N-1 \},
\end{align}
where $\hat{\alpha} = \alpha\removeLink\link{j}{j+1} \in \PP_{N-1}$.

\begin{thm}{\cite[Theorem~3.1 \& Proposition~3.7]{Kytola-Peltola:Pure_partition_functions_of_multiple_SLEs}}
\label{thm: existence of multiple SLE vectors}
There exists a unique collection 
$\left(\Puregeomtwodim_{\alpha}\right)_{\alpha\in\PP}$ of vectors 
such that the system of 
equations~\eqref{eq: multiple sle cartan eigenvalue}~--~\eqref{eq: multiple sle 
projection conditions} holds for all $\alpha\in\PP$, 
and we have $\Puregeomtwodim_\emptyset=1$. For any $N \in \bZnn$, 
the collection $\left(\Puregeomtwodim_{\alpha}\right)_{\alpha\in\PP_N}$ is a basis of 
$\HWsp_{2N}^{(0)}$.
\end{thm}

The vectors $\Puregeomtwodim_\alpha$ are related to the 
pure partition functions $\PartF_\alpha(x_1,\ldots,x_{2N})$ of multiple $\SLE_\kappa$, 
with parameter $\kappa$ associated to the deformation parameter $q$ by 
$q = e^{\ii \pi 4 / \kappa}$; see Section~\ref{sec: multiple SLEs}, 
and~\cite{Kytola-Peltola:Pure_partition_functions_of_multiple_SLEs} for details.
Our general Theorem~\ref{thm: highest weight vector space basis vectors} 
concerns basis vectors $\Puregeom_{\linkpatt}$ of the space 
$\HWsp_\multii^{(s)}$, with $\multii \in \bZpos^p$. 
To these vectors, we can also associate 
functions $\BasisF_\linkpatt(x_1,\ldots,x_{p})$, as 
stated in Theorem~\ref{thm: asymptotic properties of general basis vectors}. 
These functions are solutions to
the Benoit~\& Saint-Aubin  PDEs, and they can be interpreted as 
pure partition functions for systems of random curves, where many curves may emerge from the same point, 
see~\cite{Dubedat:SLE_and_Virasoro_representations_fusion}.

\begin{figure}
\includegraphics[scale=.75]{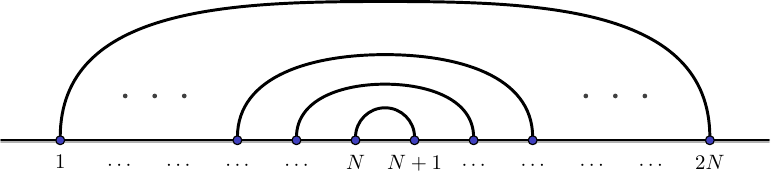}
\caption{\label{fig: rainbow pair partition}
The rainbow pattern (planar pair partition) of $2N$ points.}
\end{figure}

For the special case concerning the rainbow pattern 
defined by $\nested_{0} = \emptyset$ and
\begin{align*}
\nested_N = 
\left\lbrace\linkInEquation{1}{2N},\ldots,\linkInEquation{N-1}{N+2},
\linkInEquation{N}{N+1}\right\rbrace \,\in\, \PP_N ,
\quad\text{ for } N\in\bZpos,
\end{align*}
(see also Figure~\ref{fig: rainbow pair partition}), 
the equations~\eqref{eq: multiple sle cartan eigenvalue}~--~\eqref{eq: multiple sle projection conditions} 
involve only the rainbow patterns $\nested_N$ and $\nested_{N-1}$:
\begin{align}
(K-1).\Puregeomtwodim_{\nested_{N}} = \; & 0\label{eq: cartan eigenvalue for nested}\\
E.\Puregeomtwodim_{\nested_{N}} = \; & 0\label{eq: highest weight vector for nested}\\
\hat{\pi}^{(1)}_{N}(\Puregeomtwodim_{\nested_{N}}) = \; & \Puregeomtwodim_{\nested_{N-1}}\qquad\text{and}\qquad
\hat{\pi}^{(1)}_{j}(\Puregeomtwodim_{\nested_{N}}) =  0 , \quad\text{for }j\neq N .
\label{eq: projections for nested}
\end{align}
Therefore, the formula for $\Puregeomtwodim_{\nested_{N}}$ is 
particularly simple.
\begin{prop}{\cite[Proposition~3.3]{Kytola-Peltola:Pure_partition_functions_of_multiple_SLEs}}
\label{prop: solution for nested}
The vectors
\begin{align}\label{eq: solution for nested}
\Puregeomtwodim_{\nested_{N}} := \; & 
\frac{1}{(q^{-2}-1)^{N}}\frac{\qnum{2}^{N}}{\qfact{N+1}}\,
\sum_{l=0}^{N}(-1)^lq^{l(N-l-1)}\times
\left(\MTbas_l^{(N)}\tens\MTbas_{N-l}^{(N)}\right) \,\in\, \Wd_{2}^{\tens 2N} ,
\end{align}
for $N\in\bZnn$, determine the unique solution to
\eqref{eq: cartan eigenvalue for nested}~--~\eqref{eq: projections for nested} 
with $\Puregeomtwodim_\emptyset=1$.
\end{prop}

\subsection{\label{subsec: construction}Construction}

Now we construct the basis vectors $\Puregeom_\linkpatt$
of Theorem~\ref{thm: highest weight vector space basis vectors}.
In the construction, we use the vectors $\Puregeomtwodim_\alpha$
of Theorem~\ref{thm: existence of multiple SLE vectors},
with $N = \frac{1}{2}(n+s)$ chosen as in~\eqref{eq: particular choice of N},
and the map~\eqref{eq: link pattern to pair partition}, 
\begin{align*}
\LP^{(s)}_{\multii} \rightarrow \PP_{N},\qquad
\linkpatt \mapsto \alpha(\linkpatt),
\end{align*}
see also Figure~\ref{fig: open up map} in Section~\ref{subsec: combinatorial maps}.

For a link pattern $\linkpatt \in \LP^{(s)}_{\multii}$, the basis vector 
$\Puregeom_\linkpatt \in \HWsp_\multii^{(s)} \subset
\bigotimes_{i=1}^{p}\Wd_{d_{i}}$ is obtained from the vector
$\Puregeomtwodim_{\alpha(\linkpatt)} \in \HWsp_{2N}^{(0)} \subset \Wd_2^{\tens 2N}$ 
as follows: we let
\begin{align}\label{eq: basis vector construction}
\Puregeom_{\linkpatt} :=
 R_+^{(s)} \left(\Projectionhat^{(\multii,s)}(\Puregeomtwodim_{\alpha(\linkpatt)})\right)
\qquad\qquad \text{ with } \qquad\qquad
\xymatrixcolsep{3pc}
\xymatrix{
\HWsp_{2N}^{(0)} \; \; \ar@<1ex>[r]^{\Projectionhat^{(\multii,s)}}
   \ar@{<-^{)}}@<-2ex>[r]_{\Embedding^{(\multii,s)}}
   & \; \;  \HWsp_{(\multii,s)}^{(0)} 
   \ar[r]^{R_+^{(s)}} 
   & \; \;  \HWsp_\multii^{(s)} \\
  \Puregeomtwodim_{\alpha(\linkpatt)}  \; \; \ar@{|->}@<1ex>[r]^{\Projectionhat^{(\multii,s)}}
  \ar@{<-^{)}}@<-2ex>[r]_{\Embedding^{(\multii,s)}} &  \; \; \Puregeom_{\linkpatt}^{\infty} \; \;\ar@{|->}[r]^{R_+^{(s)}} 
  & \; \;\Puregeom_{\linkpatt},
}
\end{align}
where
\begin{itemize}
\item $\Projectionhat^{(\multii,s)} := \Projectionhat^{(s)}\tens\Projectionhat^{(s_p)}\tens\cdots\tens
\Projectionhat^{(s_1)}$ and 
$\Embedding^{(\multii,s)} := \Embedding^{(s)}\tens\Embedding^{(s_p)}\tens\cdots\tens
\Embedding^{(s_1)}$ (recall Section~\ref{subsec: powers of two-dimensionals}),
\item we denote by 
$\Puregeom_{\linkpatt}^{\infty} := \Projectionhat^{(\multii,s)}(\Puregeomtwodim_{\alpha(\linkpatt)})$, and
\item $R_+^{(s)} \colon \HWsp_{(\multii,s)}^{(0)} \to \HWsp_\multii^{(s)}$ 
is a linear isomorphism, 
which will be defined in more detail in Section~\ref{subsec: normalization}.
\end{itemize}

The idea is to think the tensor power 
$\Wd_2^{\tens 2N}$ of as a chain of blocks of smaller tensor powers of $\Wd_2$,
\begin{align*}
\Wd_2^{\tens 2N} = \Wd_2^{\tens s}\tens\Wd_2^{\tens s_p}\tens\Wd_2^{\tens s_{p-1}}\tens\cdots\tens\Wd_2^{\tens s_2}\tens\Wd_2^{\tens s_1},
\end{align*}
where each block $\Wd_2^{\tens r}$ is mapped onto the 
$\projdmn = r+1$-dimensional irreducible representation 
$\Wd_{\projdmn}$ under the map 
$\Projectionhat^{(\multii,s)} \colon \Wd_2^{\tens 2N} \to 
\Wd_{d}\tens\Wd_{d_p}\tens\cdots\tens\Wd_{d_1}$.
Conversely, the image of the embedding $\Embedding^{(\multii,s)}$ 
can be characterized by projection properties inside the blocks, as we show next.

\begin{prop} \label{prop: pure geometries lie in projected space}
The image of the space $\HWsp_{(\multii,s)}^{(0)}$
under the embedding $\Embedding^{(\multii,s)}$ is the space
\begin{align*}
\ImgofEmbedding^{(\multii,s)}_N :=
\set{ v \in \HWsp_{2N}^{(0)} \; \Big| \;
\hat{\pi}^{(1)}_{j}(v) = 0 , \text{ for all } j \in 
\big\{1,\ldots,2N-1\big\}\setminus
\bigg\{\sum_{i=1}^ks_i\;\big|\;1\leq k\leq p\bigg\}} .
\end{align*}
The projection $\Projectionhat^{(\multii,s)}$ defines 
an isomorphism of representations of $\Uqsltwo$,
\begin{align*}
\Projectionhat^{(\multii,s)} \;\colon\;
\ImgofEmbedding^{(\multii,s)}_N \to
\HWsp_{(\multii,s)}^{(0)},
\end{align*}
and its inverse is $\Embedding^{(\multii,s)}$.
For any $\linkpatt\in\LP^{(s)}_{\multii}$,
the vector $\Puregeomtwodim_{\alpha(\linkpatt)}$ lies in the space
$\ImgofEmbedding^{(\multii,s)}_N$ and, in particular, 
\begin{align}\label{eq: pure geometries lie in projected space}
\Embedding^{(\multii,s)} 
\left(\Projectionhat^{(\multii,s)}
(\Puregeomtwodim_{\alpha(\linkpatt)})\right) 
= \Embedding^{(\multii,s)} 
\left(\Puregeom_{\linkpatt}^{\infty}\right)
= \Puregeomtwodim_{\alpha(\linkpatt)}.
\end{align}
\end{prop}
\begin{proof}
The property
$\Embedding^{(\multii,s)} 
\big( \HWsp_{(\multii,s)}^{(0)}\big)
= \ImgofEmbedding^{(\multii,s)}_N$
follows from Lemma~\ref{lem: all projections vanish} and the fact that 
$\Embedding^{(\multii,s)}$ commutes with the action of the algebra 
$\Uqsltwo$. Since $\Projectionhat^{(\multii,s)}$ also commutes with the action of $\Uqsltwo$, it follows that restricted to 
$\ImgofEmbedding^{(\multii,s)}_N$, it is an isomorphism of representations, with inverse $\Embedding^{(\multii,s)}$.

Let then $\linkpatt\in\LP^{(s)}_{\multii}$. 
By definition of the map $\linkpatt \mapsto \alpha(\linkpatt)$
in Section~\ref{subsec: combinatorial maps}, the planar pair partition
$\alpha(\linkpatt)$ can contain a link of type $\link{j}{j+1}$ only if these 
points correspond to different points in the link pattern $\linkpatt$,
that is, if $j \in \set{\sum_{i=1}^ks_i\;\big|\;1\leq k\leq p}$.
In particular, by the projection 
properties~\eqref{eq: multiple sle projection conditions} of 
$\Puregeomtwodim_{\alpha(\linkpatt)}$, we have 
$\hat{\pi}^{(1)}_j(\Puregeomtwodim_{\alpha(\linkpatt)})=0$,
for all $j \notin \set{\sum_{i=1}^ks_i\;\big|\;1\leq k\leq p}$, so  
$\Puregeomtwodim_{\alpha(\linkpatt)} \in 
\ImgofEmbedding^{(\multii,s)}_N$.
This concludes the proof.
\end{proof}

It now follows almost immediately from the definitions that the 
vectors~\eqref{eq: basis vector construction} form a basis
of the highest weight vector space. This proves part (b)
of Theorem~\ref{thm: highest weight vector space basis vectors}.

\begin{prop}\label{prop: basis of highest weight vector space}
The collection
$\left(\Puregeom_{\linkpatt}\right)_{\linkpatt\in\LP^{(s)}_{\multii}}$
defined in~\eqref{eq: basis vector construction}
is a basis of the vector space $\HWsp_\multii^{(s)}$.
\end{prop}
\begin{proof}
Because $R_+^{(s)} \colon \HWsp_{(\multii,s)}^{(0)} \to \HWsp_\multii^{(s)}$ 
is a linear isomorphism 
(by~\cite[Lemma~5.3]{Kytola-Peltola:Conformally_covariant_boundary_correlation_functions_with_quantum_group}), the vectors 
$\Puregeom_{\linkpatt} := R_+^{(s)} \left(\Projectionhat^{(\multii,s)}(\Puregeomtwodim_{\alpha(\linkpatt)})\right)$ 
belong to the space $\HWsp_\multii^{(s)}$
by construction. Their linear independence follows the facts that, first,
the maps $R_+^{(s)}$ and 
$\Projectionhat^{(\multii,s)}$ are linear isomorphisms,
by~\cite[Lemma~5.3(a)]{Kytola-Peltola:Conformally_covariant_boundary_correlation_functions_with_quantum_group}
and Proposition~\ref{prop: pure geometries lie in projected space}, respectively, and second, the vectors 
$\Puregeomtwodim_{\alpha(\linkpatt)}$ are linearly independent, by 
Theorem~\ref{thm: existence of multiple SLE vectors}. Finally, by
Lemma~\ref{lem: cardinalities are equal}, the linear span of the vectors 
$\Puregeom_{\linkpatt}$, for $\linkpatt \in \LP^{(s)}_{\multii}$,
has the correct dimension 
$\# \LP^{(s)}_{\multii} = \dmn \HWsp_\multii^{(s)}$.
\end{proof}

To prove part (a)
of Theorem~\ref{thm: highest weight vector space basis vectors},
we still have to show that the vectors $\Puregeom_{\linkpatt}$ satisfy
\eqref{eq: projection conditions}~--~\eqref{eq: normalization of basis vectors}.
The projection properties~\eqref{eq: projection conditions} will be verified in  
Section~\ref{subsec: projection properties}.
The normalization conditions~\eqref{eq: normalization of basis vectors}
follow by considering the action of the map $R_+^{(s)}$ 
on the vectors $\Puregeomtwodim_{\nested_{s}}$, associated to 
the rainbow link patterns $\nested_{s}$.

\subsection{\label{subsec: normalization}Normalization}

For any partition $\partition$ of $s$,
the vectors $\Puregeom_{\defpatt_\partition}$ correspond to $\Puregeomtwodim_{\nested_{N}}$ with $N=s$
under the map $\linkpatt \mapsto \alpha(\linkpatt)$ --- see Figure~\ref{fig: defects and rainbow} for an illustration.
This observation gives rise to the normalization constant in~\eqref{eq: normalization of basis vectors}, 
as we show next. 

\begin{figure}
\includegraphics[scale=.75]{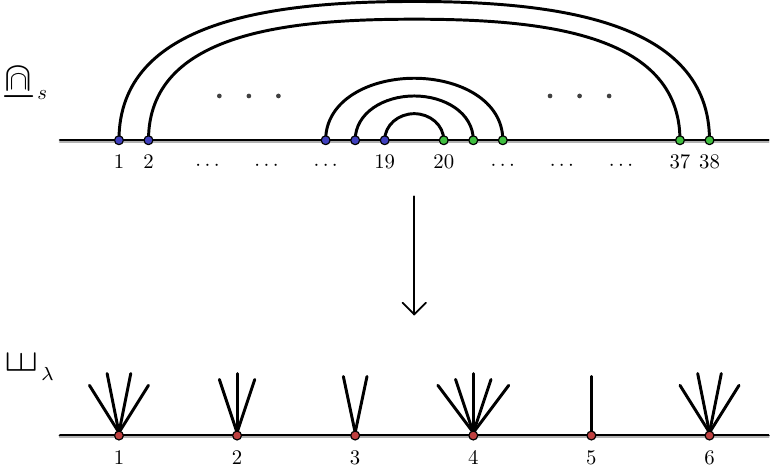}
\caption{\label{fig: defects and rainbow}
For the link pattern {$\defpatt_\partition$}, 
consisting of {$s$} defects,
the corresponding planar pair partition {$\alpha(\defpatt_\partition)$}
is the rainbow pattern {$\nested_{s} = \alpha(\defpatt_\partition)$}.
}
\end{figure}

First, we give the precise definition of the linear isomorphism $R_+^{(s)}$ already used above 
in Equation~\eqref{eq: basis vector construction}.
By \cite[Lemma~5.3(a)]{Kytola-Peltola:Conformally_covariant_boundary_correlation_functions_with_quantum_group},
any vector $v \in \HWsp_{(\multii,s)}^{(0)}$
can be written in the form 
\[v=\sum_{l=0}^{s}(-1)^{s-l}
q^{(l+1)(s-l)} 
\times (\Wbas_l^{(d)}\tens F^{s-l}.\tau_0^+) ,\]
for a unique vector 
$\tau_0^+\in\HWsp_{\multii}^{(s)}$, with $d = s + 1$.
The map (compare with $\Rpluscomb^{(s)}$ in 
Section~\ref{subsec: combinatorial maps})
\begin{align}\label{eq: R plus}
R_+ = R_+^{(s)} \;\colon\; \HWsp_{(\multii,s)}^{(0)} \to
\HWsp_{\multii}^{(s)}, \qquad \qquad 
R_+^{(s)}(v) := \tau_0^+,
\end{align}
is thus well-defined. 
It was shown in~\cite[Lemma~5.3]{Kytola-Peltola:Conformally_covariant_boundary_correlation_functions_with_quantum_group}
that $R_+^{(s)}$ is a linear isomorphism.

\begin{rem}\label{rem: Rplus commutes with projections}
\emph{
The map $R_+^{(s)}$ commutes with the maps 
$\widetilde{\pi}_{j}^{(\projdmn)}$
defined in~\eqref{eq: generalized projection},
for any $j \in \{1,2,\ldots,p-1\}$, $m \in \{0,1,\ldots,\min(s_j,s_{j+1})\}$, 
and $\projdmn=s_{j}+s_{j+1}+1-2m$, because the maps $\widetilde{\pi}_{j}^{(\projdmn)}$ 
act on the tensor components $(j,j+1)$ of the tensor product
$\Wd_d \tens \Wd_{d_p} \tens \Wd_{d_{p-1}} \tens \cdots \tens \Wd_{d_2} \tens \Wd_{d_1}$, 
away from the tensor position involving $\Wd_d$
--- see also~\cite[Lemma~5.3 \& Equation~(5.2)]{Kytola-Peltola:Conformally_covariant_boundary_correlation_functions_with_quantum_group}.
}
\end{rem}

\begin{lem}\label{lem: images of rainbow vectors}
Let $\partition=(s_{1},\ldots,s_{|\partition|})$ be a partition of 
$s \in \bZpos$. Then we have 
$\alpha(\defpatt_\partition) = \nested_{s}$, and
\begin{align*}
\Puregeom_{\defpatt_\partition} :=
R_+^{(s)}\left(\Projectionhat^{(\partition,s)}(\Puregeomtwodim_{\nested_{s}})\right) =
\frac{1}{(q-q^{-1})^{s}}\frac{\qnum{2}^{s}}{\qfact{s+1}}
\times (\Wbas_0^{(s_{|\partition|}+1)}\tens\cdots\tens\Wbas_0^{(s_1+1)}) 
\, \in\,\HWsp_\partition^{(s)}.
\end{align*}
\end{lem}
\begin{proof}
The first assertion $\alpha(\defpatt_\partition) = \nested_{s}$
is immediate from the definition of the map 
$\linkpatt \mapsto \alpha(\linkpatt)$.
For the second assertion, using the formula~\eqref{eq: solution for nested} for the vector 
$\Puregeomtwodim_{\nested_{s}}$, we calculate 
the image of $\Puregeomtwodim_{\nested_{s}}$
under the map
$\Projectionhat^{(\partition,s)} 
= \Projectionhat^{(s)}\tens\Projectionhat^{\partition} 
= \Projectionhat^{(s)}\tens
(\Projectionhat^{(s_{|\partition|})}\tens\cdots\tens\Projectionhat^{(s_1)})$:
\begin{align*}
\Projectionhat^{(\partition,s)}(\Puregeomtwodim_{\nested_{s}}) = \; &
\frac{1}{(q^{-2}-1)^{s}}\frac{\qnum{2}^{s}}{\qfact{s+1}}\,
\sum_{l=0}^{s}(-1)^lq^{l(s-l-1)}\times
\left(\Projectionhat^{(s)}(\MTbas_l^{(s)}) \tens 
\Projectionhat^{\partition}(\MTbas_{s-l}^{(s)})\right) \\
= \; & \frac{1}{(q^{-2}-1)^{s}}\frac{\qnum{2}^{s}}{\qfact{s+1}}\,
\sum_{l=0}^{s}(-1)^lq^{l(s-l-1)}\times
\left(\Wbas_l^{(d)} \tens F^{s-l}.(\Wbas_0^{(s_{|\partition|}+1)}\tens\cdots\tens\Wbas_0^{(s_1+1)})\right).
\end{align*}
The second assertion now follows from the definition of the map $R_+^{(s)}$ in~\eqref{eq: R plus}.
\end{proof}

\subsection{\label{subsec: projection properties}Projection properties}

We show next that the vectors $\Puregeom_{\linkpatt}$ 
defined by~\eqref{eq: basis vector construction}
indeed satisfy the projection properties~\eqref{eq: projection conditions}.
To establish this, we need some auxiliary calculations, 
given in Appendix~\ref{app: auxiliary calculations}.
The crucial observation is the following commutative diagram.
\begin{lem}\label{lem: commutative diagram}
Let $s_1,s_2 \in \bZpos$ and $m \in \{1,2,\ldots,\min(s_1,s_2) \}$, 
and denote $r = s_1+s_2-2m$ and $\projdmn = r+1$.  
The following diagram commutes, up to a non-zero multiplicative constant,
given below.
\begin{align*}
\xymatrixcolsep{3pc}\xymatrixrowsep{3.2pc}
\xymatrix{
   \Wd_{2}^{\tens s_2}\tens\Wd_{2}^{\tens s_1} \; \; \ar@{<-^{)}}[rr]^{\hspace{.9cm}\Embedding^{(s_2)} \, \tens\, \Embedding^{(s_1)}}\ar[d]_{\hat{\pi}_{s_1}^{(1)}\;}
  & & \; \; \Wd_{d_2} \tens \Wd_{d_1} \ar[ddddd]^{\hat{\pi}^{(\projdmn)}} \\
    \Wd_{2}^{\tens (s_2-1)} \tens \Wd_{2}^{\tens (s_1-1)} \; \; \ar[d]_{\hat{\pi}_{s_1-1}^{(1)}\;}
  & & \; \; \\
    \vdots \ar[d]
  & & \; \; \\
    \Wd_{2}^{\tens (s_2-m+1)} \tens \Wd_{2}^{\tens (s_1-m+1)} \; \; \ar[d]_{\hat{\pi}_{s_1-m+1}^{(1)}\;}
  & & \; \; \\  
    \Wd_2^{\tens r}  \; \; \ar[d]_{\Projectionhat^{(r)}\;}
  & & \; \; \\  
    \Wd_{\projdmn} \; \; \ar@{<->}[rr]_{\isom}
  & & \; \;  \Wd_{\projdmn}
} 
\end{align*}
More precisely, we have
\begin{align*}
\Projectionhat^{(r)}\circ\left(\hat{\pi}^{(1)}_{s_1-m+1}\circ\cdots\circ\hat{\pi}^{(1)}_{s_1-1}
\circ\hat{\pi}^{(1)}_{s_1}\right)\circ
\left(\Embedding^{(s_2)}\tens\Embedding^{(s_1)}\right)
= & \; \constantfromdiagram{m}{s_1}{s_{2}}\times\hat{\pi}^{(\projdmn)},
\end{align*}
where the non-zero constant equals 
\[ \constantfromdiagram{m}{s_1}{s_{2}} =
\frac{\qfact{s_1-m}\qfact{s_2-m}\qfact{s_1+s_2-m+1}}{\qnum{2}^m\qfact{s_1}\qfact{s_2}\qfact{s_1+s_2-2m+1}} =
\frac{\qbin{s_1+s_2-m+1}{m}}{\qnum{2}^m\qfact{m}\qbin{s_1}{m}\qbin{s_2}{m}} .\]
\end{lem}
\begin{proof}
The subrepresentation isomorphic to $\Wd_{\projdmn}$ appears in the tensor product 
$\Wd_{d_2}\tens\Wd_{d_1}$  with multiplicity one. By Schur's lemma, to prove that 
the diagram commutes, it therefore suffices to show that the map
$\Projectionhat^{(r)}\circ\left(\hat{\pi}^{(1)}_{s_1-m+1}\circ\cdots\circ\hat{\pi}^{(1)}_{s_1-1}
\circ\hat{\pi}^{(1)}_{s_1}\right)\circ
\left(\Embedding^{(s_2)}\tens\Embedding^{(s_1)}\right)$
is non-zero. But, by Lemma~\ref{lem: projections in the middle for hwv},
the vector $\Tbas_{0}^{(\projdmn;d_1,d_2)}\in\Wd_{d_2}\tens\Wd_{d_1}$ 
maps to a non-zero multiple of $\Wbas_0^{(\projdmn)} \in \Wd_{\projdmn}$ in this 
map, with the explicit, non-zero proportionality constant $\constantfromdiagram{m}{s_1}{s_{2}}$. This finishes the proof.
\end{proof}

\begin{prop}\label{prop: commutative diagram}
The collection of vectors
$\left(\Puregeom_{\linkpatt}\right)_{\linkpatt\in\LP^{(s)}}$, 
defined in~\eqref{eq: basis vector construction}, 
satisfies the equations~\eqref{eq: projection conditions}.
\end{prop}
\begin{proof}
By Remark~\ref{rem: Rplus commutes with projections}, the maps 
$\tilde{\pi}_j^{(\projdmn)}$ appearing in 
the equations~\eqref{eq: projection conditions} 
commute with the linear isomorphism 
$R_+^{(s)}$, for any $j \in \{1,2,\ldots,p-1 \}$,
$m \in \{0,1,\ldots,\min(s_j,s_{j+1})\}$, and $\projdmn = s_{j}+s_{j+1}+1-2m$.
Therefore, it suffices to show that 
the vector $\Puregeom_{\linkpatt}^{\infty} := 
\Projectionhat^{(\multii,s)}(\Puregeomtwodim_{\alpha(\linkpatt)})$ 
satisfies the properties~\eqref{eq: projection conditions}. 
Using the commutative diagram of Lemma~\ref{lem: commutative diagram}
together with Proposition~\ref{prop: pure geometries lie in projected space},
the properties~\eqref{eq: projection conditions} can be written in terms of $\Puregeomtwodim_{\alpha(\linkpatt)}$, which, in turn, are known to satisfy
the properties~\eqref{eq: multiple sle projection conditions},
by Theorem~\ref{thm: existence of multiple SLE vectors}.

Fix $j \in \{1,2,\ldots,p-1 \}$, $m \in \{1,2,\ldots,\min(s_j,s_{j+1}) \}$,
and denote by $k_j = \sum_{i=1}^{j}s_i$. We first note that,
by definition of the map $\linkpatt \to \alpha(\linkpatt)$
(see Section~\ref{subsec: combinatorial maps}),
the link pattern $\alpha(\linkpatt)$ contains the nested links
\[ \linkInEquation{k_j}{k_j+1},\linkInEquation{k_j-1}{k_j+2},\ldots,\linkInEquation{k_j-m+1}{k_j+m} \]
if and only if the link pattern $\linkpatt$ contains at least 
$m$ links $\link{j}{j+1}$, and if this is the case, then we have
\begin{align*}
\alpha(\linkpatt) \; \removeLink\linkInEquation{k_j}{k_j+1} \; 
\removeLink\linkInEquation{k_j-1}{k_j+2} \; \removeLink \; \ldots \;\removeLink\linkInEquation{k_j-m+1}{k_j+m} 
= \; & \alpha \Big(\linkpatt \removeLink (m\times\linkInEquation{j}{j+1})\Big)
= \alpha( \hat{\linkpatt} ),
\end{align*}
where we denote by 
$\hat{\linkpatt} = \linkpatt \removeLink (m\times\link{j}{j+1})$.
The projection properties~\eqref{eq: multiple sle projection conditions}
for the vector $\Puregeomtwodim_{\alpha(\linkpatt)}$ show that 
\begin{align}\label{eq: taking nested links off}
\left(\hat{\pi}_{k_j-m+1}^{(1)} \circ \cdots \circ \hat{\pi}_{k_j-1}^{(1)} \circ \hat{\pi}_{k_j}^{(1)}\right)
(\Puregeomtwodim_{\alpha(\linkpatt)})
= \; & \begin{cases} 
\Puregeomtwodim_{\alpha (\hat{\linkpatt})} & \mbox{if there are at least $m$ links } \linkInEquation{j}{j+1} \text{ in }\linkpatt \\
0 & \mbox{otherwise}. \end{cases}
\end{align}
Denote by $r = \projdmn - 1 = s_{j}+s_{j+1}-2m$ and
$\hat{\multii} = (s_1,\ldots,s_{j-1},r,s_{j+2},\ldots,s_p)$.
Using the commutative diagram of Lemma~\ref{lem: commutative diagram},
Equation~\eqref{eq: pure geometries lie in projected space},
and Equation~\eqref{eq: taking nested links off}, we obtain
\begin{align}\label{eq: projection hat taking nested links off}
\constantfromdiagram{m}{s_j}{s_{j+1}} \times
\hat{\pi}_j^{(\projdmn)}(\Puregeom_\linkpatt^{\infty})
= \; & \left( \Projectionhat^{(\hat{\multii},s)} \circ \left(\hat{\pi}_{k_j-m+1}^{(1)} \circ \cdots \circ \hat{\pi}_{k_j-1}^{(1)} \circ \hat{\pi}_{k_j}^{(1)} \right)
\circ\Embedding^{(\multii,s)}\right)(\Puregeom_\linkpatt^{\infty})
\\
= \; & \left( \Projectionhat^{(\hat{\multii},s)} \circ
\left(\hat{\pi}_{k_j-m+1}^{(1)} \circ \cdots \circ \hat{\pi}_{k_j-1}^{(1)} \circ \hat{\pi}_{k_j}^{(1)}\right)\right)
(\Puregeomtwodim_{\alpha(\linkpatt)}) 
\nonumber \\
= \; & \begin{cases} \Projectionhat^{(\hat{\multii},s)}
(\Puregeomtwodim_{\alpha (\hat{\linkpatt})}) 
& \mbox{if there are at least $m$ links } \linkInEquation{j}{j+1} \text{ in }\linkpatt \\
0 & \mbox{otherwise}. \end{cases} \nonumber
\end{align}

Now, it follows directly from the definitions that we have
\begin{align}\label{eq: consistency of projections}
\Projectionhat^{(s_{j+1}-m)}\tens\Projectionhat^{(s_{j}-m)} =
\iota^{(\projdmn)} \circ \Projectionhat^{(r)} \; \colon \;
\Wd_{2}^{\tens r} \to \Wd_{d_{j+1}-m} \tens \Wd_{d_{j}-m} ,
\end{align}
where the maps
$\Projectionhat^{(r)} \colon \Wd_{2}^{\tens r} \to \Wd_\projdmn$ and
$\iota^{(\projdmn)} = \iota^{(\projdmn;d_j-m,d_{j+1}-m)} \colon \Wd_\projdmn \hookrightarrow \Wd_{d_{j+1}-m} \tens \Wd_{d_{j}-m}$
were defined in~\eqref{eq: projections}
and~\eqref{eq: embedding of a subrepresentation}, respectively. 
Using Equation~\eqref{eq: generalized projection},
Equation~\eqref{eq: projection hat taking nested links off}, 
and Equation~\eqref{eq: consistency of projections}, we obtain
\begin{align*}
\constantfromdiagram{m}{s_j}{s_{j+1}} \times 
\tilde{\pi}_j^{(\projdmn)}(\Puregeom_\linkpatt^{\infty})
= \; & \constantfromdiagram{m}{s_j}{s_{j+1}} \times 
\iota_{j}^{(\projdmn)} \left( \hat{\pi}_{j}^{(\projdmn)} (\Puregeom_\linkpatt^{\infty}) \right) \\
= \; & \begin{cases} 
\left( \iota_{j}^{(\projdmn)} \circ \Projectionhat^{(\hat{\multii},s)} \right)
(\Puregeomtwodim_{\alpha (\hat{\linkpatt})}) & \mbox{if there are at least $m$ links } \linkInEquation{j}{j+1} \text{ in }\linkpatt \\
0 & \mbox{otherwise} \end{cases} \\
= \; & \begin{cases}
\Projectionhat^{(\tilde{\multii},s)}
(\Puregeomtwodim_{\alpha (\hat{\linkpatt})}) & \mbox{if there are at least $m$ links } \linkInEquation{j}{j+1} \text{ in }\linkpatt \\
0 & \mbox{otherwise} \end{cases} \\
= \; & \begin{cases} \Puregeom_{\hat{\linkpatt}}^{\infty}
 & \mbox{if there are at least $m$ links } \linkInEquation{j}{j+1} \text{ in }\linkpatt \\
0 & \mbox{otherwise}, \end{cases}
\end{align*}
where $\tilde{\multii} = (s_1,\ldots,s_{j-1},s_j-m,s_{j+1}-m,s_{j+2},\ldots,s_p)$,
and $\Puregeom_{\hat{\linkpatt}}^{\infty} = \Projectionhat^{(\tilde{\multii},s)} (\Puregeomtwodim_{\alpha (\hat{\linkpatt})})$.
This is the property~\eqref{eq: projection conditions} for 
$\Puregeom_{\linkpatt}^{\infty}$. Finally, we obtain the asserted property for 
$\Puregeom_\linkpatt$ by applying the map $R_+^{(s)}$:
\begin{align*}
\tilde{\pi}_j^{(\projdmn)}(\Puregeom_\linkpatt)
= \; & \begin{cases} 
\frac{1}{\constantfromdiagram{m}{s_j}{s_{j+1}}}\times\Puregeom_{\hat{\linkpatt}} & \mbox{if there are at least $m$ links } \linkInEquation{j}{j+1} \text{ in }\linkpatt\\ 0 & \mbox{otherwise}. \end{cases}
\end{align*}
\end{proof}

\subsection{\label{subsec: uniqueness}Uniqueness}

We finish by proving that the solutions 
to~\eqref{eq: cartan eigenvalue}~--~\eqref{eq: projection conditions}
are necessarily unique, up to normalization.
Fixing the normalization~\eqref{eq: normalization of basis vectors}, uniqueness follows from the observation 
that the homogeneous system, in which all of the projections vanish, admits 
a non-trivial solution only when $n = \sum_{i=1}^p s_i = s$, and in this 
case, the space $\HWsp_\multii^{(s)}$ is one-dimensional
and spanned by $\Puregeom_{\defpatt_\multii}$ (see 
Remark~\ref{rem: one dimensionality of partition space}).

\begin{lem}\label{lem: all projections vanish gives zero}
Assume that $n > s$ and that the vector $v\in\bigotimes_{i=1}^{p}\Wd_{d_{i}}$ satisfies
$E.v=0$, $K.v=q^{s}\,v$, and $\widetilde{\pi}_{j}^{(\projdmn)}(v) = 0$, for all $j \in \{1,2,\ldots,p-1\}$,
all $\projdmn=s_{j}+s_{j+1}+1-2m$, and all $m \in \{ 1,2,\ldots,\min(s_{j+1},s_{j})\}$.
Then we have $v=0.$  
\end{lem}
\begin{proof}
The properties $E.v=0$ and $K.v=q^{s}\,v$ show that 
$v$ belongs to the highest weight vector space 
$\HWsp_\multii^{(s)}$, so we have
$\Embedding^{(\multii)}(v) = (\Embedding^{(s_p)}\tens\cdots\tens
\Embedding^{(s_1)})(v) \in \HWsp_n^{(s)}$ as well.
Furthermore, the properties $\widetilde{\pi}_{j}^{(\projdmn)}(v) = 0$,
for all $j$, $\projdmn$, and $m$ imply that in 
the tensor product~\eqref{eq: order of tensorands}, in the direct sum 
decomposition of any two consecutive tensorands $\Wd_{d_{j+1}}\tens\Wd_{d_j}$
into irreducibles, the vector $v$ lies in the highest dimensional subrepresentation
isomorphic to $\Wd_{d_j+d_{j+1}-1}$. For the vector 
$\Embedding^{(\multii)}(v)$, this and Lemma~\ref{lem: all projections vanish}
show that we have
\[\hat{\pi}^{(1)}_k\left(\Embedding^{(\multii)}(v)\right)= 0 ,
\qquad\text{for all }\;k \in \{1,2,\ldots,n-1\} .\] 
Therefore, Lemma~\ref{lem: all projections vanish} applied to the whole tensor product $\Wd_{2}^{\tens n}$ shows that the vector $\Embedding^{(\multii)}(v)$ 
belongs to the highest dimensional subrepresentation
$\Wd_{n+1}\subset\Wd_2^{\tens n}$. We conclude that
\[ \Embedding^{(\multii)}(v) \in \Wd_{s+1} \cap \Wd_{n+1} 
\subset \Wd_2^{\tens n} .\]
Now, by assumption $n > s$, we have 
$\Wd_{s+1} \cap \Wd_{n+1}  = \set{0}$, so we get
$\Embedding^{(\multii)}(v) = 0$, and $v = 0$ as well.
\end{proof}

\begin{prop} \label{prop: uniqueness}
Let $s\in\bZnn$, and let 
$\left(\Puregeom_{\linkpatt}\right)_{\linkpatt\in\LP^{(s)}}$ and
$\left(\Puregeom'_{\linkpatt}\right)_{\linkpatt\in\LP^{(s)}}$ 
be two collections of solutions 
to~\eqref{eq: cartan eigenvalue}~--~\eqref{eq: projection conditions},
such that we have
$\Puregeom_{\defpatt_\partition},\Puregeom_{\defpatt_\partition}'\neq0$,
for all partitions $\partition$ of $s$.
Then, there are constants $c_\partition\in\bC\setminus\{0\}$ such that
\begin{align*}
\Puregeom'_{\linkpatt}=c_\partition\,\Puregeom_{\linkpatt}, \qquad\text{for all }
\linkpatt \in \LP^{(s)}(\partition).
\end{align*}
\end{prop}
\begin{proof}
Fix a partition $\partition$ of $s$. By assumption, we have
$\Puregeom_{\defpatt_\partition}' = c_\partition\,\Puregeom_{\defpatt_\partition}$,
for some $c_\partition\in\bC\setminus\{0\}$, because the vectors
$\Puregeom_{\defpatt_\partition}'$ and $\Puregeom_{\defpatt_\partition}$
belong to a one-dimensional space 
(see Remark~\ref{rem: one dimensionality of partition space}).
Suppose then that the condition
$\Puregeom' = c_\partition\,\Puregeom_\tau$ holds for all 
$\tau \in \LP^{(s)}(\partition) \cap \LP_\varrho^{(s)}$ 
for which the multiindex $\varrho = (r_1,\ldots,r_t)$ satisfies 
$\sum_{i=1}^t r_i = n \geq s$. Then, for any 
$\linkpatt \in \LP^{(s)}(\partition) \cap \LP_\multii^{(s)}$ with $\multii = (s_1,\ldots,s_p)$ such that
$\sum_{i=1}^p s_i = n+1$, the 
equations~\eqref{eq: cartan eigenvalue}~--~\eqref{eq: projection conditions}
for $\Puregeom_\linkpatt'$ and $c_\partition\,\Puregeom_\linkpatt$ coincide.
It thus follows from Lemma~\ref{lem: all projections vanish gives zero} that
we have $\Puregeom_\linkpatt' = c_\partition\,\Puregeom_\linkpatt$,
for all $\linkpatt \in \LP^{(s)}(\partition) \cap \LP_\multii^{(s)}$.
The assertion then follows by induction on $n$.
\end{proof}

\subsection{\label{subsec: proof}Proof of Theorem~\ref{thm: highest weight vector space basis vectors}}

For $\linkpatt\in\LP^{(s)}$, the vectors 
$\Puregeom_{\linkpatt}$ defined in Equation~\eqref{eq: basis vector construction}
are solutions to~\eqref{eq: cartan eigenvalue}~--~\eqref{eq: projection conditions};
see Proposition~\ref{prop: basis of highest weight vector space}
for the conditions \eqref{eq: cartan eigenvalue}
and \eqref{eq: highest weight vector}, and 
Proposition~\ref{prop: commutative diagram} for~\eqref{eq: projection conditions}.
By Lemma~\ref{lem: images of rainbow vectors}, 
the vectors $\Puregeom_{\defpatt_\partition}$ satisfy the asserted 
normalization~\eqref{eq: normalization of basis vectors},
for all partitions $\partition$ of $s$.
Uniqueness of the solutions follows from Proposition~\ref{prop: uniqueness}. 
This proves part (a). Part (b) follows from 
Proposition~\ref{prop: basis of highest weight vector space}.
This concludes the proof of
Theorem~\ref{thm: highest weight vector space basis vectors}. $\hfill\qed$

\bigskip{}
\section{\label{sec: cyclic permutation symmetry}Cyclic permutation symmetry of the basis vectors}

Next, we derive a further property of the basis vectors 
$\Puregeom_\linkpatt$, also very natural in terms of 
the link patterns $\linkpatt$. This is a symmetry property under cyclic 
permutations of the tensor components $\Wd_{d_i}$ in the trivial subrepresentation
$\HWsp_\multii^{(0)} \subset \bigotimes_{i=1}^{p}\Wd_{d_{i}}$.
We show in Corollary~\ref{cor: cyclic symmetry} that under such a cyclic 
permutation, the vectors $\Puregeom_\linkpatt \in \HWsp_\multii^{(0)}$, with
$\multii = (s_1,\ldots,s_p)$, are mapped to constant multiples of similar 
vectors $\Puregeom_{\linkpatt'} \in \HWsp^{(0)}_{\multii'}$, 
where the link pattern $\linkpatt' \in \LP^{(0)}_{\multii'}$
is obtained by applying the combinatorial bijections $\Rpluscomb$ and 
$\Rminuscomb$ of Section~\ref{subsec: combinatorial maps}, 
so that we either have $\multii' = (s_p,s_1,s_2,\ldots,s_{p-1})$ or 
$\multii' = (s_2,s_3,\ldots,s_p,s_1)$, depending on the orientation of 
the permutation --- see Figure~\ref{fig: cyclic permutation} for 
an illustration of the former case. From this property, it also follows 
(Corollary~\ref{cor: iterate of Smap is constant multiple of identity})
that the $p$:th iterate of the cyclic permutation of the tensor components is 
a constant multiple of the identity map on $\HWsp_\multii^{(0)}$.

\subsection{\label{subsec: two infinities}Placing tensor components at infinity}

We now consider the linear isomorphism $R_+$
defined in Section~\ref{subsec: normalization}, and a similar linear isomorphism
\begin{align}
R_- = R_-^{(s)} \;\colon\; \HWsp_{(s,\multii)}^{(0)} \to
\HWsp_{\multii}^{(s)}, \qquad \qquad 
R_-^{(s)}(v) := \tau_0^-,
\end{align}
where $\tau_0^-\in\HWsp_{\multii}^{(s)}$ is the unique vector such that we have
\[v=\sum_{l=0}^{s}(-1)^{s-l}q^{(l-1)(s-l)} \times ( F^{s-l}.\tau_0^-\tens \Wbas_{l}^{(d)}), \]
with $d = s + 1$; see \cite[Lemma~5.3]{Kytola-Peltola:Conformally_covariant_boundary_correlation_functions_with_quantum_group}
(compare also with $\Rpluscomb$ and $\Rminuscomb$ in 
Section~\ref{subsec: combinatorial maps}).

We also define the composed map $S = R_-^{-1} \circ R_+$, permuting the tensor 
components cyclically, 
\begin{align}\label{eq: Smap}
S = S^{(s)} \;\colon\; \HWsp_{(\multii,s)}^{(0)} \to \HWsp_{(s,\multii)}^{(0)}.
\end{align}
Iterating the map $S$, we define a linear map on $\HWsp_{\multii}^{(0)}$,
with $\multii = (s_1,\ldots,s_p)$,
\begin{align}\label{eq: iteration of Smap}
S^{(s_1)}\circ S^{(s_2)}\circ\cdots\circ S^{(s_{p-1})}\circ S^{(s_p)} \;\colon\;
\HWsp_{\multii}^{(0)} \to \HWsp_{\multii}^{(0)}.
\end{align}

Analogously to the map~\eqref{eq: Smap}, we denote the composition of the maps
defined in Section~\ref{subsec: combinatorial maps} by
\begin{align}\label{eq: combinatorial Smap}
\Scomb = \Scomb^{(s)}  \;\colon\; \LP^{(0)}_{(\multii,s)} 
\to \LP^{(0)}_{(s,\multii)},\qquad
\Scomb := \Rminuscomb^{-1} \circ \Rpluscomb.
\end{align}
When $\multii = (s_1,\ldots,s_p)$, the link pattern $\Scomb(\linkpatt)$ is 
obtained from $\linkpatt\in\LP^{(0)}_{(\multii,s)}$ 
by moving the rightmost index $p+1$ of $\linkpatt$ (with valence $s$)
to the left of all others, and relabeling the indices from left to right 
by $1,2,3,\ldots,p+1$. This is illustrated in Figure~\ref{fig: cyclic permutation}.

\begin{figure}
\includegraphics[scale=.75]{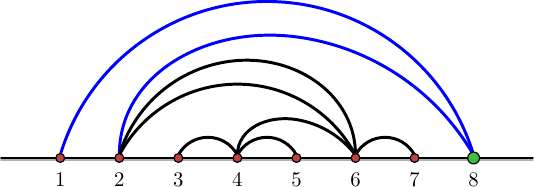}
\quad \; \raisebox{1em}{\Large $\mapsto$}  \quad \;
\includegraphics[scale=.75]{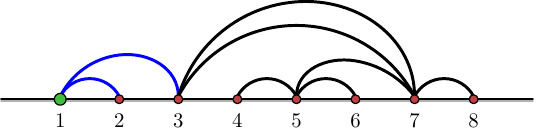}
\caption{\label{fig: cyclic permutation}
In the cyclic permutation $\Scomb = \Rminuscomb^{-1} \circ \Rpluscomb$, 
the rightmost point $p+1 = 8$ 
is moved to the left of all others, and the points are relabeled 
by $1,2,\ldots,8$.
}
\end{figure}

\subsection{\label{subsec: cyclic permutations}Cyclic permutations of tensor components}

In Section~\ref{subsec: construction}, we constructed the vectors
$\Puregeom_{\linkpatt}$ using the map $R_+$, 
see~\eqref{eq: basis vector construction}.
It follows from Proposition~\ref{prop: uniqueness} below that the construction 
could have been established as well using the map $R_-$ instead, only 
changing the normalization~\eqref{eq: normalization of basis vectors} of $\Puregeom_{\linkpatt}$.

Let $N$ be chosen as in~\eqref{eq: particular choice of N} and recall the map $\linkpatt \mapsto \alpha(\linkpatt)$ 
from Section~\ref{subsec: combinatorial maps}, illustrated in Figure~\ref{fig: open up map}.
For convenience, we denote the $s$:th iterate of the map
$\Scomb^{(1)} \colon \PP_N \to \PP_N$ by
$\Scomb^{\circ s} = \Scomb^{(1)} \circ \cdots \circ \Scomb^{(1)}$.
We then define, for any $\linkpatt \in \LP^{(s)}_{\multii}$,
the following vectors 
(compare with Equation~\eqref{eq: basis vector construction}):
\[ \Puregeom_{\linkpatt}' := R_-^{(s)}
\left(\Projectionhat^{(s,\multii)}(\Puregeomtwodim_{\alpha'(\linkpatt)})\right), \qquad \qquad
\text{where} \qquad \qquad 
\alpha'(\linkpatt) = \Scomb^{\circ s}\big(\alpha(\linkpatt)\big) .\]

\begin{prop}\label{prop: relation of maps R+ and R-}
We have
$\Puregeom_{\linkpatt}' = (-q)^{s} \, \Puregeom_{\linkpatt}$,
for all $\linkpatt \in \LP^{(s)}$.
\end{prop}
\begin{proof}
For any $\linkpatt \in \LP_\multii^{(s)}$,
the vector $\Puregeom_{\linkpatt}'$ belongs to the space 
$\HWsp_\multii^{(s)}$ by construction. Also, similarly as in 
the proof of Proposition~\ref{prop: commutative diagram}, we see that 
the collection 
$\left(\Puregeom_{\linkpatt}'\right)_{\linkpatt\in\LP^{(s)}}$, 
satisfies the equations~\eqref{eq: projection conditions}.
Therefore, $\Puregeom_{\linkpatt}'$ satisfy 
the system~\eqref{eq: cartan eigenvalue}~--~\eqref{eq: projection conditions}
of equations, and it follows from Proposition~\ref{prop: uniqueness} that,
for all partitions $\partition$ of $s$, there are constants 
$c_\partition \in \bC\setminus\{0\}$ such that we have
$\Bdryvec_{\linkpatt}' = c_\partition \Bdryvec_{\linkpatt}$, 
for all $\linkpatt \in \LP^{(s)}(\partition)$.
We evaluate the constants $c_\partition$ by studying the pattern 
$\linkpatt = \defpatt_\partition$ consisting of defects only.

By Lemma~\ref{lem: images of rainbow vectors}, 
for any partition $\partition = (s_1,\ldots,s_{|\partition|})$ of $s$,
we have
\[ \alpha(\defpatt_\partition) = \nested_{s} 
= \Scomb^{\circ s}(\nested_{s}) 
= \Scomb^{\circ s}(\alpha(\defpatt_\partition)) 
=  \alpha'(\defpatt_\partition), \]
where we used the observation that the link pattern $ \nested_{s}$
is invariant under the map $\Scomb^{\circ s}$.
Similarly as in the proof of Lemma~\ref{lem: images of rainbow vectors},
using the formula~\eqref{eq: solution for nested} for the vector 
$\Puregeomtwodim_{\nested_{s}}$, we calculate the action of
the map
$\Projectionhat^{(s,\partition)} 
= \Projectionhat^{\partition} \tens \Projectionhat^{(s)}
= (\Projectionhat^{(s_{|\partition|})}\tens\cdots\tens\Projectionhat^{(s_1)})
\tens \Projectionhat^{(s)}$,
\begin{align*}
\Projectionhat^{(s,\partition)}(\Puregeomtwodim_{\nested_{s}}) = \; &
\frac{1}{(q^{-2}-1)^{s}}\frac{\qnum{2}^{s}}{\qfact{s+1}}\,
\sum_{l=0}^{s}(-1)^lq^{l(s-l-1)}\times
\left(\Projectionhat^{\partition}(\MTbas_l^{(s)}) \tens 
\Projectionhat^{(s)}(\MTbas_{s-l}^{(s)})\right) \\
= \; & \frac{1}{(q^{-2}-1)^{s}}\frac{\qnum{2}^{s}}{\qfact{s+1}}\,
\sum_{l=0}^{s}(-1)^lq^{l(s-l-1)}\times
\left(F^{l}.(\Wbas_0^{(s_{|\partition|}+1)}\tens\cdots\tens\Wbas_0^{(s_1+1)}) \tens \Wbas_{s-l}^{(d)} \right) \\
= \; & \frac{1}{(q^{-2}-1)^{s}}\frac{\qnum{2}^{s}}{\qfact{s+1}}\,
\sum_{l=0}^{s}(-1)^{s-l} q^{(s-l)(l-1)}\times
\left(F^{s-l}.(\Wbas_0^{(s_{|\partition|}+1)}\tens\cdots\tens\Wbas_0^{(s_1+1)}) \tens \Wbas_{l}^{(d)} \right).
\end{align*}
It now follows from the definition of the map $R_-^{(s)}$ 
and Lemma~\ref{lem: images of rainbow vectors} that we have 
\begin{align*}
\Puregeom_{\defpatt_\partition}' = \; &  R_-^{(s)}
\left(\Projectionhat^{(s,\partition)}(\Puregeomtwodim_{\nested_{s}})\right)
= \frac{1}{(q^{-2}-1)^{s}}\frac{\qnum{2}^{s}}{\qfact{s+1}} \times (\Wbas_0^{(s_{|\partition|}+1)}\tens\cdots\tens\Wbas_0^{(s_1+1)})
= (-q)^{s} \, \Puregeom_{\defpatt_\partition},
\end{align*}
so $c_\partition = (-q)^s$, independently of the partition 
$\partition$. This concludes the proof.
\end{proof}

The above observation gives the cyclic permutation symmetry of the basis vectors
$\Puregeom_\linkpatt$ of the trivial subrepresentation $\HWsp_\multii^{(0)}$.
(Note that for $\HWsp_\multii^{(s)}$ with $s \geq 1$, the statement would not make sense.)

\begin{cor}\label{cor: cyclic symmetry}
The vectors 
$\Puregeom_{\linkpatt} \in \HWsp_\multii^{(0)}$ satisfy
\[\Puregeom_{\Scomb(\linkpatt)} 
= (-q)^{s_p} S^{(s_p)}(\Puregeom_{\linkpatt}) .\]
\end{cor}
\begin{proof}
By definition, we have 
$\Puregeom_\linkpatt = \Projectionhat^{(\multii)} (\Puregeomtwodim_{\alpha(\linkpatt)})$.
On the other hand, we have 
$\Puregeom_{\Scomb(\linkpatt)} = \Projectionhat^{(\multii')} (\Puregeomtwodim_{\alpha'(\linkpatt)})$,
where $\multii' = (s_p,s_1,\ldots,s_{p-1})$ and 
$\alpha'(\linkpatt) = \Scomb^{\circ s_p}\big(\alpha(\linkpatt)\big)$. 
Proposition~\ref{prop: relation of maps R+ and R-} now gives
\begin{align*}
R_-^{(s_p)} \left( \Puregeom_{\Scomb(\linkpatt)} \right)
= R_-^{(s_p)} \left( \Projectionhat^{(\multii')} 
(\Puregeomtwodim_{\alpha'(\linkpatt)}) \right)
= (-q)^{s_p} R_+^{(s_p)}  \left( \Projectionhat^{(\multii)} 
(\Puregeomtwodim_{\alpha(\linkpatt)}) \right)
= (-q)^{s_p} R_+^{(s_p)}  \left( \Puregeom_\linkpatt \right).
\end{align*}
Applying the map $(R_-^{(s_p)})^{-1}$ to both sides and using 
the definition~\eqref{eq: Smap} of the map $S$ we get
\begin{align*}
\Puregeom_{\Scomb(\linkpatt)}
= (-q)^{s_p} 
\left( (R_-^{(s_p)})^{-1} \circ R_+^{(s_p)} \right)
\left( \Puregeom_\linkpatt \right)
= (-q)^{s_p} S^{(s_p)} \left( \Puregeom_\linkpatt \right).
\end{align*}
\end{proof}

\begin{cor}\label{cor: iterate of Smap is constant multiple of identity}
The composed map~\eqref{eq: iteration of Smap}
is a constant multiple of the identity: we have
\[ S^{(s_1)}\circ S^{(s_2)}\circ\cdots\circ S^{(s_{p-1})}\circ S^{(s_p)} 
= \left(\prod_{i=1}^p (-q)^{-s_i}\right)\times \id_{\HWsp_\multii^{(0)}} .\]
\end{cor}
\begin{proof}
By Theorem~\ref{thm: highest weight vector space basis vectors},
the vectors $\Puregeom_\linkpatt$ with $\linkpatt \in \LP_\multii^{(0)}$
form a basis of the trivial subrepresentation $\HWsp_\multii^{(0)}$.
The assertion follows by iterating Corollary~\ref{cor: cyclic symmetry}
for each basis vector $\Puregeom_\linkpatt$ 
and noticing that we have
\[ \Scomb^{(s_1)}\circ \Scomb^{(s_2)}\circ\cdots\circ \Scomb^{(s_{p-1})}\circ \Scomb^{(s_p)} (\linkpatt) = \linkpatt .\]
\end{proof}

\bigskip{}
\section{\label{sec: basis functions}Solutions to the Benoit~\& Saint-Aubin  PDEs with particular asymptotics properties}

Now we construct solutions to partial differential equations 
of Benoit~\& Saint-Aubin type~\cite{BSA:Degenerate_CFTs_and_explicit_expressions_for_some_null_vectors}.
These PDEs have been well-known in CFT for many decades.
From statistical physics point of view, scaling limits of correlations in critical models 
have been observed to satisfy this type of PDEs, 
in, e.g.,~\cite{BPZ:Infinite_conformal_symmetry_of_critical_fluctuations_in_2D,
Cardy:Critical_percolation_in_finite_geometries, 
Watts:A_crossing_probability_for_critical_percolation_in_two_dimensions,
Bauer-Bernard:Conformal_field_theories_of_SLEs,
Gamsa-Cardy:The_scaling_limit_of_two_cluster_boundaries_in_critical_lattice_models}, 
with a few rigorous results now 
established too~\cite{Dubedat:Excursion_decomposition_for_SLE_and_Watts_crossing_formula,
Sheffield-Wilson:Schramms_proof_of_Watts_formula,
BJV:Some_remarks_on_SLE_bubbles_and_Schramms_2point_observable,
KKP:Conformal_blocks_pure_partition_functions_and_KW_binary_relation,
LV:Coulomb_gas_for_commuting_SLEs, LV:Coulomb_gas_integrals_PART2,
PW:Global_multiple_SLEs_and_pure_partition_functions}.
Solutions to 
such PDEs have also been associated with of random curves,
in, e.g.,~\cite{Friedrich-Werner:Conformal_restriction_highest_weight_representations_and_SLE,
Kontsevich:CFT_SLE_and_phase_boundaries,
Friedrich-Kalkkinen:On_CFT_and_SLE,
BBK:Multiple_SLEs_and_statistical_mechanics_martingales,
Dubedat:Commutation_relations_for_SLE,
Kontsevich-Suhov:On_Malliavin_measures_SLE_and_CFT,
Dubedat:SLE_and_Virasoro_representations_fusion,
Kytola-Peltola:Pure_partition_functions_of_multiple_SLEs,
PW:Global_multiple_SLEs_and_pure_partition_functions}.

The main result of this article is the construction 
of particular solutions to these PDEs, with specific asymptotic
boundary conditions, given in Theorem~\ref{thm: asymptotic properties of general basis vectors}.
Such asymptotics can be thought of as specifying the fusion channels
if the solutions are thought of as CFT correlation functions. 
In terms of random curves, this corresponds to coalescing the 
starting points of the curves.

As the main tool in our construction, we use the quantum group method 
developed in the article~\cite{Kytola-Peltola:Conformally_covariant_boundary_correlation_functions_with_quantum_group}
and summarized in Theorem~\ref{thm: SCCG correspondence},
together with the results obtained in Section~\ref{sec: basis vectors}. 
The basis functions $\BasisF_\linkpatt$ of our main Theorem~\ref{thm: asymptotic properties of general basis vectors}
are constructed from the vectors $\Puregeom_\linkpatt$ of 
Theorem~\ref{thm: highest weight vector space basis vectors} as
\[ \BasisF_\linkpatt (x_1, \ldots, x_p) = \sF[\Puregeom_\linkpatt] (x_1, \ldots, x_p) , \] 
where $\sF$ denotes a map from the highest weight vector space
$\HWsp_\multii^{(s)}$ to the solution space of the PDEs.
In this map, the projection 
properties~\eqref{eq: projection conditions} of the vectors 
$\Puregeom_\linkpatt$ correspond with the required asymptotics properties of 
the basis functions, as stated explicitly in Theorems~\ref{thm: SCCG correspondence} 
and~\ref{thm: asymptotic properties of general basis vectors}.

In Lemma~\ref{lem: limits at same rate} and Propositions~\ref{prop: strong limits} and~\ref{prop: cyclicity of limits},
we prove additional properties of the basis functions $\BasisF_\linkpatt$,
concerning asymptotics when taking several variables together simultaneously,
or taking a variable to infinity. 
These properties are needed in further applications,
e.g., in Section~\ref{sec: multiple SLEs}
and~\cite{Flores-Peltola:Solution_space_of_BSA_PDEs}.

\subsection{\label{subsec: SCCG correspondence}Solutions to the Benoit~\& Saint-Aubin PDEs}

Fix a parameter $\kappa > 0$. 
Given a multiindex $\multii = (s_1,\ldots,s_p) \in \bZpos^p$, we use the notations
of~\eqref{eq: s in terms of d}~and~\eqref{eq: definition of n and s}
throughout. We also denote
\begin{align*}
h_{1,d} :=\; & \frac{(d-1)(2(d+1)-\kappa)}{2\kappa} \qquad \qquad \text{ and }
\qquad \qquad 
\Delta^{d_1,\ldots,d_p}_d := h_{1,d}-\sum_{i=1}^{p} h_{1,d_{i}}.
\end{align*}
For fixed $j \in \{1,2,\ldots,p\}$, 
the Benoit~\& Saint-Aubin  partial differential operators
\begin{align}\label{eq: BSA differential operator}
\sD_{d_{j}}^{(j)} 
:= \; & \sum_{k=1}^{d_{j}}\sum_{\substack{n_{1},\ldots,n_{k}\geq1\\
n_{1}+\ldots+n_{k}=d_{j}}}
\frac{(-4/\kappa)^{d_{j}-k}\,(d_{j}-1)!^{2}}{\prod_{l=1}^{k-1}(\sum_{i=1}^{l}n_{i})(\sum_{i=l+1}^{k}n_{i})}\times\sL_{-n_{1}}^{(j)}\cdots
\sL_{-n_{k}}^{(j)},
\end{align}
homogeneous of order $d_j$, 
are defined in terms of the first order differential operators\footnote{
The operators $\sL_{m}^{(j)} $ are related to the generators of the Virasoro algebra~\cite{DMS:CFT,
BDIZ:Covariant_differential_equations_and_singular_vectors_in_Virasoro_representations},
and the formulas~\eqref{eq: BSA differential operator} are obtained from 
the similar formulas for singular vectors in representations of 
the Virasoro algebra found by L.~Benoit and Y.~Saint-Aubin in
\cite{BSA:Degenerate_CFTs_and_explicit_expressions_for_some_null_vectors}.
}
\begin{align*}
\sL_{m}^{(j)} 
:= \; & -\sum_{i\neq j}\left((x_{i}-x_{j})^{1+m}\pder{x_{i}}+(1+m)\, h_{1,d_{i}}\,(x_{i}-x_{j})^{m}\right).
\end{align*}

We are interested in solutions $\BasisF \colon \chamber_p \to \bC$
to the PDE system
\begin{align}\label{eq: BSA differential equations}
\tag{PDE}
\sD_{d_j}^{(j)} \BasisF(x_1,\ldots,x_p) = 0 , \qquad \text{ for all } j \in \{1,2,\ldots,p\},
\end{align}
defined on the chamber domain
\begin{align}\label{eq: chamber}
\chamber_p := \; & 
\big\{(x_{1},\ldots,x_p)\in\bR^{p}\;\big|\; x_{1}<\cdots<x_p\big\}.
\end{align}

We very briefly summarize the method of~\cite{Kytola-Peltola:Conformally_covariant_boundary_correlation_functions_with_quantum_group} 
for constructing solutions to the Benoit~\& Saint-Aubin PDE 
systems~\eqref{eq: BSA differential equations}. For details about this method, we refer to Sections~3 and~4 in the 
article~\cite{Kytola-Peltola:Conformally_covariant_boundary_correlation_functions_with_quantum_group}.
The idea is to construct solutions in terms of Dotsenko-Fateev (Feigin-Fuchs) integrals
\cite{Dotsenko-Fateev:Conformal_algebra_and_multipoint_correlation_functions_in_2D_statistical_models},
which appear in the Coulomb gas formalism of CFT. The solutions are of the form
\begin{align}\label{eq: ansatz for solution to PDEs}
F(\boldsymbol{x}) = \int_\Gamma f^{(\ell)}(\boldsymbol{x};\boldsymbol{w}) \ud w_1 \cdots \ud w_\ell,
\end{align}
with $\boldsymbol{w} = (w_1, \ldots, w_\ell)$, defined for $\boldsymbol{x} = (x_1,\ldots,x_p) \in \chamber_p$ as follows.
First, 
the integrand $f^{(\ell)}(\boldsymbol{x};\boldsymbol{w})$ is a branch of the following multivalued function, 
a product of powers of differences,
\begin{align} \label{eq: integrand}
f^{(\ell)}(\boldsymbol{x};\boldsymbol{w}) 
=\; & \prod_{1\leq i<j\leq p}(x_{j}-x_{i})^{\frac{2}{\kappa} s_i s_j}
\prod_{\substack{1\leq i\leq p\\1\leq r\leq\ell}}
(w_{r}-x_{i})^{-\frac{4}{\kappa} s_i}
\prod_{1\leq r<s\leq\ell}(w_{s}-w_{r})^{\frac{8}{\kappa}},
\end{align}
with parameters $s_i\in\bZnn$, for $i=1,\ldots,p$, and $\kappa > 0$, and $\ell \in \bZnn$.
Second, the integration contours $\Gamma$ are closed
$\ell$-surfaces which can be written as linear combinations of surfaces
corresponding to the natural basis 
$\{ \Wbas_{l_p}^{(d_p)} \tens \cdots \tens  \Wbas_{l_{2}}^{(d_{2})} \tens \Wbas_{l_1}^{(d_1)} \}$
of the tensor product representation~\eqref{eq: order of tensorands}
of the quantum group $\Uqsltwo$, with dimensions $d_i$ of the tensorands related to the parameters $s_i$ as in~\eqref{eq: s in terms of d},
and with $\ell = \sum_{i=1}^p l_i$.
For the detailed relation, see Figure~\ref{fig: basis} and 
\cite[Sections~3.3 and~4.1]{Kytola-Peltola:Conformally_covariant_boundary_correlation_functions_with_quantum_group}. 
In the figure, an auxiliary point $x_0$ appears; 
however by \cite[Proposition~4.5]{Kytola-Peltola:Conformally_covariant_boundary_correlation_functions_with_quantum_group}, 
the functions $\BasisF_\linkpatt$ constructed in this article do not depend on $x_0$.

The relation of vectors in the tensor product representation~\eqref{eq: order of tensorands} and functions of 
type~\eqref{eq: ansatz for solution to PDEs} is called in~\cite{Kytola-Peltola:Conformally_covariant_boundary_correlation_functions_with_quantum_group} 
``the spin chain -- Coulomb gas correspondence'' $\sF$. We state its main features in Theorem~\ref{thm: SCCG correspondence} below. 
We restrict our attention to the space $\HWsp_\multii^{(s)}$ of highest weight vectors, 
because these are the vectors that yield solutions to~\eqref{eq: BSA differential equations}. 
We will prove in~\cite{Flores-Peltola:Solution_space_of_BSA_PDEs} that $\sF$ is in fact injective on $\HWsp_\multii^{(s)}$
when $q$ is not a root of unity.

\begin{thm}{\cite[Theorem~4.17]{Kytola-Peltola:Conformally_covariant_boundary_correlation_functions_with_quantum_group}}
\label{thm: SCCG correspondence}
Let $\kappa \in (0,\infty) \setminus \bQ$ and $q = e^{\ii \pi 4 / \kappa}$, and  
$s \in \bZnn$. There exist linear maps 
$\sF \colon \HWsp_\multii^{(s)} \to \sC^\infty(\chamber_{p})$,
for all $\multii \in \bZpos^p$, such that the following holds for any 
$v\in \HWsp_\multii^{(s)}$.

\begin{description}
\item [{(PDE)}] The function $\sF[v] \colon \chamber_{p} \to \bC$ satisfies 
the system~\eqref{eq: BSA differential equations} of partial differential equations.
\item [{(COV)}] The function $\sF[v]$ is translation invariant and 
homogeneous of degree $\Delta = \Delta^{d_1,\ldots,d_p}_d:$
\begin{align}\label{eq: scaling covariance}
\sF[v](\lambda x_1 + \xi,\ldots,\lambda x_p + \xi)
= \lambda^{\Delta} \times \sF[v](x_1,\ldots,x_p) ,
\end{align}
for any $\xi \in \bR$ and $\lambda \in \bRpos$.
Moreover, if $s = 0$, then $\sF[v]$ satisfies the following 
covariance property under any M\"obius transformation $\Mob\colon\bH\to\bH$ 
such that $\Mob(x_{1})<\Mob(x_{2})<\cdots<\Mob(x_{p}):$
\begin{align}\label{eq: Mobius covariance}
\sF[v](x_{1},\ldots, x_p) = 
\prod_{i=1}^{p}\Mob'(x_{i})^{h_{1,d_{i}}} \times 
\sF[v] \left(\Mob(x_{1}),\ldots,\Mob(x_p)\right) .
\end{align}
\item [{(ASY)}] 
Let $j \in \{1,2,\ldots,p-1\}$ and
$m = \frac{1}{2}\left(d_{j}+d_{j+1}-\projdmn-1\right) 
\in \{0,1,\ldots,\min(d_{j},d_{j+1})-1 \}$, 
and suppose that we have $\pi_{j}^{(\projdmn)}(v) = v$.
Then, the function $\sF[v] \colon \chamber_{p} \to \bC$ has the asymptotics
\begin{align*}
\lim_{x_j,x_{j+1}\to\xi}
\frac{\sF[v](x_1,\ldots,x_{p})}{|x_{j+1}-x_j|^{\Delta^{d_j,d_{j+1}}_{\projdmn}}} 
= \; & B^{d_j,d_{j+1}}_{\projdmn} \times 
\sF[\hat{\pi}_{j}^{(\projdmn)}(v)]
(x_1,\ldots,x_{j-1},\xi,x_{j+2}\ldots,x_{p}),
\end{align*}
for any $\xi \in (x_{j-1},x_{j+2})$, where $\sF[1] = 1$ in the case $p=2$, and
the multiplicative constant is 
\begin{align}\label{eq: Selberg integral}
B^{d_j,d_{j+1}}_{\projdmn} = \; & \frac{1}{m!} \; \prod_{u=1}^m
    \frac{\Gamma\big( 1 - \frac{4}{\kappa}(d_j-u)\big) \; \Gamma\big( 1 - \frac{4}{\kappa}(d_{j+1}-u)\big) \; \Gamma\big( 1 + \frac{4}{\kappa}u \big)}
        {\Gamma\big( 1 + \frac{4}{\kappa}\big) \; \Gamma\big( 2 - \frac{4}{\kappa}(d_j+d_{j+1}-m-u)\big)}.
\end{align}
\end{description}
\end{thm}

We record an explicit formula of a special case.

\begin{lem}\label{lem: basis functions with only defects}
For any partition $\partition = (s_1,\ldots,s_{|\partition|})$ 
of $s \in \bZnn$, the image of the vector 
$\Puregeom_{\defpatt_\partition} \in \HWsp_\partition^{(s)}$
has the explicit formula 
\begin{align*}
\sF[\Puregeom_{\defpatt_\partition}] (x_1,\ldots,x_{|\partition|})
= \; & \frac{1}{(q-q^{-1})^{s}}
\frac{\qnum{2}^{s}}{\qfact{s+1}}
\times \prod_{1 \leq i < j \leq |\partition|} (x_j-x_i)^{\frac{2}{\kappa}s_is_j}.
\end{align*}
\end{lem}
\begin{proof}
The assertion follows immediately from the definition of the correspondence map $\sF$ given in 
\cite[Section~4.1]{Kytola-Peltola:Conformally_covariant_boundary_correlation_functions_with_quantum_group}
and the formula~\eqref{eq: normalization of basis vectors} of 
$\Puregeom_{\defpatt_\partition}$.
\end{proof}

Images of more general vectors $v \in \HWsp_\multii^{(s)}$ 
under the map $\sF$ have a similar form, but we need to integrate over $\ell$
so-called screening variables, as in~\eqref{eq: ansatz for solution to PDEs}, where $\ell$ is the number of links in 
the link patterns in $\LP_\multii^{(s)}$, 
as in~\eqref{eq: definition of n and s}. The integration 
$\ell$-surface is determined by the vector $v$ as explained in 
the article~\cite{Kytola-Peltola:Conformally_covariant_boundary_correlation_functions_with_quantum_group}.

\begin{figure}
\includegraphics[scale=.28]{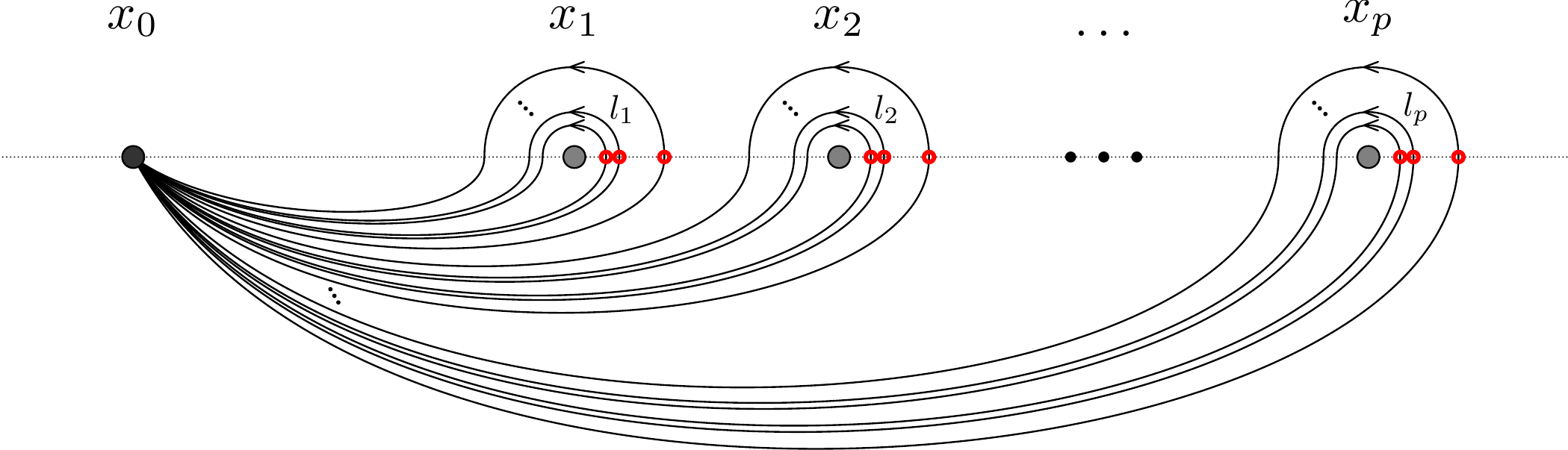}
\bigskip

\caption{\label{fig: basis}
Illustration of the integration surface for a ``basis integral function'' $\varphi^{(x_0)}_{l_1,l_2,\ldots,l_p}(x_1, x_2, \ldots,x_p)$, that is,
a Coulomb gas integral of type~\eqref{eq: ansatz for solution to PDEs} with $\Gamma$ the surface depicted in the figure.
The red circles indicate a choice of branch for the integrand in~\eqref{eq: ansatz for solution to PDEs},
so that it is real and positive when the integration variables lie at those points,
see~\cite{Kytola-Peltola:Conformally_covariant_boundary_correlation_functions_with_quantum_group} for details.
These integrals are the images under the spin chain~--~Coulomb gas correspondence map $\sF$
of the natural basis vectors
$\Wbas_{l_p}^{(d_p)} \tens \cdots \tens  \Wbas_{l_{2}}^{(d_{2})} \tens \Wbas_{l_1}^{(d_1)}$
of the tensor product representation~\eqref{eq: order of tensorands}, with $\ell = \sum_{i=1}^p l_i$.
}
\end{figure}

\subsection{\label{subsec: basis functions}The solutions with particular asymptotics}

By Theorem~\ref{thm: SCCG correspondence}, the projection 
properties~\eqref{eq: projection conditions} of the vectors 
$\Puregeom_\linkpatt$ of Theorem~\ref{thm: highest weight vector space basis vectors}
give explicit asymptotic behavior 
for the solutions $\BasisF_\linkpatt = \sF[\Puregeom_\linkpatt]$
when two variables $x_j,x_{j+1}$ tend to a common limit.
Furthermore, in Proposition~\ref{prop: strong limits} in the next section,
we establish similar recursive asymptotics
when taking many variables to a common limit.

Recall that, for a link pattern $\linkpatt$, we denote by 
$\ell_{j,j+1} = \ell_{j,j+1}(\linkpatt)$ the multiplicity of the link 
$\link{j}{j+1}$ in $\linkpatt$.

\begin{thm}\label{thm: asymptotic properties of general basis vectors}
Let $\kappa \in (0,8) \setminus \bQ$. The functions
$\BasisF_\linkpatt = \sF[\Puregeom_\linkpatt] \colon \chamber_{p} \to \bC$ 
have the following properties.
\begin{description}
\item[(1)] For any $\linkpatt \in \LP_\multii^{(s)}$, the function
$\BasisF_\linkpatt$ satisfies the Benoit $\&$ Saint-Aubin
PDE system~\eqref{eq: BSA differential equations}.

\item[(2)] The function $\BasisF_\linkpatt$ is translation invariant and homogeneous 
as in~\eqref{eq: scaling covariance}, with $d = s + 1$. 

\item[(3)] If $s = 0$, then $\BasisF_\linkpatt$ satisfies the full M\"obius covariance \eqref{eq: Mobius covariance}.

\item[(4)] For any $j \in \{1,2,\ldots,p-1 \}$ and
$m = \frac{1}{2}\left(s_{j}+s_{j+1}-\projdmn+1\right) 
\in \{0,1,\ldots,\min(s_{j},s_{j+1}) \}$, 
and $\xi \in (x_{j-1},x_{j+2})$,
the function $\BasisF_\linkpatt$ has the asymptotics
\begin{align}\label{eq: asymptotic properties}
\lim_{x_j,x_{j+1}\to\xi}
\frac{\BasisF_\linkpatt(x_1,\ldots,x_{p})}{|x_{j+1}-x_j|^{\Delta^{d_j,d_{j+1}}_{\projdmn}}}
=\; & \begin{cases}
0\quad 
& \mbox{if } \ell_{j,j+1} < m\\
\frac{B^{d_j,d_{j+1}}_{\projdmn}}{\constantfromdiagram{m}{s_j}{s_{j+1}}} \times 
\BasisF_{\hat{\linkpatt}}(x_1,\ldots,x_{j-1},\xi,x_{j+2}\ldots,x_{p})
& \mbox{if } \ell_{j,j+1} = m,\\
\end{cases}
\end{align}
where $\hat{\linkpatt} = \linkpatt\removeLink (m\times\link{j}{j+1})$,
and the constants $B^{d_j,d_{j+1}}_{\projdmn}$, 
given in Equation~\eqref{eq: Selberg integral},
and $\constantfromdiagram{m}{s_j}{s_{j+1}}$, given 
in Equation~\eqref{eq: constant from diagram}, are non-zero.
\end{description}
\end{thm}
\begin{proof}
Assertions (1)~--~(3) follow immediately from 
the properties \eqref{eq: cartan eigenvalue}~--~\eqref{eq: highest weight vector} 
of the basis vectors $\Puregeom_\linkpatt$ of 
Theorem~\ref{thm: highest weight vector space basis vectors}, and
(PDE) and (COV) parts of Theorem~\ref{thm: SCCG correspondence}.
To prove assertion (4), we first note that, when $\kappa \in (0,8)$, 
then the exponents in the property
(ASY) in Theorem~\ref{thm: SCCG correspondence} satisfy
$\Delta_{\projdmn}^{d_j,d_{j+1}} < \Delta_{\projdmn'}^{d_j,d_{j+1}}$,
for any $2 \leq \projdmn < \projdmn'$,
and for $\projdmn = 1$ and $\projdmn' \geq 3$, we also have
$\Delta_{\projdmn'}^{d_j,d_{j+1}} - \Delta_{1}^{d_j,d_{j+1}} > 0$. 
Because in~\eqref{eq: asymptotic properties}, 
$\projdmn$ increases in steps of two, we conclude that  
assertion (4) follows from the properties~\eqref{eq: projection conditions} 
of the basis vectors $\Puregeom_\linkpatt$ of 
Theorem~\ref{thm: highest weight vector space basis vectors} and
the (ASY) part of Theorem~\ref{thm: SCCG correspondence}.
This concludes the proof.
\end{proof}

%
%
%
%

We remark that in the above theorem, the range of the parameter $\kappa$ is restricted 
to $(0,8) \setminus \bQ$. The restriction to the interval $(0,8]$ is necessary in order to establish the asymptotics property (4).
Indeed, when $\kappa > 8$, the mutual order of the exponents in the formula~\eqref{eq: asymptotic properties}
may change, resulting in the leading powers in the asymptotics to change.
On the other hand, we expect the statement of Theorem~\ref{thm: asymptotic properties of general basis vectors} 
to be morally true also when $\kappa \in (0,8) \cap \bQ$: functions $\BasisF_\linkpatt$ with properties (1)~--~(4) should still exist.
In principle, the functions of Theorem~\ref{thm: asymptotic properties of general basis vectors} 
can be analytically continued 
to rational values of $\kappa$ --- to do this systematically, further care would be needed. 

\begin{cor}
The functions $\BasisF_\linkpatt$ are not identically zero.
\end{cor}
\begin{proof}
This follows from 
Theorem~\ref{thm: asymptotic properties of general basis vectors}
by induction on the number $n = \sum_{i=1}^p s_i =: |\multii|$ for
the link pattern $\linkpatt \in \LP_\multii^{(s)}$ with
$\multii = (s_1,\ldots,s_p)$. 
By Lemma~\ref{lem: basis functions with only defects}, the base case is immediate, 
as $\BasisF_\emptyset = \sF[\Puregeom_\emptyset] = \sF[1] = 1$. 
Fix $s \in \bZnn$ and assume that no function $\BasisF_\tau$ with 
$\tau \in \LP_\varrho^{(s)}$ and $|\varrho| \leq n$ is
identically zero. Consider a function $\BasisF_\linkpatt$ with
$\linkpatt \in \LP_\multii^{(s)}$, and $|\multii| = n+1$.
First, if $\linkpatt = \defpatt_\partition$ only consists of defects, then
the function $\BasisF_\linkpatt = \BasisF_{\defpatt_\partition}$ 
is not identically zero, by the explicit 
formula in Lemma~\ref{lem: basis functions with only defects}.
On the other hand, if $\linkpatt$ contains links, then there is an innermost link
$\link{j}{j+1} \in \linkpatt$. Applying the asymptotics property 
\eqref{eq: asymptotic properties} with $m = \ell_{j,j+1}$, the induction 
hypothesis shows that $\BasisF_\linkpatt$ cannot be identically zero.
\end{proof}

\subsection{\label{subsec: further limit properties}Limits when collapsing several variables}

\begin{figure}
\includegraphics[scale=.75]{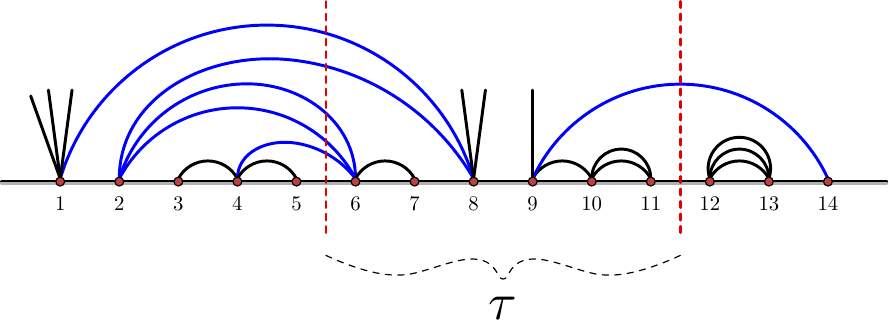}

\bigskip
\bigskip
\bigskip

\raisebox{1em}{\includegraphics[scale=.75]{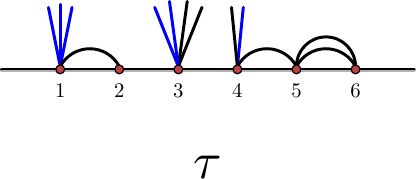}} \qquad \qquad  \qquad 
\includegraphics[scale=.75]{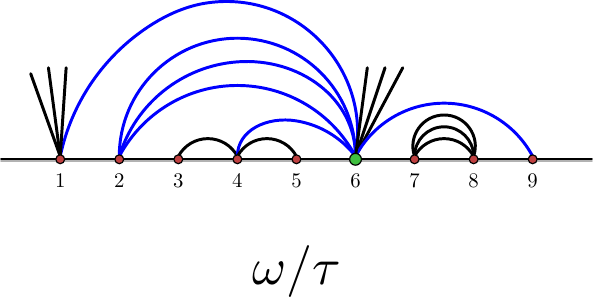}
\caption{\label{fig: removal of sub pattern}
The sub-link pattern $\tau$ of a link pattern $\linkpatt$, and
the link pattern $\linkpatt / \tau$, obtained from $\linkpatt$ by
removing the links of $\tau$ and collapsing the indices involved in $\tau$ into 
one point, colored in green. The lines which are colored in blue are common in
both $\tau$ and $\linkpatt / \tau$ --- note that in $\tau$, these links are 
cut when separating $\tau$ from $\linkpatt$ and they become defects, whereas 
$\linkpatt / \tau$ has equally many defects as $\linkpatt$, 
and the blue links remain.}
\end{figure}

We now consider the limit of the function $\BasisF_\linkpatt$ as several of its variables tend to a common limit simultaneously.
For this, we need some notation.

Fix a link pattern
\begin{align*}
\linkpatt=\Big\{\linkInEquation{a_1}{b_1},\ldots,\linkInEquation{a_\ell}{b_\ell}\Big\}
\bigcup
\Big\{\defectInEquation{c_1},\ldots,\defectInEquation{c_{s}}\Big\} \,
\in \, \LP^{(s)}_{\multii}
\end{align*}
and indices $1 \leq j < k \leq p$, and denote by 
\begin{align*}
\mathfrak{A}_{j,k} := \; & \set{a_i \;|\; a_i \in \set{j,j+1,\ldots,k} 
\text{ and } \, b_i \notin \set{j,j+1,\ldots,k} } \subset \set{a_1,a_2,\ldots,a_\ell}, \\
\mathfrak{B}_{j,k} := \; & \set{b_i \;|\; a_i \notin \set{j,j+1,\ldots,k} 
\text{ and } \, b_i \in \set{j,j+1,\ldots,k} } \subset \set{b_1,b_2,\ldots,b_\ell}, \\
\mathfrak{C}_{j,k} := \; & \set{c_i \;|\; c_i \in \set{j,j+1,\ldots,k}} 
\subset \set{c_1,c_2,\ldots,c_s} ,
\end{align*}
and $r = \# \mathfrak{A}_{j,k} + \# \mathfrak{B}_{j,k} + \# \mathfrak{C}_{j,k}$,
and let $\tau \in \LP_{\multii_{j,k}}^{(r)}$ be the sub-link pattern of 
$\linkpatt$ with index valences $\multii_{j,k} = (s_j,s_{j+1}\ldots,s_k)$, 
consisting of the lines of $\linkpatt$ attached to the indices $j,j+1,\ldots,k$, 
that is,
\begin{align*}
\tau = \tau_{j,k}(\linkpatt) =
\Big\{\linkInEquation{a_i}{b_i} \;\Big|\; a_i,b_i \in \set{j,j+1,\ldots,k} \Big\}
\bigcup
\Big\{\defectInEquation{c} \;\Big|\; 
c \in \mathfrak{A}_{j,k} \cup \mathfrak{B}_{j,k} \cup \mathfrak{C}_{j,k} \Big\} \,
\in \, \LP_{\multii_{j,k}}^{(r)},
\end{align*}
see Figure~\ref{fig: removal of sub pattern}. Also, denote by 
$\linkpatt \removeLink \tau$ the link pattern obtained from $\linkpatt$
by ``removing $\tau$'', that is, removing from $\linkpatt$ the links 
$\link{a\;}{\;b\,}$ with indices $a,b \in \{ j, j+1,\ldots, k \}$, collapsing
the indices $j,j+1,\ldots,k$ of $\linkpatt$ into one point, and relabeling
the indices thus obtained from left to right by $1,2,\ldots$,
as emphasized in Figure~\ref{fig: removal of sub pattern}.

The function $\BasisF_\linkpatt$ has the following limiting behavior. 
\begin{lem} \label{lem: limits at same rate}
Let $1 \leq j < k \leq p$ and $x_{j-1} < \xi < x_{k+1}$, and suppose that
\begin{align}\label{eq: limit at same rate}
\begin{split}
& x_j, x_{j+1}, \ldots, x_k \to \xi
    \qquad \quad \text{ in such a way that } \qquad 
    \frac{x_{i}-x_j}{x_k-x_j} \to \eta_{i} ,
    \qquad \text{ for } i \in \{j,j+1,\ldots,k \} , \\
& \text{with } 0=\eta_{j} < \eta_{j+1} < \cdots < \eta_{k-1} < \eta_{k} = 1 .
\end{split}
\end{align}
Denote $\Delta = \Delta_{\projdmn}^{d_j,\ldots,d_k}$, with $\projdmn = r+1$.
Then, in the limit~\eqref{eq: limit at same rate}, we have
\begin{align*}
\frac{\BasisF_\linkpatt(x_1,\ldots,x_p)}{|x_k - x_j|^{\Delta}}
\; \longrightarrow \; \; & \BasisF_\tau(\eta_{j}, \ldots, \eta_{k})
\times \BasisF_{\linkpatt \removeLink \tau}(x_1,\ldots,x_{j-1},\xi,x_{k+1},\ldots,x_p).
\end{align*}
\end{lem}
\begin{proof}
By Lemma~\ref{lem: details for proof of further limit properties}, 
for some constants $c_{l_1,\ldots,l_{j-1};l;l_{k+1},\ldots,l_p} \in \bC$,
we have
\begin{align} 
\label{eq: vector omega}
\Puregeom_\linkpatt
= \; & \sum_{l=0}^r \sum_{\substack{l_1,\ldots,l_{j-1},\\l_{k+1},\ldots,l_p}}
c_{l_1,\ldots,l_{j-1};l;l_{k+1},\ldots,l_p} \times 
\left(\Wbas_{l_p} \tens \cdots \tens \Wbas_{l_{k+1}}
\tens F^l.\Puregeom_{\tau} \tens \Wbas_{l_{j-1}}
\tens \cdots \tens \Wbas_{l_1}\right), \\
 \label{eq: vector omega minus tau}
\Puregeom_{\linkpatt \removeLink \tau} 
= \; & \sum_{l=0}^r \sum_{\substack{l_1,\ldots,l_{j-1},\\l_{k+1},\ldots,l_p}}
c_{l_1,\ldots,l_{j-1};l;l_{k+1},\ldots,l_p} \times
\left(\Wbas_{l_p} \tens \cdots \tens \Wbas_{l_{k+1}}
\tens \Wbas_l \tens \Wbas_{l_{j-1}}
\tens \cdots \tens \Wbas_{l_1}\right).
\end{align}
Therefore, by \cite[Proposition~5.1]{Kytola-Peltola:Conformally_covariant_boundary_correlation_functions_with_quantum_group},
the limit~\eqref{eq: limit at same rate} has the asserted form:
\begin{align*}
\frac{\sF[\Puregeom_\linkpatt](x_1,\ldots,x_p)}{|x_k - x_j|^{\Delta}}
\; \longrightarrow \; \; &
\sF[\Puregeom_\tau](\eta_{j}, \ldots, \eta_{k}) \times
\sF[\Puregeom_{\linkpatt \removeLink \tau}]
(x_1,\ldots,x_{j-1},\xi,x_{k+1},\ldots,x_p).
\end{align*}
\end{proof}

In the proof of Lemma~\ref{lem: limits at same rate}, we 
use~\cite[Proposition~5.1]{Kytola-Peltola:Conformally_covariant_boundary_correlation_functions_with_quantum_group}.
To prove the latter, the idea is to rearrange the integrations in the Coulomb gas type integral representation~\eqref{eq: ansatz for solution to PDEs} 
of the function $\BasisF_\linkpatt= \sF[\Puregeom_\linkpatt]$ in such a way that the collapsing variables 
$x_j,  \ldots, x_k$ are surrounded by nested contours --- see Figure~\ref{fig: limits}.
After this rearranging, also ``hypercube type'' integrations between the variables $x_j,  \ldots, x_k$ might occur.
Now, if the limit $x_j, \ldots,x_k \to \xi$ is taken as in~\eqref{eq: limit at same rate}, then the function~\eqref{eq: ansatz for solution to PDEs} 
with integration surface of Figure~\ref{fig: limits}(top) converges to a function of type~\eqref{eq: ansatz for solution to PDEs} with   
integration surface of Figure~\ref{fig: limits}(bottom) times a constant which depends on the convergence rate~\eqref{eq: limit at same rate}.
This multiplicative constant results in the constant $\mathscr{F}_\tau(\eta_j, \ldots,\eta_k)$ in 
Lemma~\ref{lem: limits at same rate}.
From its dependence on the convergence rate~\eqref{eq: limit at same rate},
we see that if the variables $x_j,  \ldots, x_k$ tend to $\xi$ in a different way, the limit can be different or fail to exist.

However, in some cases, no integrations between the collapsing variables $x_j,  \ldots, x_k$ occur, and then the limit
Lemma~\ref{lem: limits at same rate} in fact exists along any sequence 
$x_j,  \ldots, x_k \to \xi$ and not only along sequences of type~\eqref{eq: limit at same rate}. 
Indeed, if in Figure~\ref{fig: limits} there are no contours between $x_j,  \ldots, x_k$, but only around them, then
similarly as in the proof of 
\cite[Proposition~5.1]{Kytola-Peltola:Conformally_covariant_boundary_correlation_functions_with_quantum_group},
dominated convergence theorem allows us to collapse these variables inside the integration in~\eqref{eq: ansatz for solution to PDEs} 
along any sequence $x_j,  \ldots, x_k \to \xi$.

\begin{figure}
\includegraphics[scale=.28]{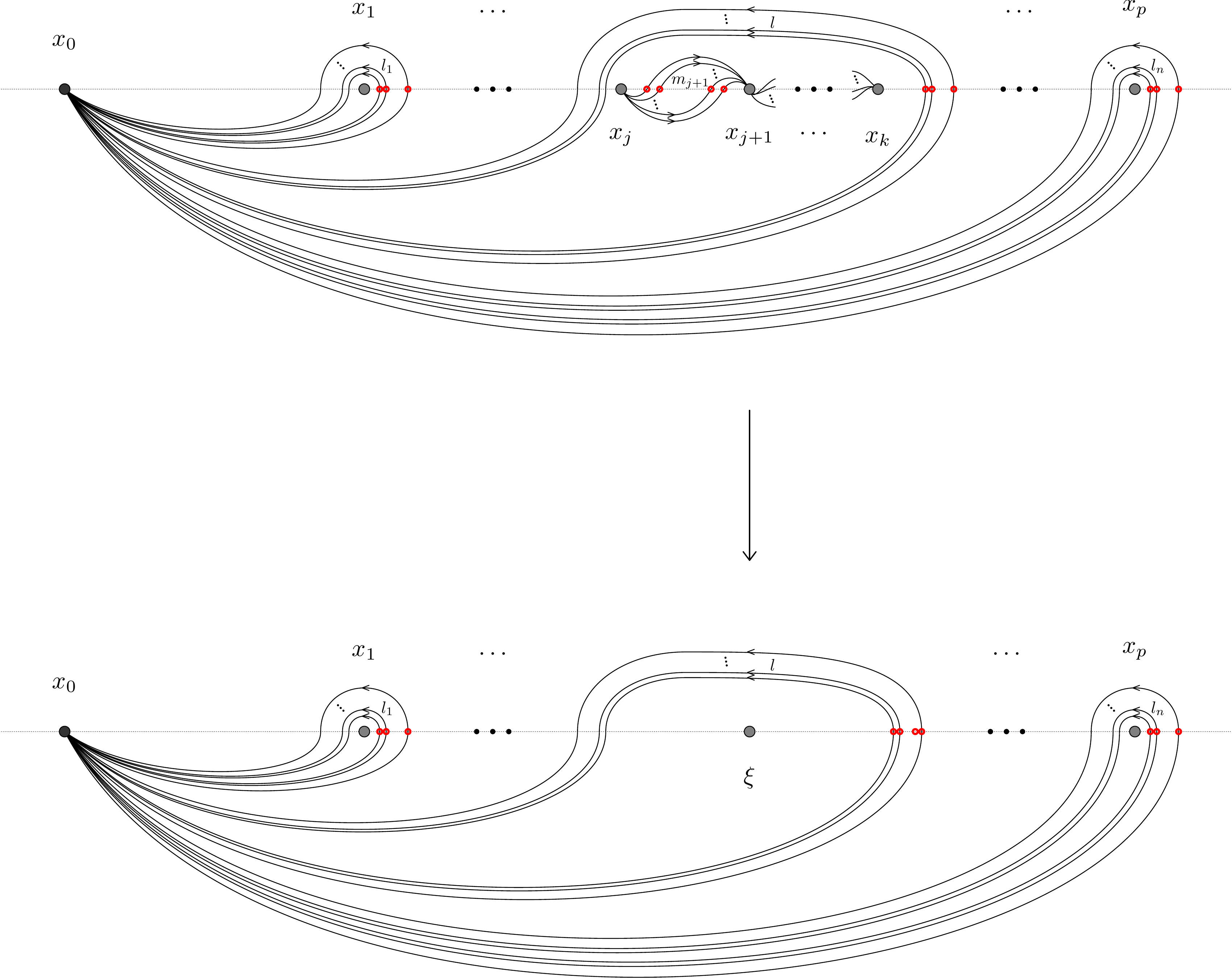}
\bigskip

\caption{\label{fig: limits}
To deal with limits $x_j,  \ldots,x_k \to \xi$, the main idea is to write the integration surface $\Gamma$ of~\eqref{eq: ansatz for solution to PDEs}
as a linear combination of surfaces appearing in the top figure, 
where a number $l \geq 0$ of non-intersecting nested loops surround the collapsing points $x_j,  \ldots,x_k$ and 
there are $m_{j+1}, \ldots, m_{k} \geq 0$ contours between these points.
The latter form a ``deformed hypercube integral'' $\widetilde{\rho}_{0,m_{j+1},\ldots,m_{k}}(x_j, \ldots, x_k)$,
illustrated in Figure~\ref{fig: rho}, which can in some cases be evaluated explicitly.
The red circles indicate a choice of branch for the integrand in~\eqref{eq: ansatz for solution to PDEs},
so that it is real and positive when the integration variables lie at those points,
see~\cite{Kytola-Peltola:Conformally_covariant_boundary_correlation_functions_with_quantum_group} for details.
The integral function~\eqref{eq: ansatz for solution to PDEs} with surface $\Gamma$ as 
in the top figure is denoted by $\alpha^{(x_0)}_{l_1,\ldots l_{j-1};l,\{m_{j+1},\ldots,m_{k}\};l_{k+1},\ldots,l_p}(x_1, \ldots, x_p)$.
}
\end{figure}

By changing the normalization of the function $\BasisF_\linkpatt$
in Lemma~\ref{lem: limits at same rate}, we can remove the restriction~\eqref{eq: limit at same rate}.

\begin{prop} \label{prop: strong limits}
Let $1 \leq j < k \leq p$, and $x_{j-1} < \xi < x_{k+1}$. Then we have
\begin{align} \label{eq: strong limit}
\lim_{x_j, x_{j+1}, \ldots, x_k \to \xi}
\frac{\BasisF_\linkpatt(x_1,\ldots,x_p)}{\BasisF_\tau(x_j,\ldots,x_k)}
= \; & \BasisF_{\linkpatt \removeLink \tau}(x_1,\ldots,x_{j-1},\xi,x_{k+1},\ldots,x_p).
\end{align}
\end{prop}
\begin{proof}
First, by~\cite[Proposition~4.5]{Kytola-Peltola:Conformally_covariant_boundary_correlation_functions_with_quantum_group},
we can write $\BasisF_\tau = \sF[\Puregeom_\tau]$ in the form
\begin{align*}
\sF[\Puregeom_\tau] (x_j, \ldots, x_k)
    = \sum_{m_{j+1},\ldots,m_{k} \geq 0} a_{m_{j+1},\ldots,m_{k}} \times  \widetilde{\rho}_{0,m_{j+1},\ldots,m_{k}}(x_j, \ldots, x_k) , 
\end{align*}
where $a_{m_{j+1},\ldots,m_{k}} \in \bC$ are some constants, each
$\widetilde{\rho}_{0,m_{j+1},\ldots,m_{k}}$ denotes a Coulomb gas integral of type~\eqref{eq: ansatz for solution to PDEs}
with $\Gamma$ a surface of type depicted in Figure~\ref{fig: rho}, and 
\begin{align}  \label{eq: ms}
m_{j+1} + \cdots + m_{k} = \frac{1}{2} \left(\sum_{i=j}^k d_i - k + j - \projdmn \right) .
\end{align}

Second, by~\cite[Proof of Proposition~5.1]{Kytola-Peltola:Conformally_covariant_boundary_correlation_functions_with_quantum_group},
we can write the functions $\sF^{(x_0)}[v]$ for vectors
\[ v = \left(\Wbas_{l_p} \tens \cdots \tens \Wbas_{l_{k+1}} \tens F^l.\Puregeom_{\tau} \tens \Wbas_{l_{j-1}} \tens \cdots \tens \Wbas_{l_1}\right) \]
in the form
\[ \sF^{(x_0)}[v] (x_1, \ldots, x_p)
    = \sum_{m_{j+1},\ldots,m_{k} \geq 0} a_{m_{j+1},\ldots,m_{k}}
        \times \alpha^{(x_0)}_{l_1,\ldots,l_{j-1};l,\set{m_{j+1},\ldots,m_{k}};l_{k+1},\ldots,l_p}(x_1, \ldots, x_p) , \]
where each $\alpha^{(x_0)}_{\cdots;l,\set{m_{j+1},\ldots,m_{k}};\cdots}$ denotes a Coulomb gas integral of type~\eqref{eq: ansatz for solution to PDEs}
with $\Gamma$ a surface of type depicted in Figure~\ref{fig: limits}(top).
We note that these integration surfaces a priori depend on an auxiliary point $x_0$,
but as proved in \cite[Proposition~4.5]{Kytola-Peltola:Conformally_covariant_boundary_correlation_functions_with_quantum_group},
all solutions to~\eqref{eq: BSA differential equations} are independent of $x_0$.

All in all, we can write the ratio appearing in the asserted equation~\eqref{eq: strong limit} in the form
\begin{align*}
& \frac{\sF[\Puregeom_\linkpatt](x_1,\ldots,x_p)}{\sF[\Puregeom_\tau](x_j,\ldots,x_k)} \\
= \; & \sum_{l=0}^r \sum_{\substack{l_1,\ldots,l_{j-1},\\l_{k+1},\ldots,l_p}} c_{l_1,\ldots,l_{j-1};l;l_{k+1},\ldots,l_p} 
\frac{\underset{m_{j+1},\ldots,m_{k} \geq 0}{\sum} a_{m_{j+1},\ldots,m_{k}}
        \times \alpha^{(x_0)}_{l_1,\ldots,l_{j-1};l,\set{m_{j+1},\ldots,m_{k}};l_{k+1},\ldots,l_p}(x_1, \ldots, x_p)}{\underset{m_{j+1},\ldots,m_{k} \geq 0}{\sum} a_{m_{j+1},\ldots,m_{k}} \times  \widetilde{\rho}_{0,m_{j+1},\ldots,m_{k}}(x_j, \ldots, x_k)} ,
\end{align*}
where we also used Equation~\eqref{eq: vector omega}.

Then, using Equation~\eqref{eq: vector omega minus tau}, we write the right hand side of the asserted equation~\eqref{eq: strong limit} in the form
\begin{align*}
& \sF[\Puregeom_{\linkpatt \removeLink \tau}] (x_1,\ldots,x_{j-1},\xi,x_{k+1},\ldots,x_p) \\
= \; & \sum_{l=0}^r \sum_{\substack{l_1,\ldots,l_{j-1},\\l_{k+1},\ldots,l_p}} c_{l_1,\ldots,l_{j-1};l;l_{k+1},\ldots,l_p} \times
\varphi^{(x_0)}_{l_1,\ldots,l_{j-1},l,l_{k+1},\ldots,l_p} (x_1,\ldots,x_{j-1},\xi,x_{k+1},\ldots,x_p) ,
\end{align*}
where each $\varphi^{(x_0)}_{l_1,\ldots,l_{j-1},l,l_{k+1},\ldots,l_p}$ denotes a Coulomb gas integral of type~\eqref{eq: ansatz for solution to PDEs}
with $\Gamma$ a surface of type depicted in Figure~\ref{fig: basis}.

Now, to evaluate the limit~\eqref{eq: strong limit}, we can apply dominated convergence theorem 
to the integration over all variables in $\alpha^{(x_0)}_{\cdots;l,\set{m_{j+1},\ldots,m_{k}};\cdots}$
whose contour is a loop, since these contours remain bounded away from the points $x_j , \ldots , x_k$ 
and any hypercube type integration contours between them. On the other hand, we note that the hypercube 
integrals are the same in $\alpha^{(x_0)}_{\cdots;l,\set{m_{j+1},\ldots,m_{k}};\cdots}$ and in $\widetilde{\rho}_{0,m_{j+1},\ldots,m_{k}}$,
and they cancel each other in the limit~\eqref{eq: strong limit}.
We also note that the remaining loop integrals in $\alpha^{(x_0)}_{\cdots;l,\set{m_{j+1},\ldots,m_{k}};\cdots}$ are the same as in 
$\varphi^{(x_0)}_{l_1,\ldots,l_{j-1},l,l_{k+1},\ldots,l_p}$.

To finish, we consider the integrand~\eqref{eq: integrand} for 
$\alpha^{(x_0)}_{\cdots;l,\set{m_{j+1},\ldots,m_{k}};\cdots} (x_1, \ldots, x_p)$:
\begin{align*}
\prod_{1\leq i<j\leq p}(x_{j}-x_{i})^{\frac{2}{\kappa} s_i s_j}
\prod_{\substack{1\leq i\leq p\\1\leq r\leq\ell}}
(w_{r}-x_{i})^{-\frac{4}{\kappa} s_i}
\prod_{1\leq r<s\leq\ell}(w_{s}-w_{r})^{\frac{8}{\kappa}} .
\end{align*}
Using the identity~\eqref{eq: ms}, we see that as $x_j, \ldots, x_k \to \xi$, the factors containing these variables converge to
the corresponding factors in the integrand for $\varphi^{(x_0)}_{l_1,\ldots,l_{j-1},l,l_{k+1},\ldots,l_p}(x_1,\ldots,x_{j-1},\xi,x_{k+1},\ldots,x_p)$,
where those terms have the form $(\xi-x_{i})^{\frac{2}{\kappa} s_i(\projdmn-1)}$ and 
$(w_{r}-\xi)^{-\frac{4}{\kappa} s_i(\projdmn-1)}$.
This gives the asserted result.
\end{proof}

\begin{figure}
\includegraphics[scale=.28]{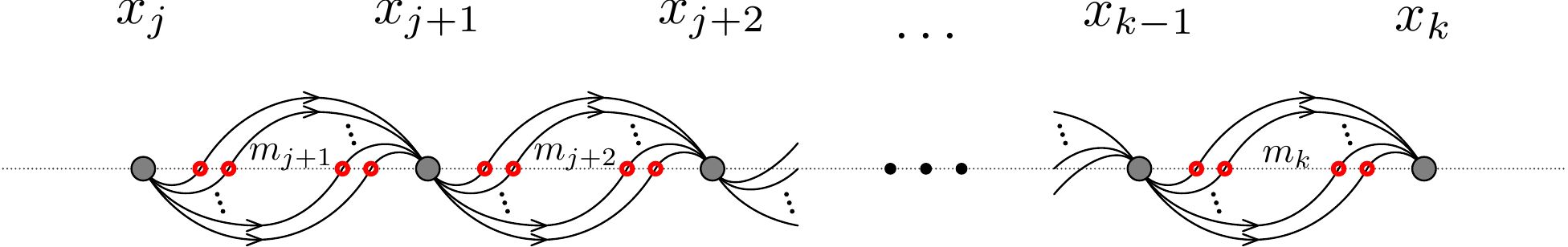}
\bigskip

\caption{\label{fig: rho}
Illustration of the integration surface for a ``deformed hypercube integral'' $\widetilde{\rho}_{0,m_{j+1},\ldots,m_{k}}(x_j, \ldots, x_k)$.
The red circles indicate a choice of branch for the integrand in~\eqref{eq: ansatz for solution to PDEs},
so that it is real and positive when the integration variables lie at those points,
see~\cite{Kytola-Peltola:Conformally_covariant_boundary_correlation_functions_with_quantum_group} for details.
}
\end{figure}

As a final remark, we observe that the functions $\BasisF_\linkpatt$ of Theorem~\ref{thm: asymptotic properties of general basis vectors}
can be realized as limits of functions $\BasisF_{\alpha(\linkpatt)} = \sF[\Puregeomtwodim_{\alpha(\linkpatt)}]$,
where $\linkpatt \mapsto \alpha(\linkpatt)$ is the map~\eqref{eq: link pattern to pair partition}
and $\Puregeomtwodim_\alpha$ the vectors of Theorem~\ref{thm: existence of multiple SLE vectors}.
For the precise statement, we denote $n = |\multii| := \sum_{i=1}^p s_i$, as in~\eqref{eq: definition of n and s}, 
and for $i \in \{1,\ldots, p\}$, we let $\partition_i$ be
the partition of $s_i$ into positive integers with size
$|\partition_i| = s_i$, that is, $\partition_i = (1,1,\ldots,1)$ with $s_i$ parts all equal to one.

\begin{cor}\label{cor: fusion}
Along any sequence $(y_1, \ldots, y_n) \in \chamber_n$ 
converging to $(x_1,\ldots,x_p) \in \chamber_p$ as shown,
we~have 
\begin{align*}
\BasisF_\linkpatt(x_1,\ldots,x_p) = \lim_{\substack{\vspace*{1mm} \\ \hspace*{6.5mm} y_1, \ldots,y_{s_1} \to x_1 \\ \hspace*{2mm} y_{s_1+1}, \ldots,y_{s_2} \to x_2 \\ \vdots \\ y_{n-s_p+1}, \ldots,y_{n} \to x_p}}
\frac{\BasisF_{\alpha(\linkpatt)} (y_1, \ldots, y_n)}{\BasisF_{\defpatt_{\partition_1}} (y_1, \ldots,y_{s_1}) 
\BasisF_{\defpatt_{\partition_2}} (y_{s_1+1}, \ldots,y_{s_2}) \cdots
\BasisF_{\defpatt_{\partition_p}} (y_{n-s_p+1}, \ldots,y_{n})} .
\end{align*}
\end{cor}
\begin{proof}
This follows from definitions and Proposition~\ref{prop: strong limits}.
The integral forms of the functions guarantee that we can take the limits in any order, and Proposition~\ref{prop: strong limits} 
that the limits exist along any sequence.
\end{proof}

The statement of Corollary~\ref{cor: fusion} is very natural in the sense of fusion in CFT. Indeed, viewed as correlation functions,
the solutions $\BasisF_\linkpatt$ should be obtained by fusion from the solutions $\BasisF_{\alpha(\linkpatt)}$.
We also note that the functions $\BasisF_{\defpatt_{\partition_i}}$ appearing in the denominator
in Corollary~\ref{cor: fusion} have a simple form, given by Lemma~\ref{lem: basis functions with only defects}.

\subsection{\label{subsec: limits of functions at infinity}Limits when taking variables to infinity}

From the cyclic permutation symmetry of the vectors 
$\Puregeom_{\linkpatt} \in \HWsp_\multii^{(0)}$ (Corollary~\ref{cor: cyclic symmetry})
we can derive a similar property for the M\"obius covariant functions 
$\BasisF_\linkpatt = \sF[\Puregeom_\linkpatt]$,
concerning the limit when the rightmost variable tends to $+\infty$. Indeed, this 
limit is equal to the limit of the function $\BasisF_{\Scomb(\linkpatt)}$ as its 
leftmost variable tends to $-\infty$, where $\Scomb$ is 
the cyclic permutation map defined in Equation~\eqref{eq: combinatorial Smap}
in Section~\ref{sec: cyclic permutation symmetry}, and illustrated in Figure~\ref{fig: cyclic permutation}.

\begin{prop}\label{prop: cyclicity of limits}
For any $\linkpatt \in \LP_{(\multii,s)}^{(0)}$, we have
$($with $d = s + 1)$
\begin{align*}
\lim_{y \to + \infty} \Big( y^{2h_{1,d}} \times 
\sF[\Puregeom_{\linkpatt}](x_1,\ldots,x_{p},y) \Big)
= \lim_{y \to - \infty} \Big( |y|^{2h_{1,d}} \times 
\sF[\Puregeom_{\Scomb(\linkpatt)}](y,x_1,\ldots,x_{p}) \Big).
\end{align*} 
\end{prop}
\begin{proof}
By~\cite[Proposition~5.4]{Kytola-Peltola:Conformally_covariant_boundary_correlation_functions_with_quantum_group},
we have
\begin{align*}
\lim_{y \to + \infty} \Big( y^{2h_{1,d}} \times 
\sF[\Puregeom_{\linkpatt}](x_1,\ldots,x_{p},y) \Big)
    = & \; (q-q^{-1})^{d-1} \qfact{d-1}^2 \times B^{d,d}_1
     \times \sF[R_+^{(s)}(\Puregeom_{\linkpatt})](x_1,\ldots,x_{p}) , \\
\lim_{y \to - \infty} \Big( |y|^{2h_{1,d}} \times 
\sF[\Puregeom_{\Scomb(\linkpatt)}](y,x_1,\ldots,x_{p}) \Big)
    =  & \; (q^{-2}-1)^{d-1} \qfact{d-1}^2 \times B^{d,d}_1
     \times \sF[R_-^{(s)}(\Puregeom_{\Scomb(\linkpatt)})]
    (x_1,\ldots,x_{p}) ,
\end{align*}
where $B^{d,d}_1$ is the constant defined in~\eqref{eq: Selberg integral}.
Using Corollary~\ref{cor: cyclic symmetry}, we calculate
\begin{align*}
& \lim_{y \to - \infty} \Big( |y|^{2h_{1,d}} \times 
\sF[\Puregeom_{\Scomb(\linkpatt)}](y,x_1,\ldots,x_{p}) \Big) \\
= & \; \frac{(q^{-2}-1)^{d-1} \qfact{d-1}^2 \times B^{d,d}_1 \times (-q)^{d-1}}{(q-q^{-1})^{d-1} \qfact{d-1}^2 \times B^{d,d}_1}
\lim_{y \to + \infty} \Big( y^{2h_{1,d}} \times 
\sF[\Puregeom_{\linkpatt}](x_1,\ldots,x_{p},y) \Big) \\
= & \; \lim_{y \to + \infty} \Big( y^{2h_{1,d}} \times 
\sF[\Puregeom_{\linkpatt}](x_1,\ldots,x_{p},y) \Big).
\end{align*} 
\end{proof}

\bigskip{}
\section{\label{sec: multiple SLEs}Cyclic permutation symmetry of the pure partition functions of multiple $\SLE$s}

Multiple $(\SLEk)_{\kappa \geq 0}$ is a collection of random conformally invariant curves  
started from given boundary points 
of a simply connected domain, and connecting them pairwise without crossing
\cite{BBK:Multiple_SLEs_and_statistical_mechanics_martingales,
Dubedat:Commutation_relations_for_SLE, 
Kozdron-Lawler:Configurational_measure_on_mutually_avoiding_SLEs,
Kytola-Peltola:Pure_partition_functions_of_multiple_SLEs,
PW:Global_multiple_SLEs_and_pure_partition_functions}.
Such curves describe scaling limits of interfaces in statistical mechanics models.
Indeed, convergence of a single interface to the $\SLEk$ has now been proven for a number of 
models, see, e.g.,~\cite{Smirnov:Critical_percolation_in_the_plane,
LSW:Conformal_invariance_of_planar_LERW_and_UST,
Schramm-Sheffield:Harmonic_explorer_and_its_convergence_to_SLE4,
Smirnov:Towards_conformal_invariance_of_2D_lattice_models,
Zhan:Scaling_limits_of_planar_LERW_in_finitely_connected_domains,
Hongler-Kytola:Ising_interfaces_and_free_boundary_conditions,
CDHKS:Convergence_of_Ising_interfaces_to_SLE},
and convergence of several interfaces to multiple $\SLE$s 
has also been established in some cases~\cite{Camia-Newman:2D_percolation_full_scaling_limit,
Izyurov:Critical_Ising_interfaces_in_multiply_connected_domains,
BPW:On_the_uniqueness_of_global_multiple_SLEs, KS-configurations_of_FK_Ising_interfaces}.

A multiple $\SLE$ can be constructed as a growth process encoded in 
a Loewner chain, see~\cite{Schramm:Scaling_limits_of_LERW_and_UST, 
Dubedat:Commutation_relations_for_SLE,
Kytola-Peltola:Pure_partition_functions_of_multiple_SLEs, PW:Global_multiple_SLEs_and_pure_partition_functions}.
As an input for the construction, 
one uses a function $\PartF \colon \chamber_{2N} \to \bRpos$, 
called a partition function of the multiple $\SLEk$. 
This function appears in the Radon-Nikodym derivative of the multiple $\SLE_\kappa$ measure 
with respect to the product measure of independent $\SLE_\kappa$ curves.
It must satisfy the second order PDE system
\begin{align}\label{eq: multiple SLE PDEs}
& \left[\frac{\kappa}{2}\pdder{x_{i}} + \sum_{j\neq i}
\left(\frac{2}{x_{j}-x_{i}}\pder{x_{j}}-\frac{2h_{1,2}}{(x_{j}-x_{i})^{2}}\right)\right]\PartF(x_{1},\ldots,x_{2N}) = 0 ,
\qquad\text{for all } i \in \{1,2,\ldots,2N\} .
\end{align}
With translation invariance, \eqref{eq: multiple SLE PDEs} 
are equivalent to the second order Benoit~\& Saint-Aubin PDEs.
Furthermore, by conformal invariance of the multiple $\SLE_\kappa$,
the partition function $\PartF$ must be covariant under all 
M\"obius transformations $\Mob\colon\bH\to\bH$ 
such that $\Mob(x_{1})<\Mob(x_{2})<\cdots<\Mob(x_{p})$:
\begin{align*}
\PartF(x_{1},\ldots, x_{2N}) = 
\prod_{i=1}^{2N}\Mob'(x_{i})^{h_{1,2}} \times 
\PartF \left(\Mob(x_{1}),\ldots,\Mob(x_{2N})\right) .
\end{align*}

The law of the multiple $\SLE_\kappa$ is not unique, for the random curves may have several 
topological connectivities of the marked boundary points. 
The connectivities are encoded in planar pair partitions $\alpha \in \PP_N$.
In fact, the convex set of probability measures of (local) multiple $\SLEk$ processes
is in one-to-one correspondence with the set of positive (and normalized) partition functions $\PartF$ ---
see~\cite{Dubedat:Commutation_relations_for_SLE, Kytola-Peltola:Pure_partition_functions_of_multiple_SLEs}.
The extremal points of this convex set correspond to the different possible connectivities~\cite{PW:Global_multiple_SLEs_and_pure_partition_functions}.

Pertaining to the construction of the extremal processes,
in~\cite{Kytola-Peltola:Pure_partition_functions_of_multiple_SLEs}
a basis $\left( \PartF_\alpha \right)_{\alpha \in \PP_N}$ of M\"obius covariant
solutions to the PDE system~\eqref{eq: multiple SLE PDEs} was constructed, using 
Theorem~\ref{thm: SCCG correspondence} and the vectors $\Puregeomtwodim_\alpha$ 
of Theorem~\ref{thm: existence of multiple SLE vectors}.
A defining property of the basis functions $\PartF_\alpha$ is 
the recursive asymptotics property
\begin{align}\label{eq: multiple SLE asymptotics}
\lim_{x_{j},x_{j+1}\to\xi}
\frac{\PartF_{\alpha}(x_{1},\ldots,x_{2N})}{|x_{j+1}-x_{j}|^{-2h_{1,2}}} =\; & \begin{cases}
0\quad & \text{if }\linkInEquation{j}{j+1}\notin\alpha\\
\PartF_{\hat{\alpha}}(x_{1},\ldots,x_{j-1},x_{j+2},\ldots,x_{2N}) & \text{if }\linkInEquation{j}{j+1}\in\alpha ,
\end{cases}
\end{align}
with $\hat{\alpha} = \alpha \removeLink \link{j}{j+1}$,
for any $x_{j-1} < \xi < x_{j+2}$ and $j \in \{1,2,\ldots,2N-1\}$, and
$\kappa \in (0,8) \setminus \bQ$.

In~\cite{Kytola-Peltola:Pure_partition_functions_of_multiple_SLEs},
these functions $\PartF_{\alpha}$ were called the pure partition functions of 
the multiple $\SLEk$. They were argued to be the partition functions of the extremal multiple $\SLEk$ processes, 
with the deterministic connectivities $\alpha$. A proof for this fact appeared recently in~\cite{PW:Global_multiple_SLEs_and_pure_partition_functions}
in the case $0 < \kappa \leq 4$.

Specifically, with $q = e^{\ii \pi 4 / \kappa}$,
the pure partition functions were constructed in~\cite{Kytola-Peltola:Pure_partition_functions_of_multiple_SLEs} as
\[ \PartF_\alpha := (B^{2,2}_1)^{-N} \sF[\Puregeomtwodim_{\alpha}] , \qquad \text{for } \alpha \in \PP_N , \]
with the normalization constant chosen in such a way that the functions 
$\PartF_\alpha$ 
satisfy the asymptotics~\eqref{eq: multiple SLE asymptotics}
with no constants in front.
The property~\eqref{eq: multiple SLE asymptotics} is in fact a special case 
of Theorem~\ref{thm: asymptotic properties of general basis vectors}, and indeed, the more general functions 
$\BasisF_\linkpatt = \sF[\Puregeom_\linkpatt]$ should 
provide pure partition functions of systems of multiple $\SLEk$ curves
growing from the boundary, 
in the spirit of~\cite{Friedrich-Werner:Conformal_restriction_highest_weight_representations_and_SLE,
Kontsevich:CFT_SLE_and_phase_boundaries,
Friedrich-Kalkkinen:On_CFT_and_SLE,
Duplantier:Conformal_random_geometry,
Kontsevich-Suhov:On_Malliavin_measures_SLE_and_CFT,
Dubedat:SLE_and_Virasoro_representations_fusion}. 
Also, the functions $\BasisF_\linkpatt$ describe 
observables concerning geometric properties of interfaces
--- see, e.g.,~\cite{Gamsa-Cardy:The_scaling_limit_of_two_cluster_boundaries_in_critical_lattice_models,
BJV:Some_remarks_on_SLE_bubbles_and_Schramms_2point_observable,
Dubedat:SLE_and_Virasoro_representations_fusion,
JJK:SLE_boundary_visits,
KKP:Conformal_blocks_pure_partition_functions_and_KW_binary_relation,
LV:Coulomb_gas_for_commuting_SLEs, LV:Coulomb_gas_integrals_PART2,
PW:Global_multiple_SLEs_and_pure_partition_functions}.

In Corollary~\ref{cor: removing link around infinity} we show that 
the property~\eqref{eq: multiple SLE asymptotics} of the pure partition functions
$\PartF_\alpha$ is also true when taking 
the limit $x_1 \to -\infty$ and $x_{2N} \to +\infty$, corresponding to the removal
of the link $\link{1}{2N}$. 
We also consider the more general pure partition functions
\begin{align}\label{eq: pure partition functions with defects}
\PartF_\alpha := (B_1^{2,2})^{-N} \BasisF_\alpha,  \qquad \text{for } \alpha \in \PP_N^{(s)},
\end{align}
which are homogeneous
solutions to the second order PDEs~\eqref{eq: multiple SLE PDEs},
but, when $s > 0$, not covariant under all M\"obius maps in the covariance 
formula~\eqref{eq: Mobius covariance}.
We prove in Proposition~\ref{prop: FK dual elements}
that these functions are linearly independent, and thus obtain 
a basis of a solution space of the PDE system~\eqref{eq: multiple SLE PDEs}.

\subsection{\label{subsec: cyclic permutation symmetry for multiple SLEs}Cyclic permutation symmetry}

Corollary~\ref{cor: cyclic symmetry} gives a general cyclic permutation symmetry of the vectors
$\Puregeom_{\linkpatt}\in\HWsp_\multii^{(0)} \subset \bigotimes_{i=1}^{p}\Wd_{d_{i}}$
in the trivial subrepresentation.
The special case of $n=2N$, with $d_i=2$, for all $i$, can be used in applications 
to the properties of the pure partition functions $\PartF_\alpha$.

First, from Proposition~\ref{prop: cyclicity of limits} we immediately get
the following corollary.

\begin{cor}\label{cor: cyclicity of limits for pure geometries}
Let $\kappa \in (0,\infty) \setminus \bQ$. Then we have
\begin{align*}
\lim_{y \to + \infty} \Big( y^{2h_{1,2}} \times 
\PartF_\alpha(x_1,\ldots,x_{2N-1},y) \Big)
= \lim_{y \to - \infty} \Big( |y|^{2h_{1,2}} \times 
\PartF_{\Scomb(\alpha)}(y,x_1,\ldots,x_{2N-1}) \Big).
\end{align*} 
\end{cor}
\begin{proof}
The assertion follows directly from the definition
$\PartF_\alpha := (B_1^{2,2})^{-N} \sF[\Puregeomtwodim_{\alpha}]$ and Proposition~\ref{prop: cyclicity of limits}.
\end{proof}

Using this, we can extend the cascade property~\eqref{eq: multiple SLE asymptotics} 
for the pure partition functions to $j = 2N$. 

\begin{cor}\label{cor: removing link around infinity}
Let $\kappa \in (0,8) \setminus \bQ$. Denote by
$\hat{\alpha} = \alpha\removeLink\link{1}{2N}$. Then we have
\begin{align*}
\lim_{\substack{x_{1}\to -\infty, \\ x_{2N}\to +\infty}}
|x_{2N}-x_{1}|^{2h_{1,2}} \times \PartF_{\alpha}(x_{1},\ldots,x_{2N})
=\; & \begin{cases}
0\quad & \text{if }\linkInEquation{1}{2N}\notin\alpha\\
\PartF_{\hat{\alpha}}(x_{2},\ldots,x_{2N-1}) 
& \text{if } \linkInEquation{1}{2N}\in\alpha.
\end{cases}
\end{align*} 
\end{cor}
\begin{proof}
We can take either limit first and get the same result.
Observe first that we have $\link{1}{2N}\in\alpha$ if and only if 
$\link{1\;}{\;2\;}\in\Scomb(\alpha)$, and in that case, 
$\Scomb(\alpha)\removeLink\link{1\;}{\;2\;} = \alpha\removeLink\link{1}{2N}$.
Using Corollary~\ref{cor: cyclicity of limits for pure geometries},
the asymptotics property~\eqref{eq: multiple SLE asymptotics}
for $\PartF_{\Scomb(\alpha)}$, and the above observation concerning the links, 
we can calculate the limit
\begin{align*}
& \lim_{x_{1}\to -\infty} \lim_{x_{2N}\to +\infty}
|x_{2N}-x_{1}|^{2h_{1,2}} \times \PartF_{\alpha}(x_{1},\ldots,x_{2N}) \\
= & \; \lim_{x_{1}\to -\infty} \lim_{x_{2N}\to -\infty}
|x_{2N}-x_{1}|^{2h_{1,2}} \times 
\PartF_{\Scomb(\alpha)}(x_{2N},x_{1},x_{2}\ldots,x_{2N-1}) \\
= & \; \begin{cases}
0\quad & \text{if }\linkInEquation{1\;}{\;2\;}\notin\Scomb(\alpha) \\
\PartF_{\Scomb(\alpha)\removeLink\linksmallerInEquation{1\;}{\;2\;}}(x_{2},\ldots,x_{2N-1})  & \text{if } \linkInEquation{1\;}{\;2\;}\in\Scomb(\alpha)
\end{cases} \\
= & \; \begin{cases}
0\quad & \text{if }\linkInEquation{1\,}{\,2N}\notin\alpha \\
\PartF_{\alpha\removeLink\linksmallerInEquation{1\,}{\,2N}}(x_{2},\ldots,x_{2N-1})  & \text{if } \linkInEquation{1\,}{\,2N}\in\alpha.
\end{cases}
\end{align*} 
\end{proof}
 
Corollary~\ref{cor: removing link around infinity} combined with
Equation~\eqref{eq: multiple SLE asymptotics} shows that linear 
combinations of the basis functions $\PartF_{\alpha}$ have a cascade property
with respect to removing any link connecting consecutive points, when the boundary
$\bdry \bH = \bR$ is viewed as the circle 
$\mathbb{S}^1=\set{z\in\bC\;\big|\;|z=1|}$, say.
Such a property is natural for the random $\SLE_\kappa$ type curves --- see 
Figure~\ref{fig: sle connectivities} for an illustration.
In fact, this cascade property can be taken as a defining property of a global multiple $\SLE_\kappa$,
see~\cite{Kozdron-Lawler:Configurational_measure_on_mutually_avoiding_SLEs, PW:Global_multiple_SLEs_and_pure_partition_functions}.

\begin{figure}
\includegraphics[scale=.6]{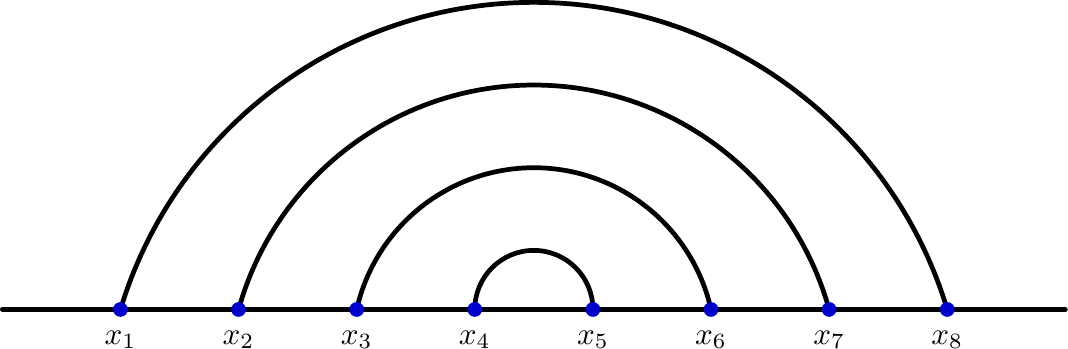} \qquad 
\includegraphics[scale=.6]{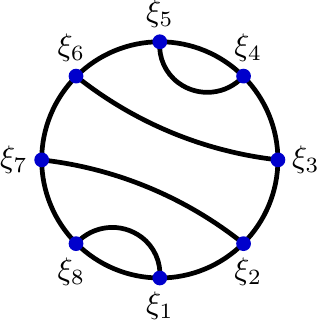}
\caption{\label{fig: sle connectivities}
The probability measure of multiple $\SLEk$ curves 
is conformally invariant. The figure depicts how a connectivity of the curves 
is mapped under a conformal map from the upper half-plane $\bH$ to the disc 
$\mathbb{D}$. The starting points $x_1 < x_2 < \ldots < x_{2N}$ of the curves 
in $\bH$ are mapped to the points 
$\bdrypt_1,\bdrypt_2,\ldots,\bdrypt_{2N}$ appearing in 
counterclockwise order along the boundary of $\mathbb{D}$.
The cascade property~\eqref{eq: multiple SLE asymptotics}
given in Corollary~\ref{cor: removing link around infinity}
for the outermost link connecting $x_1$ and $x_{2N}$ is manifest in the disc on the right.}
\end{figure}

\subsection{\label{subsec: injectivity}Linear independence for solutions to second order PDEs}

We now consider the functions $\PartF_\alpha$, with 
$\alpha \in \PP_N^{(s)}$,
defined in Equation~\eqref{eq: pure partition functions with defects}.
These functions form a basis of the solution space of the second order PDE 
system~\eqref{eq: multiple SLE PDEs}, consisting of homogeneous solutions,
in the sense of items (1) and (2) in 
Theorem~\ref{thm: asymptotic properties of general basis vectors}.
With $n = 2N + s$, this solution space is the image $\sF[\HWsp_{n}^{(s)}]$
of the highest weight vector space $\HWsp_{n}^{(s)}$
under the map $\sF$ of Theorem~\ref{thm: SCCG correspondence}.
We prove the linear independence of the functions $\PartF_\alpha$
by constructing a basis for the dual space
\[ \sF[\HWsp_{n}^{(s)}]^* 
= \set{\psi\colon\HWsp_{n}^{(s)} 
\to \bC \; | \; \psi \text{ is a linear map}} , \]
using similar ideas as 
in~\cite[Section~4.2]{Kytola-Peltola:Pure_partition_functions_of_multiple_SLEs},
where the case $s = 0$ was treated.

\begin{figure}
\includegraphics[scale=.75]{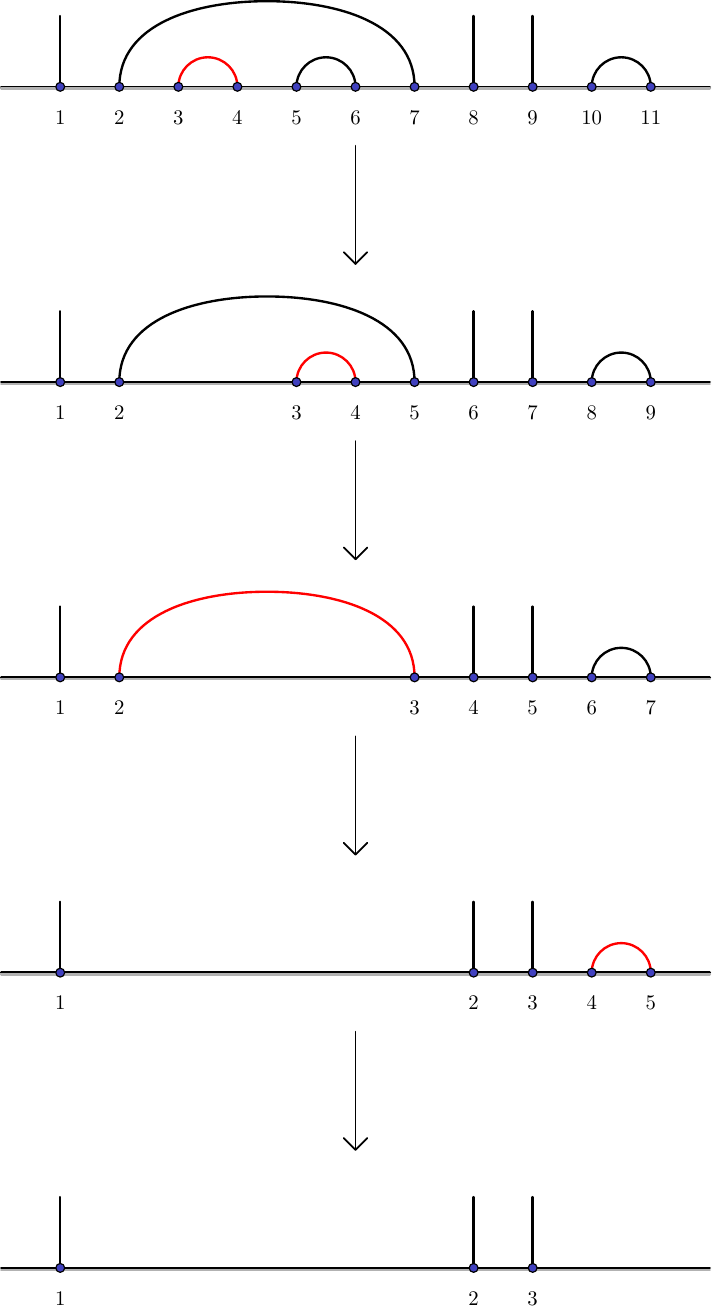}
\caption{\label{fig: allowable removal of links example}
Example of an allowable ordering to remove links. After the removal of all the links, the defects remain.
Notice also the relabeling of the indices after each step}
\end{figure}

\begin{prop}\label{prop: FK dual elements}
Let $\kappa \in (0,8) \setminus \bQ$,  
$s \in \bZnn$, and $n = 2N + s \in \bZpos$. The collection
$\left(\PartF_\alpha\right)_{\alpha\in\PP_N^{(s)}}$ is a
basis of the solution space $\sF[\HWsp_n^{(s)}]$ of
dimension $\frac{s + 1}{N+s+1}\binom{2N+s}{N+s}$.
\end{prop}
\begin{proof}
The case $s = 0$ was proved in
\cite[Proposition~4.2]{Kytola-Peltola:Pure_partition_functions_of_multiple_SLEs}.
The case $s > 0$ is very similar, so we only give the idea of the proof.
We consider the links in the link pattern 
\begin{align*}
\alpha = \Big\{\linkInEquation{a_1}{b_1},\ldots,\linkInEquation{a_N}{b_N}\Big\}
\bigcup
\Big\{\defectInEquation{c_1},\ldots,\defectInEquation{c_{s}}\Big\} 
\,\in\,\PP_N^{(s)}
\end{align*} 
as an ordered set, 
see Appendix~\ref{app: dual elements} 
and~\cite[Section~3.5]{Kytola-Peltola:Pure_partition_functions_of_multiple_SLEs}
for details.
We say that the ordering of the links is allowable for $\alpha$ 
if all links of $\alpha$ can be removed in such a way that at each step, the link 
to be removed connects two consecutive indices --- see 
Figure~\ref{fig: allowable removal of links example} for an illustration.
The precise definition of ``allowability'' was given in
\cite[Section~3.5]{Kytola-Peltola:Pure_partition_functions_of_multiple_SLEs}
for the case $s = 0$, but as the defects of $\alpha$ play no role in 
the link removal and no defects lie inside any link, the notion of an allowable 
ordering of links is the same for any $\alpha \in \PP_N^{(s)}$.

Suppose that the ordering of the links in $\alpha$ is allowable.
Then, by Theorem~\ref{thm: asymptotic properties of general basis vectors},
the iterated limit 
\begin{align*}
\FKdual_\alpha (\sF[v])
:= \; & \lim_{x_{a_{N}},x_{b_{N}}\to\xi_{N}} 
\cdots \lim_{x_{a_{1}},x_{b_{1}}\to\xi_{1}}
(x_{b_{N}}-x_{a_{N}})^{2h_{1,2}}
\cdots (x_{b_{1}}-x_{a_{1}})^{2h_{1,2}} \times \sF[v](x_1,\ldots,x_{n})
\end{align*}
exists for any $v \in \HWsp_n^{(s)}$. 
Consider the image $\BasisF_\alpha = \sF[\Puregeom_\alpha]$ of the basis vector 
$\Puregeom_\alpha \in \HWsp_n^{(s)}$.
Suppose that $c_1 < c_2 < \cdots < c_{s}$, and denote by
$y_i = x_{c_i}$, for $i \in \{ 1,2,\ldots,s \}$. 
Using the property~\eqref{eq: asymptotic properties} of $\BasisF_\alpha$ 
with the constant $\constantfromdiagram{1}{1}{1} = 1$ given 
by Equation~\eqref{eq: constant from diagram}, 
we evaluate the limit $\FKdual_\alpha (\BasisF_\alpha)$ as
\begin{align*}
\FKdual_\alpha (\BasisF_\alpha) (y_1,\ldots,y_{s}) =
(B_1^{2,2})^{N} \times \BasisF_{\defpatt_\partition}
(y_1,\ldots,y_{s}),
\qquad \text{for } \alpha \in \PP_N^{(s)},
\end{align*}
where $\partition = (1,1,\ldots,1,1) \in \bZ^s$
and $\BasisF_{\defpatt_\partition}$ has the explicit formula
given in Lemma~\ref{lem: basis functions with only defects}.
With the identification $\Puregeom_{\defpatt_\partition} \mapsto 1$
as in Remark~\ref{rem: one dimensionality of partition space},
and the formula in  Lemma~\ref{lem: basis functions with only defects},
we may interpret $\FKdual_\alpha (\PartF_\alpha) = 1$.

On the other hand, if $\beta \neq \alpha$, then the limit 
$\FKdual_\alpha (\BasisF_\beta)$ evaluates to zero, because of 
the property~\eqref{eq: asymptotic properties} and the fact that when 
$\beta \neq \alpha$, then we have $\link{a_j}{b_j} \notin \beta$, for some link
$\link{a_j}{b_j} \in \alpha$ in the allowable ordering.

It follows that the map 
$\FKdual_\alpha \colon \sF[\HWsp_n^{(s)}]\rightarrow\bC$
is well-defined and independent of the choice of the allowable ordering 
for $\alpha$. In particular, the collection
$\left(\FKdual_\alpha\right)_{\alpha\in\PP_N^{(s)}}$ is a
basis of the dual space 
$\sF[\HWsp_n^{(s)}]^*$,
such that 
\begin{align*}
\FKdual_\alpha ( \PartF_\beta ) = \delta_{\alpha,\beta} =\; & \begin{cases}
1\quad & \text{if } \beta = \alpha \\
0 & \text{if }\beta\neq\alpha .
\end{cases}
\end{align*}
Therefore, $\left(\PartF_\alpha\right)_{\alpha\in\PP_N^{(s)}}$ 
is a basis of the solution space 
$\sF[\HWsp_n^{(s)}]$, dual to 
$\left(\FKdual_\alpha\right)_{\alpha\in\PP_N^{(s)}}$.
The formula for the dimension of this space follows from 
Lemma~\ref{lem: cardinalities are equal} 
and~\cite[Lemma~2.2]{Kytola-Peltola:Pure_partition_functions_of_multiple_SLEs}:
$\# \PP_N^{(s)} = 
\frac{s + 1}{N+s+1}\binom{2N+s}{N+s}$.
\end{proof}

The linear independence of the functions $\PartF_\alpha$ immediately 
gives injectivity of the ``spin chain~--~Coulomb gas correspondence'' map $\sF$
in the case of 
$\multii = (1,1,\ldots,1,1)$ and $s \geq 0$.
This generalizes the previous injectivity result 
\cite[Corollary~4.3]{Kytola-Peltola:Pure_partition_functions_of_multiple_SLEs}.
We prove the injectivity of $\sF$ in full generality 
in forthcoming work~\cite{Flores-Peltola:Solution_space_of_BSA_PDEs},
where  we study solution spaces of the Benoit~\& Saint-Aubin PDEs in detail.

\begin{cor}\label{cor: injectivity}
For $s \in \bZnn$ and $n = 2N + s \in \bZpos$, the map 
$\sF \colon \HWsp_n^{(s)} \rightarrow \sC^{\infty}(\chamber_{n})$ is injective.
\end{cor}
\begin{proof}
The assertion follows by linearity from 
Propositions~\ref{prop: basis of highest weight vector space}~and~\ref{prop: FK dual elements} 
--- the images $\BasisF_\alpha = \sF[\Puregeom_\alpha]$ of the basis vectors 
$\Puregeom_\alpha$ of $\HWsp_n^{(s)}$ are linearly independent,
because the functions $\PartF_\alpha = (B_1^{2,2})^{-N}\sF[\Puregeom_\alpha]$
are.
\end{proof}

\appendix

\bigskip{}
\section{\label{app: q-combinatorics}$q$-combinatorics}

In this appendix, we prove ``$q$-combinatorial formulas'' needed in this article.
We first recall the definitions 
\begin{align*}
\qnum{m} = \; & \frac{q^{m}-q^{-m}}{q-q^{-1}}, \qquad \qquad
\qfact n = \prod_{m=1}^{n}\qnum m, \qquad \qquad
\qbin nk = \frac{\qfact n}{\qfact k \qfact{n-k}} ,
\end{align*}
for $q\in\bC\setminus\set{0}$ not a root of unity, and
$m\in\bZ$, and $n,k\in\bN$, with $0\leq k\leq n$. 

\begin{lem}
\label{lem: q-combinatorics}\
\begin{description}
\item [{(a)}] The $q$-binomial coefficients satisfy the recursion 
\begin{align*}
\qbin{n}{k} = q^k \qbin{n-1}{k} + q^{k-n} \qbin{n-1}{k-1}.
\end{align*}
\item [{(b)}]{\cite[Lemma~2.1(b)]{Kytola-Peltola:Conformally_covariant_boundary_correlation_functions_with_quantum_group}} 
For a permutation $\sigma\in \SymmGrp_{n}$ of $\set{1,2,\ldots,n}$,
denote by
\begin{align*}
{\rm inv}(\sigma)= \; & \set{(i,j) \;\big|\; i<j\text{ and }\sigma(i)>\sigma(j)}
\end{align*}
the set of inversions of $\sigma$. Then we have 
\begin{align*}
\sum_{\sigma\in \SymmGrp_{n}}q^{2\times\#{\rm inv}(\sigma)} = \; 
& q^{{\binom{n}{2}}} \qfact n.
\end{align*}
\item [{(c)}] For any $\nu_1,\nu_2 \in \bZnn$ and $n \in \bN$, we have
\begin{align*}
\sum_{k=0}^n \qbin{n}{k} q^{k(2n-\nu_1-\nu_2-2)}
\qfact{\nu_1-n+k} \qfact{\nu_2-k}
\; = \; \; & q^{n(n-\nu_1-1)} 
\frac{\qfact{\nu_1-n}\qfact{\nu_2-n}\qfact{\nu_1+\nu_2-n+1}}{\qfact{\nu_1+\nu_2-2n+1}}.
\end{align*}
\item [{(d)}] For any $\nu_1,\nu_2 \in \bZnn$ and $n \in \bN$, we have
\begin{align*}
\sum_{k=0}^n \qbin{n}{k} q^{k(\nu_1+\nu_2-2n+2)}
\qfact{\nu_1-k} \qfact{\nu_2-n+k}
\; = \; \; & q^{n(\nu_2+1-n)} 
\frac{\qfact{\nu_1-n}\qfact{\nu_2-n}\qfact{\nu_1+\nu_2-n+1}}{\qfact{\nu_1+\nu_2-2n+1}}.
\end{align*}
\end{description}
\end{lem}
\begin{proof}
The proof of (a) is a straightforward calculation using the definition 
of $q$-integers. Part (b) was proved in
\cite[Lemma~2.1(b)]{Kytola-Peltola:Conformally_covariant_boundary_correlation_functions_with_quantum_group}.
To prove part (c), we proceed by induction on $n$. For $n=0$, both sides of the 
equation are equal to $\qfact{\nu_1}\qfact{\nu_2}$. Denote by $L_n(\nu_1,\nu_2)$ 
and $R_n(\nu_1,\nu_2)$ the left and right hand sides of the asserted equation, 
respectively, and assume that we have $L_n(\nu_1,\nu_2) = R_n(\nu_1,\nu_2)$, 
for any $\nu_1,\nu_2 \in \bZnn$. 
Using part (a), we write $L_{n+1}(\nu_1,\nu_2)$ as
\begin{align*}
L_{n+1}(\nu_1,\nu_2)
= \; & \sum_{k=0}^{n} \qbin{n}{k} q^{k(2n+1-\nu_1-\nu_2)}
\qfact{\nu_1-n+k} \qfact{\nu_2-k}
\left( \frac{1}{\qnum{\nu_1-n+k}} + q^{n-\nu_1-\nu_2} \frac{1}{\qnum{\nu_2-k}} \right)\\
= \; & \sum_{k=0}^{n} \qbin{n}{k} q^{k(2n-\nu_1-\nu_2)}
\qfact{\nu_1-n+k} \qfact{\nu_2-k} \, q^{n-\nu_1}
\frac{\qnum{\nu_1+\nu_2-n}}{\qnum{\nu_2-k}\qnum{\nu_1-n+k}} \\
= \; & \qnum{\nu_1+\nu_2-n} q^{n-\nu_1} \times 
\sum_{k=0}^{n} \qbin{n}{k} q^{k(2n-\nu_1-\nu_2)} \qfact{\nu_1-1-n+k} \qfact{\nu_2-1-k},
\end{align*}
where we used the identity
\begin{align*}
q^{n-k-\nu_1}\qnum{\nu_1+\nu_2-n} = 
\qnum{\nu_2-k} + q^{n-\nu_1-\nu_2} \qnum{\nu_1-n+k}.
\end{align*}
By the induction hypothesis, the above sum is equal to 
$L_n(\nu_1-1,\nu_2-1) = R_n(\nu_1-1,\nu_2-1)$, so 
\begin{align*}
L_{n+1}(\nu_1,\nu_2) 
= \; & \qnum{\nu_1+\nu_2-n} q^{n-\nu_1} \times R_n(\nu_1-1,\nu_2-1) \\
= \; & \qnum{\nu_1+\nu_2-n} q^{n-\nu_1} \times 
q^{n(n-\nu_1)} 
\frac{\qfact{\nu_1-1-n}\qfact{\nu_2-1-n}\qfact{\nu_1+\nu_2-n-1}}{\qfact{\nu_1+\nu_2-2n-1}} \\
= \; & R_{n+1}(\nu_1,\nu_2),
\end{align*}
as claimed. This concludes the proof of (c).

Assertion (d) follows immediately from (c) and the symmetries
$q \leftrightarrow q^{-1}$ and $\nu_1 \leftrightarrow \nu_2$ of the identity.
\end{proof}

\bigskip{}
\section{\label{app: auxiliary calculations}Some auxiliary calculations}

In this appendix, we perform some auxiliary calculations needed in the proof of
Lemma~\ref{lem: commutative diagram} and 
Proposition~\ref{prop: commutative diagram}.
We will repeatedly use the notations~\eqref{eq: s in terms of d},
\begin{align*}
s = d - 1, \qquad
s_i = d_i - 1 , \quad \text{ for all } i \in \{1,2,\ldots,p\}, \qquad 
\text{ and } \qquad \multii = (s_1,s_2,\ldots,s_p) \in \bZnn^p.
\end{align*}
In the calculations, we consider the embedding from 
Section~\ref{subsec: powers of two-dimensionals}, defined for any 
$s = d-1 \in \bZpos$ as
\begin{align*}
\Embedding^{(s)} \colon \Wd_{d} \hookrightarrow \Wd_2^{\tens s} ,\qquad
\Embedding^{(s)}(\Wbas_l^{(d)}) := \MTbas_{l}^{(s)} ,
\quad\text{ for } l \in \{0, 1, \ldots, s \} .
\end{align*}

The vectors $\MTbas_{l}^{(s)}$ can be written explicitly as follows.
\begin{lem}\label{lem: action of F}
In the tensor product $\Wd_2^{\tens s}$, we have
\begin{align*}
\Embedding^{(s)}(\Wbas_k^{(d)}) = \MTbas_{k}^{(s)} 
= q^{\binom{k}{2}} \qfact k
\sum_{1\leq r_1<\cdots<r_k\leq s} q^{\sum_{i=1}^k(1-r_i)} \times \left(\Wbas_{l_1(\varrho)}\tens\cdots\tens\Wbas_{l_{s}(\varrho)}\right),
\end{align*}
$0\leq k\leq s = d-1$,
where we denote $(r_1,\ldots,r_k)=\varrho$, and,
for each $i \in \{1,2,\ldots,s \}$, 
\begin{align*}
l_i(\varrho) = \; & \begin{cases}
1 & \text{when }i\in\set{r_1,\ldots,r_k}\\
0 & \text{otherwise} .
\end{cases}
\end{align*}
\end{lem}
\begin{proof}
Using the coproduct~\eqref{eq: multiple coproducts} of $F$, and simplifying with 
Lemma~\ref{lem: q-combinatorics}(b), we calculate
\begin{align*}
\MTbas_{k}^{(s)} = F^k.\MTbas_{0}^{(s)} 
= \; & \sum_{\sigma\in\SymmGrp_{k}} q^{2\times\#{\rm inv}(\sigma)}
\sum_{1\leq r_1<\cdots<r_k\leq s} q^{\sum_{i=1}^k(1-r_i)}
\times \left(\Wbas_{l_1(\varrho)}\tens\cdots\tens\Wbas_{l_{s}(\varrho)}\right) \\
= \; & q^{\binom{k}{2}}
\qfact k\sum_{1\leq r_1<\cdots<r_k\leq s} q^{\sum_{i=1}^k(1-r_i)} 
\times \left(\Wbas_{l_1(\varrho)}\tens\cdots\tens\Wbas_{l_{s}(\varrho)}\right).
\end{align*}
\end{proof}

We will make use of the following formulas for the projection
$\hat{\pi} \colon \Wd_{2}\tens\Wd_{2} \to \bC$.
\begin{lem}{\cite[Lemma~2.3]{Kytola-Peltola:Pure_partition_functions_of_multiple_SLEs}}
\label{lem: projection formulas} For any $v \in \Wd_{2}\tens\Wd_{2}$, we have 
\begin{align*}
&\hat{\pi}^{(1)}(\Wbas_{0}\tens\Wbas_{0})=0,\qquad\qquad\quad\qquad\qquad\hat{\pi}^{(1)}(\Wbas_{1}\tens\Wbas_{1})=0,\\
&\hat{\pi}^{(1)}(\Wbas_{0}\tens\Wbas_{1})=\frac{q^{-1}-q}{\qnum 2},\qquad\qquad\qquad\hat{\pi}^{(1)}(\Wbas_{1}\tens\Wbas_{0})=\frac{1-q^{-2}}{\qnum 2}.
\end{align*}
\end{lem}

The next two lemmas explain how to calculate the projections appearing in the 
left column of the commutative diagram in Lemma~\ref{lem: commutative diagram}.

\begin{lem}\label{lem: projection in the middle}
Let $s_1,s_2 \in \bZpos$. Interpreting $\MTbas_{-1}^{(s)}=0$, we have
\begin{align*}
\hat{\pi}^{(1)}_{s_1} \left( \MTbas_{l}^{(s_2)}\tens\MTbas_{k}^{(s_1)} \right)
= \frac{(q-q^{-1})}{\qnum 2} \Big( q^{l-s_2-1-k} \, \; & \qnum {l} \times \left(\MTbas_{l-1}^{(s_2-1)}\tens\MTbas_{k}^{(s_1-1)}\right) 
- \qnum {k} \times \left(\MTbas_{l}^{(s_2-1)}\tens\MTbas_{k-1}^{(s_1-1)}\right) \Big).
\end{align*}
\end{lem}
\begin{proof}
Using Lemma~\ref{lem: action of F}, we write
\begin{align*}
\MTbas_{l}^{(s_2)}\tens\MTbas_{k}^{(s_1)}
= \; & q^{\binom{l}{2}}\qfact {l}\;
q^{\binom{k}{2}} \qfact {k}\sum_{\substack{1\leq r_1<\cdots<r_{l}\leq s_2 \\ 1\leq t_1<\cdots<t_{k}\leq s_1}} q^{\sum_{j=1}^{l}(1-r_j)+\sum_{j=1}^{k}(1-t_j)}\\
&\qquad\qquad\qquad\qquad\qquad\qquad\qquad\times
\left( \Wbas_{l_1(\varrho)}\tens\cdots\tens\Wbas_{l_{s_2}(\varrho)}\tens\Wbas_{k_1(\vartheta)}\tens\cdots
\tens\Wbas_{k_{s_1}(\vartheta)} \right),
\end{align*}
where $(r_1,\ldots,r_l)=\varrho$, and $(t_1,\ldots,t_k)=\vartheta$, and
\begin{align*}
l_i(\varrho) = \begin{cases}
1 & \text{when }i\in\set{r_1,\ldots,r_l}\\
0 & \text{otherwise} 
\end{cases} 
\qquad \qquad \text{ and } \qquad \qquad
k_i(\vartheta) = \begin{cases}
1 & \text{when }i\in\set{t_1,\ldots,t_k}\\
0 & \text{otherwise} .
\end{cases} 
\end{align*}
Using Lemma~\ref{lem: projection formulas}, we calculate
the action of the middle projection $\hat{\pi}^{(1)}_{s_1}$ 
on each term in the sum,
\begin{align*}
\hat{\pi}^{(1)} (\Wbas_{l_{s_2}(\varrho)}\tens\Wbas_{k_1(\vartheta)})
= \frac{q^{-1}-q}{\qnum 2} \left( \delta_{l_{s_2}(\varrho),0}\;\delta_{k_1(\vartheta),1}-q^{-1} \,\delta_{l_{s_2}(\varrho),1}\;\delta_{k_1(\vartheta),0}\right).
\end{align*}
Not all terms survive. First, when $k_1(\vartheta)=1$, we must have $t_1=1$, and 
similarly, when $l_{s_2}(\varrho)=1$, we must have $r_{l}=s_2$.
On the other hand, when $l_{s_2}(\varrho)=0$, then $r_{l}\neq s_2$,
so $r_{l}\leq s_2-1$, and similarly, when $k_1(\vartheta)=0$, then
$t_1\neq1$, so $2\leq t_1$. We thus obtain
\begin{align*}
& \hat{\pi}^{(1)}_{s_1} \left( \MTbas_{l}^{(s_2)}\tens\MTbas_{k}^{(s_1)} \right) \\
= \; & q^{\binom{l}{2}} \qfact{l} \; q^{\binom{k}{2}} \qfact{k} \times
\sum_{\substack{1\leq r_1<\cdots<r_{l-1}\leq s_2-1 \\ 2\leq t_1<\cdots<t_{k}\leq s_1}}\delta_{l_{s_2}(\varrho),1}\;\delta_{k_1(\vartheta),0}\;
\times \; q^{1-s_2+\sum_{j=1}^{l-1}(1-r_j)+\sum_{j=1}^{k}(1-t_j)} \\
& \qquad\qquad\qquad\qquad\qquad\qquad\qquad \times \hat{\pi}^{(1)}_{s_1}
\left(\Wbas_{l_1(\varrho)} \tens\cdots\tens\Wbas_{l_{s_2-1}(\varrho)} \tens\Wbas_1 \tens\Wbas_0 \tens\Wbas_{k_2(\vartheta)} \tens\cdots
\tens\Wbas_{k_{s_1}(\vartheta)} \right) \\
& + q^{\binom{l}{2}} \qfact{l}\; q^{\binom{k}{2}} \qfact{k} \times 
\sum_{\substack{1\leq r_1<\cdots<r_{l}\leq s_2-1\\ 2\leq t_2<\cdots<t_{k}\leq s_1}}\delta_{l_{s_2}(\varrho),0} \;\delta_{k_1(\vartheta),1}\;
\times\;q^{\sum_{j=1}^{l}(1-r_j)+\sum_{j=2}^{k}(1-t_j)} \\
& \qquad\qquad\qquad\qquad\qquad\qquad\qquad\times\hat{\pi}^{(1)}_{s_1}
\left(\Wbas_{l_1(\varrho)} \tens\cdots\tens\Wbas_{l_{s_2-1}(\varrho)} \tens\Wbas_0 \tens\Wbas_1 \tens\Wbas_{k_2(\vartheta)} \tens\cdots
\tens\Wbas_{k_{s_1}(\vartheta)} \right) \\
= \; & \frac{q^{\binom{l}{2}} \qfact{l}\;
q^{\binom{k}{2}} \qfact{k} \,q^{-s_2}}{\qnum 2}(q-q^{-1}) \times
\sum_{\substack{1\leq r_1<\cdots<r_{l-1}\leq s_2-1 \\ 2\leq t_1<\cdots<t_{k}\leq s_1}}q^{\sum_{j=1}^{l-1}(1-r_j)+\sum_{j=1}^{k}(1-t_j)} \\
& \qquad\qquad\qquad\qquad\qquad\qquad\qquad\qquad\qquad\qquad \times\;
\left(\Wbas_{l_1(\varrho)} \tens\cdots\tens\Wbas_{l_{s_2-1}(\varrho)} \tens\Wbas_{k_2(\vartheta)} \tens\cdots
\tens\Wbas_{k_{s_1}(\vartheta)} \right) \\
& - \frac{q^{\binom{l}{2}} \qfact{l}\;
q^{\binom{k}{2}} \qfact{k}}{\qnum 2}(q-q^{-1}) \times
\sum_{\substack{1\leq r_1<\cdots<r_{l}\leq s_2-1\\ 2\leq t_2<\cdots<t_{k}\leq s_1}}q^{\sum_{j=1}^{l}(1-r_j)+\sum_{j=2}^{k}(1-t_j)} \\
& \qquad\qquad\qquad\qquad\qquad\qquad\qquad\qquad\qquad\qquad\times\;
\left(\Wbas_{l_1(\varrho)} \tens\cdots\tens\Wbas_{l_{s_2-1}(\varrho)} \tens\Wbas_{k_2(\vartheta)} \tens\cdots
\tens\Wbas_{k_{s_1}(\vartheta)} \right).
\end{align*}
Changing the summation indices by $t_j\mapsto t_j-1$ in the first sum and
$r_j\mapsto r_j-1$ in the second sum, and using the formula from 
Lemma~\ref{lem: action of F} for the vectors 
$\MTbas_{l}^{(s_2-1)}$ and $\MTbas_{k}^{(s_1-1)}$, we simplify the above as
\begin{align*}
& \frac{q^{\binom{l}{2}} \qfact{l}\;
q^{\binom{k}{2}} \qfact{k} \,q^{-s_2}}{\qnum 2}(q-q^{-1}) \times
\sum_{\substack{1\leq r_1<\cdots<r_{l-1}\leq s_2-1 \\ 1\leq t_1<\cdots<t_{k}\leq s_1-1}}q^{-k+\sum_{j=1}^{l-1}(1-r_j)+\sum_{j=1}^{k}(1-t_j)} \\
& \qquad\qquad\qquad\qquad\qquad\qquad\qquad\qquad\qquad\qquad\times\;
\left(\Wbas_{l_1(\varrho)} \tens\cdots\tens\Wbas_{l_{s_2-1}(\varrho)} \tens\Wbas_{k_2(\vartheta)} \tens\cdots
\tens\Wbas_{k_{s_1}(\vartheta)} \right) \\
& - \frac{q^{\binom{l}{2}} \qfact{l}\;
q^{\binom{k}{2}} \qfact{k}}{\qnum 2}(q-q^{-1}) \times
\sum_{\substack{1\leq r_1<\cdots<r_{l}\leq s_2-1\\ 1\leq t_2<\cdots<t_{k}\leq s_1-1}}
q^{1-k+\sum_{j=1}^{l}(1-r_j)+\sum_{j=2}^{k}(1-t_j)} \\
& \qquad\qquad\qquad\qquad\qquad\qquad\qquad\qquad\qquad\qquad\times\;
\left(\Wbas_{l_1(\varrho)} \tens\cdots\tens\Wbas_{l_{s_2-1}(\varrho)} \tens\Wbas_{k_2(\vartheta)} \tens\cdots
\tens\Wbas_{k_{s_1}(\vartheta)} \right) \\
= \; & \frac{q^{\binom{l}{2}} \qfact{l}\;
q^{\binom{k}{2}} \qfact{k}}{{\qnum 2}}(q-q^{-1})
\times \left(
\frac{q^{-s_2-k}\,\times \left(\MTbas_{l-1}^{(s_2-1)}\tens\MTbas_{k}^{(s_1-1)}\right)}{q^{\binom{l-1}{2}} q^{\binom{k}{2}} {\qfact{l-1}}{\qfact{k}}}
\; - \; \frac{q^{1-k}\,\times\left(\MTbas_{l}^{(s_2-1)}\tens\MTbas_{k-1}^{(s_1-1)}\right)}{q^{\binom{l}{2}} q^{\binom{k-1}{2}} {\qfact{l}}{\qfact{k-1}}} \right) \\
= \; & \frac{(q-q^{-1})}{\qnum 2} \Big( q^{l-s_2-1-k} \, \qnum {l} \times \left(\MTbas_{l-1}^{(s_2-1)}\tens\MTbas_{k}^{(s_1-1)}\right) 
- \qnum {k} \times \left(\MTbas_{l}^{(s_2-1)}\tens\MTbas_{k-1}^{(s_1-1)}\right) \Big).
\end{align*}
which concludes the proof.
\end{proof}

We generalize the above calculation in the next lemma.

\begin{lem}\label{lem: several projections in the middle}
Let $s_1,s_2 \in \bZpos$ and $m \in \{1,2,\ldots,\min(s_1,s_2) \}$. 
Interpreting $\MTbas_{-1}^{(s)}=0$, we have
\begin{align*}
& \left(\hat{\pi}^{(1)}_{s_1-m+1}\circ\cdots\circ\hat{\pi}^{(1)}_{s_1-1}
\circ\hat{\pi}^{(1)}_{s_1}\right)
\left( \MTbas_{l}^{(s_2)}\tens\MTbas_{k}^{(s_1)} \right) \\
= \; & \frac{(q-q^{-1})^m}{\qnum{2}^{m}} \times
\sum_{j=0}^m \qbin{m}{j} (-1)^j q^{(m-j)(j+l-s_2-1-k)}
\left(\prod_{r=0}^{j-1}\qnum{k-r}\right)
\left(\prod_{s=0}^{m-j-1}\qnum{l-s}\right) \\
& \qquad \qquad \qquad \qquad \qquad \times
\left(\Wbas_{l-m+j}^{(s_2-m)}\tens\Wbas_{k-j}^{(s_1-m)} \right).
\end{align*}
\end{lem}
\begin{proof}
We prove the asserted formula by induction on $m$. The base case is given by
Lemma~\ref{lem: projection in the middle}.
Assume that the asserted formula holds for $m$.
Applying the induction hypothesis and Lemma~\ref{lem: projection in the middle},
we calculate
\begin{align*}
& \left(\hat{\pi}^{(1)}_{s_1-m}\circ\hat{\pi}^{(1)}_{s_1-m+1}\circ\cdots\circ\hat{\pi}^{(1)}_{s_1-1}
\circ\hat{\pi}^{(1)}_{s_1}\right)
\left( \MTbas_{l}^{(s_2)}\tens\MTbas_{k}^{(s_1)} \right) \\
= \; & \frac{(q-q^{-1})^m}{ \qnum{2}^{m}} \times 
\sum_{j=0}^m(-1)^jq^{(m-j)(j+l-s_2-1-k)} \qbin{m}{j} 
\left(\prod_{r=0}^{j-1}\qnum{k-r}\right)
\left(\prod_{s=0}^{m-j-1}\qnum{l-s}\right) \\
& \qquad \qquad \qquad \qquad \qquad \times
\hat{\pi}^{(1)}_{s_1-m} \left( \MTbas_{l-m+j}^{(s_2-m)}\tens\MTbas_{k-j}^{(s_1-m)} \right) \\
= \; & \frac{(q-q^{-1})^{m+1}}{ \qnum{2}^{m+1}} \times 
\sum_{j=0}^m(-1)^jq^{(m-j)(j+l-s_2-1-k)} \qbin{m}{j}
\left(\prod_{r=0}^{j-1}\qnum{k-r}\right)
\left(\prod_{s=0}^{m-j-1}\qnum{l-s}\right) \\
& \qquad \qquad \qquad \qquad \qquad \times
\Big( q^{l-k-1-s_2+2j} \, \qnum {l-m+j} \times \left(\MTbas_{l-(m+1)+j}^{(s_2-(m+1))}\tens\MTbas_{k-j}^{(s_1-(m+1))}\right) \\
& \qquad \qquad \qquad \qquad \qquad \qquad 
- \qnum {k-j} \times \left(\MTbas_{l-m+j}^{(s_2-(m+1))}\tens\MTbas_{k-j-1}^{(s_1-(m+1))}\right) \Big).
\end{align*}
Changing the summation index in the second term by $j\mapsto j-1$, we simplify
the above as
\begin{align*}
& \frac{(q-q^{-1})^{m+1}}{ \qnum{2}^{m+1}} \times 
\sum_{j=0}^{m+1} (-1)^jq^{(m-j)(j+l-s_2-1-k)} \qbin{m}{j}
\left(\prod_{r=0}^{j-1}\qnum{k-r}\right)
\left(\prod_{s=0}^{m-j-1}\qnum{l-s}\right) \\
& \qquad \qquad \qquad \qquad \qquad \times
\Big( q^{l-k-1-s_2+2j} \, \qnum {l-m+j} \times \left(\MTbas_{l-(m+1)+j}^{(s_2-(m+1))}\tens\MTbas_{k-j}^{(s_1-(m+1))}\right) \\
& \qquad \qquad \qquad \qquad \qquad \qquad 
- \qnum {k-j} \times \left(\MTbas_{l-m+j}^{(s_2-(m+1))}\tens\MTbas_{k-j-1}^{(s_1-(m+1))}\right) \Big) \\
= \; & \frac{(q-q^{-1})^{m+1}}{ \qnum{2}^{m+1}} \times 
\sum_{j=0}^{m+1}(-1)^jq^{(m+1-j)(j+l-s_2-1-k)} 
\left(\prod_{r=0}^{j-1}\qnum{k-r}\right)
\left(\prod_{s=0}^{(m+1)-j-1}\qnum{l-s}\right) \\
& \qquad \qquad \qquad \qquad \qquad \times 
\bigg( q^{j}\qbin{m}{j}+q^{-(m+1-j)}\qbin{m}{j-1}\bigg) \times
\left(\MTbas_{l-(m+1)+j}^{(s_2-(m+1))}\tens\MTbas_{k-j}^{(s_1-(m+1))}\right),
\end{align*}
which gives the asserted formula by the recursion of 
Lemma~\ref{lem: q-combinatorics}(a) for the $q$-binomial coefficients.
\end{proof}

The next lemma gives the explicit non-zero constant in 
the commutative diagram in Lemma~\ref{lem: commutative diagram}.

\begin{lem}\label{lem: projections in the middle for hwv}
Let $s_1,s_2 \in \bZpos$ and $m \in \{1,2,\ldots,\min(s_1,s_2) \}$, 
and denote $r = s_1+s_2-2m$ and $\projdmn = r+1$. 
We have
\begin{align*}
\left(\hat{\pi}^{(1)}_{s_1-m+1}\circ\cdots\circ\hat{\pi}^{(1)}_{s_1-1}
\circ\hat{\pi}^{(1)}_{s_1}\right)
\left( (\Embedding^{(s_2)}\tens\Embedding^{(s_1)})
(\Tbas_{0}^{(\projdmn;d_{1},d_{2})}) \right)
= \; & \constantfromdiagram{m}{s_1}{s_{2}} \times 
\Embedding^{(r)}(\Wbas_0^{(\projdmn)}),
\end{align*}
where
\[ \constantfromdiagram{m}{s_1}{s_{2}} =
\frac{\qfact{s_1-m}\qfact{s_2-m}\qfact{s_1+s_2-m+1}}{\qnum{2}^m\qfact{s_1}\qfact{s_2}\qfact{s_1+s_2-2m+1}}
= \frac{\qbin{s_1+s_2-m+1}{m}}{\qnum{2}^m\qfact{m}\qbin{s_1}{m}\qbin{s_2}{m}} 
\quad \neq 0. \]
\end{lem}
\begin{proof}
Recall from Lemma~\ref{lem: tensor product representations of quantum sl2}
the formulas~\eqref{eq: tensor product hwv}:
\begin{align} 
\Tbas_{0}^{(\projdmn;d_{1},d_{2})} = \; &
\sum_{l_{1},l_{2}} T_{0;m}^{l_{1},l_{2}}(s_{1},s_{2})
\times (\Wbas_{l_{2}}^{(d_2)} \tens \Wbas_{l_{1}}^{(d_1)}) ,
    \label{eq: coef T}\\
\text{where } \qquad
T_{0;m}^{l_{1},l_{2}}(s_{1},s_{2}) =\; &
    \delta_{l_{1}+l_{2},m} \times (-1)^{l_{1}}\frac{\qfact{s_{1}-l_{1}}\,\qfact{s_{2}-l_{2}}}{\qfact{l_{1}}\qfact{s_{1}}\qfact{l_{2}}\qfact{s_{2}}}\,\frac{q^{l_{1}(s_{1}-l_{1}+1)}}{(q-q^{-1})^{m}} . \nonumber 
\end{align}
By Lemma~\ref{lem: several projections in the middle}, we have
\begin{align*}
& \left(\hat{\pi}^{(1)}_{s_1-m+1}\circ\cdots\circ\hat{\pi}^{(1)}_{s_1-1}
\circ\hat{\pi}^{(1)}_{s_1}\right)
\left( (\Embedding^{(s_2)}\tens\Embedding^{(s_1)})
(\Tbas_{0}^{(\projdmn;d_{1},d_{2})}) \right) \\
= \; & \frac{(q-q^{-1})^m}{\qnum{2}^{m}} \times
\sum_{l_{1},l_{2}} T_{0;m}^{l_{1},l_{2}}(s_{1},s_{2}) \times
\sum_{j=0}^m \qbin{m}{j} (-1)^j q^{(m-j)(j+l_2-s_2-1-l_1)}  \\
\; & \qquad \qquad \qquad \qquad \times 
\left(\prod_{r=0}^{j-1}\qnum{l_1-r}\right)
\left(\prod_{s=0}^{m-j-1}\qnum{l_2-s}\right)\times
\left(\MTbas_{l_2-m+j}^{(s_2-m)}\tens\MTbas_{l_1-j}^{(s_1-m)}\right).
\end{align*}
Now, by the formula~\eqref{eq: coef T}, we have $l_{1}+l_{2}=m$ in the sum.
Therefore, only the terms with $j=l_1=m-l_2$ are non-zero. We denote $k = l_1$
and simplify the above expression as
\begin{align*}
& \frac{(q-q^{-1})^m}{\qnum{2}^{m}} \times
\sum_{k=0}^m (-1)^{k}\frac{\qfact{s_{1}-k}\,\qfact{s_{2}-m+k}}{\qfact{k}\qfact{s_{1}}\qfact{m-k}\qfact{s_{2}}}\,\frac{q^{k(s_{1}-k+1)}}{(q-q^{-1})^{m}} \times
\qbin{m}{k} (-1)^{k} q^{(m-k)(m-k-s_2-1)}  \\
\; & \qquad \qquad \qquad \qquad \times 
\left(\prod_{r=0}^{k-1}\qnum{k-r}\right)
\left(\prod_{s=0}^{m-k-1}\qnum{m-k-s}\right)\times
\left(\MTbas_{0}^{(s_2-m)}\tens\MTbas_{0}^{(s_1-m)}\right) \\
= \; & \frac{q^{m(m-s_2-1)}}{\qnum{2}^{m}\qfact{s_{1}}\qfact{s_{2}}} \times
\sum_{k=0}^m \qbin{m}{k} q^{k(s_{1}+s_{2}-2m+2)} \qfact{s_{1}-k} \qfact{s_{2}-m+k}
 \; \times \Embedding^{(r)}(\Wbas_0^{(\projdmn)}).
\end{align*}
Using Lemma~\ref{lem: q-combinatorics}(d), with $\nu_i = s_i$, $i=1,2$,
and $n = m$, we simplify this to the asserted form.
\end{proof}

\bigskip{}
\section{\label{app: dual elements}Dual elements}

This appendix contains results needed in the proof of 
Lemma~\ref{lem: limits at same rate} and Proposition~\ref{prop: strong limits}, concerning the limit of
the solution $\BasisF_\linkpatt$ of the Benoit~\& Saint-Aubin PDE system~\eqref{eq: BSA differential equations}
as several of its variables tend to a common limit simultaneously.
The core idea in the proof is to construct suitable dual elements which allow us 
to evaluate the limit. The same idea was also used in a simpler setup in 
the proof of Proposition~\ref{prop: FK dual elements},
where we constructed dual elements for the basis functions 
$\BasisF_\alpha$, for $\alpha \in \PP_N^{(s)}$, as iterated limits.

Using the projection properties~\eqref{eq: projection conditions} of the vectors
$\Puregeom_\linkpatt$, we will define iterated projections, which provide the 
(unnormalized) dual basis of 
$\Puregeom_\linkpatt \in \HWsp^{(s)}_{\multii}$.
We follow the approach of
\cite[Section~3.5]{Kytola-Peltola:Pure_partition_functions_of_multiple_SLEs},
where such dual elements for the special case of
$\Puregeomtwodim_\alpha \in \HWsp_{2N}^{(0)}$, for $\alpha \in \PP_N$, 
were constructed.
Therefore, we only give the rough reasoning of the general case
--- the details are the same as in \cite[Section~3.5]{Kytola-Peltola:Pure_partition_functions_of_multiple_SLEs}, 
but the notation for this simple construction becomes unnecessarily complicated in the general case.

\subsection{Allowable orderings of links}

We consider the links in the link pattern 
\begin{align*} 
\linkpatt = \Big\{\linkInEquation{a_1}{b_1},\ldots,\linkInEquation{a_\ell}{b_\ell}\Big\}
\bigcup
\Big\{\defectInEquation{c_1},\ldots,\defectInEquation{c_{s}}\Big\},
\end{align*} 
as an ordered multiset of $k \leq \ell = \sum_{a,b} \ell_{a,b}$ elements,
\begin{align}\label{eq: multiset of links}
\mathfrak{L}(\linkpatt) = 
\Big\{\ell_{a_1,b_1} \times \linkInEquation{a_1}{b_1},\; \ldots, \;\ell_{a_k,b_k} \times \linkInEquation{a_k}{b_k}\Big\} .
\end{align}

For instance, if $\multii = (1,1,\ldots,1,1)$, that is, all indices of $\linkpatt$ 
have valence one, then $k = \ell$ and the links can be ordered by their left 
endpoints, such that $a_1 < a_2 < \cdots < a_\ell$. If some other ordering is chosen, 
there is a permutation $\sigma \in \SymmGrp_\ell$ such that we have
$a_{\sigma(1)} < a_{\sigma(2)} < \cdots < a_{\sigma(\ell)}$.
The choice of the ordering of the links is thus encoded in
the unique permutation $\sigma$ with the above property.

For link patterns with $s_j \geq 2$, for some $j$ in 
$\multii = (s_1,\ldots,s_p)$,
the ordering of the links amounts to ordering the multiset $\mathfrak{L}(\linkpatt)$.
For example, we can first order the links in groups by their left endpoints as 
above, and then, in each group of links with the same left endpoint, we may
choose the ordering according to the right endpoint so that the link(s) 
$\link{a\;}{\;b\,}$ among the group with the smallest $b$ get the smallest running number. 
Again, choosing some other ordering amounts to choosing a permutation 
$\sigma \in \SymmGrp_k$ of the multiset of the links.

Recall that the removal of $m \leq \ell_{j,j+1}$ links $\link{j}{j+1}$ from 
$\linkpatt$ is denoted by $\linkpatt\removeLink (m\times\link{j}{j+1})$,
and if $s_j=m$ or $s_{j+1}=m$, we also have to remove the index $j$ or $j+1$,
respectively (or both), and relabel the indices of the 
remaining links and defects as illustrated in Figure~\ref{fig: removing links}.
Slightly informally, we say that the ordering of the links is allowable 
for $\linkpatt$ if all links of $\linkpatt$ can be removed in such a way that at 
each step, the links $\ell_{a,b} \times \link{a\;}{\;b\,}$
to be removed connect two consecutive indices $a = j$ and $b = j+1$ 
(when the indices are relabeled after each removal) --- see 
Figure~\ref{fig: allowable removal of links with fusion} for an illustration.
The concept of ``allowability'' was defined more formally in
\cite[Section~3.5]{Kytola-Peltola:Pure_partition_functions_of_multiple_SLEs}
in the special case of $\alpha \in \PP_N$, but the only differences in the present 
case are that, first, the links come with multiplicity, which only results in 
complications in the notation, and, second, $\linkpatt$ might have defects
$\defect{c}$, which play no role in the link removal and cannot lie inside 
any link $\link{a\;}{\;b\,}$ in the sense that $a < c < b$.

\begin{figure}
\includegraphics[scale=.75]{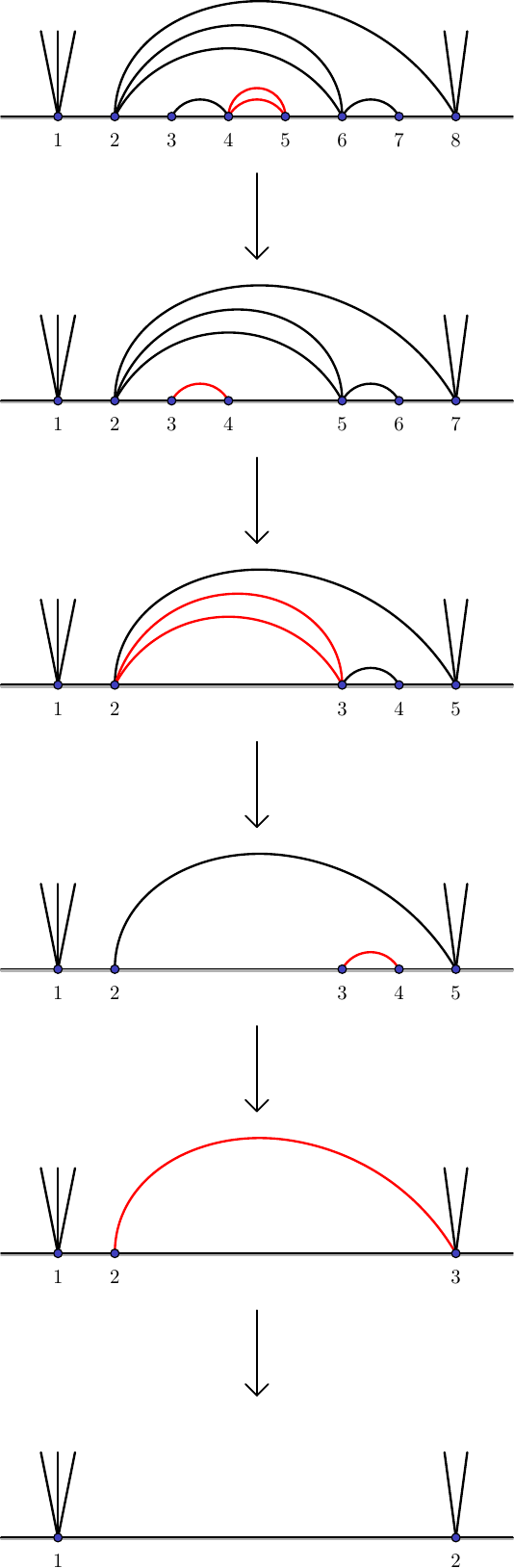}
\caption{\label{fig: allowable removal of links with fusion}
Example of an allowable ordering to remove links from 
a planar link pattern. Notice in particular the relabeling of the indices after 
each step, and the fact that after each removal,
the valences of the endpoints of the removed link decrease, as in
Equation~\eqref{eq: notations for interated projections}.
Furthermore, if after the removal the endpoint becomes empty, 
then it is removed as well. After the removal of all the links, the defects remain.}
\end{figure}

\subsection{\label{subsec: dual elements}Dual elements}

Let $\linkpatt \in \LP_\multii^{(s)}$ and suppose $\sigma \in \SymmGrp_k$
is an allowable ordering of the $k$ links of $\linkpatt$ (with multiplicity),
see~\eqref{eq: multiset of links} and Figure~\ref{fig: allowable removal of links with fusion}.
After removal of all the links of $\linkpatt$ in the order $\sigma$,
one is left with the link pattern $\defpatt_{\partition(\linkpatt)}$, which 
consists of $s$ defects only, determined by the partition 
$\partition(\linkpatt)$ of the defects of $\linkpatt$ 
(recall Section~\ref{subsec: removing links}).
In terms of the vector $\Puregeom_\linkpatt$, the link removal can be realized
as projections to subrepresentations, 
using the properties~\eqref{eq: projection conditions} of $\Puregeom_\linkpatt$
--- since the ordering $\sigma$ is allowable, the links are removed in such a way
that the removed links always connect two consecutive indices $j,j+1$.

As in Remark~\ref{rem: one dimensionality of partition space}, we identify the one-dimensional space
$\HWsp_{\partition(\linkpatt)}^{(s)}$ with $\bC$, via
$\Puregeom_{\defpatt_{\partition(\linkpatt)}} \mapsto 1$.
This identification is implicitly used in the following definition.
We set
\begin{align*}
\Quantumdual_\linkpatt^{(\sigma)} \;\colon\;
\HWsp_\multii^{(s)} \to \bC, \qquad \qquad
\Quantumdual_\linkpatt^{(\sigma)} := \; &
\tilde{\pi}_{a_k(k-1)}^{(\projdmn_k)}
\circ \cdots \circ \tilde{\pi}_{a_2(1)}^{(\projdmn_2)}\circ \tilde{\pi}_{a_1}^{(\projdmn_1)} ,
\end{align*}
where 
\begin{align}\label{eq: notations for interated projections}
\projdmn_j := \; & d_{b_j}^{(j-1)}+d_{a_j}^{(j-1)}-1-2\ell_{a_j,b_j},
\qquad \qquad 
d_{c}^{(j)} := d_{c} - 
\sum_{\substack{ i \in \{1,2,\ldots,j \}\\ c = a_i \text{ or } c = b_i}} \ell_{a_i,b_i}
\end{align}
so that $d_{a}^{(j)}-1$ denotes the valence of the point $a$ 
after removal of the $j$ links 
$\ell_{a_1,b_1} \times \link{a_1}{b_1}$, $\ldots$, 
$\ell_{a_j,b_j} \times \link{a_j}{b_j}$ from $\linkpatt$ in the order $\sigma$,
and $a_j(j-1)$ denotes the relabeled endpoint of the $j$:th link 
$\link{a_j}{b_j}$ after removal of the $j-1$ links 
$\ell_{a_1,b_1} \times \link{a_1}{b_1}$, $\ldots$, 
$\ell_{a_{j-1},b_{j-1}} \times \link{a_{j-1}}{b_{j-1}}$;
see also Figure~\ref{fig: allowable removal of links with fusion}.

We next show that $\Quantumdual_\linkpatt^{(\sigma)}$ 
is in fact independent of the choice of allowable ordering $\sigma$ 
for $\linkpatt$, and thus gives
rise to a well-defined linear map
\begin{align}\label{eq: dual basis without ordering}
\Quantumdual_\linkpatt:=\Quantumdual_\linkpatt^{(\sigma)}
\;\colon\; \HWsp_\multii^{(s)} \to \bC,
\end{align}
for any choice of
allowable ordering $\sigma$ of the links in $\linkpatt$.
Moreover, we show that
$\left(\Quantumdual_\linkpatt\right)_{\linkpatt \in \LP^{(s)}_{\multii}}$
is a basis of the dual space $(\HWsp_\multii^{(s)})^*$, namely
the (unnormalized) dual basis of 
$\left(\Puregeom_\linkpatt\right)_{\linkpatt \in \LP^{(s)}_{\multii}}$.

\begin{prop}\label{prop: quantum dual elements}
{(see also \cite[Proposition~3.7]{Kytola-Peltola:Pure_partition_functions_of_multiple_SLEs})}
\
\begin{description}
\item[(a)] Let $\linkpatt\in\LP^{(s)}_{\multii}$.
For any two allowable orderings $\sigma$ and $\sigma'$ of the links in $\linkpatt$, we have
 \begin{align*}
\Quantumdual_\linkpatt^{(\sigma)} = \Quantumdual_\linkpatt^{(\sigma')}.
\end{align*}
Thus, the linear functional
$\Quantumdual_\linkpatt \in (\HWsp_\multii^{(s)})^*$ 
in \eqref{eq: dual basis without ordering} is well-defined. 

\item[(b)]
For any $\linkpatt,\tau\in\LP^{(s)}_{\multii}$, we have
\begin{align}\label{eq: quantum dual basis}
\Quantumdual_\linkpatt(\Puregeom_\tau) = 
\const\times\delta_{\linkpatt,\tau} , \qquad\text{where}\qquad
\delta_{\linkpatt,\tau} = \; & \begin{cases}
1\quad & \text{if }\tau=\linkpatt\\
0 & \text{if }\tau\neq\linkpatt ,
\end{cases}
\end{align}
and the constant is non-zero and depends only on $\linkpatt$.
\end{description}
In particular, 
$\left(\Quantumdual_{\linkpatt}\right)_{\linkpatt\in\LP^{(s)}_{\multii}}$ 
is a basis of the dual space $(\HWsp_\multii^{(s)})^*$.
\end{prop}
\begin{proof}
We use the notations introduced 
in Equation~\eqref{eq: notations for interated projections}.
Let $\linkpatt,\tau \in \LP_\multii^{(s)}$, and let $\sigma$ be 
any allowable ordering of the links of $\linkpatt$.
Consider $\Quantumdual_\linkpatt^{(\sigma)}( \Puregeom_\tau)$. If $\tau=\linkpatt$, then
by the projection property~\eqref{eq: projection conditions}, we have
\begin{align*}
\tilde{\pi}_{a_1}^{(\projdmn_1)} ( \Puregeom_\linkpatt) =
\frac{1}{\constantfromdiagram{\ell_{a_1,b_1}}{s_{b_1}}{s_{a_1}}}
\times \Puregeom_{\linkpatt\removeLink(\ell_{a_1,b_1}\times\link{a_1}{b_1})},
\end{align*}
and recursively, 
\begin{align*}
\Big( \tilde{\pi}_{a_j(j-1)}^{(\projdmn_j)}
\circ \cdots \circ \tilde{\pi}_{a_2(1)}^{(\projdmn_2)}\circ \tilde{\pi}_{a_1}^{(\projdmn_1)} \Big)( \Puregeom_\linkpatt)
    = \const\times \Puregeom_{\linkpatt\removeLink(\ell_{a_1,b_1}\times\link{a_1}{b_1})\removeLink\cdots \removeLink(\ell_{a_j,b_j}\times\link{a_j}{b_j}) } ,
\end{align*}
for a non-zero constant which is a product of the constants appearing in the 
projection properties~\eqref{eq: projection conditions}.

For $j=k$, the above formula gives 
$\Quantumdual_\linkpatt^{(\sigma)}( \Puregeom_\linkpatt) = 
\const \times \Puregeom_{\defpatt_{\partition(\linkpatt)}}$,
which we identify with the constant times $1 \in \bC$,
via $\Puregeom_{\defpatt_{\partition(\linkpatt)}} \mapsto 1$, as in
Remark~\ref{rem: one dimensionality of partition space}.
On the other hand, if $\tau\neq\linkpatt$, then for some $j$, the link pattern 
$\tau$ does not contain $\ell_{a_j,b_j}$ links $\link{a_j}{b_j}$,
and by the property~\eqref{eq: projection conditions} we then similarly get $\Quantumdual_\linkpatt^{(\sigma)}( \Puregeom_\tau) = 0$.
Summarizing, we have
$\Quantumdual_\linkpatt( \Puregeom_\tau) = \const\times\delta_{\linkpatt,\tau}$,
independently of the choice of allowable ordering $\sigma$, and the constant is 
non-zero and only depends on $\linkpatt$. This proves
Equation~\eqref{eq: quantum dual basis} and assertion (b).

By Proposition~\ref{prop: basis of highest weight vector space}, 
the vectors $\Puregeom_\tau$, with $\tau \in \LP_\multii^{(s)}$,
form a basis of the space $\HWsp_\multii^{(s)}$.
It thus follows from Equation~\eqref{eq: quantum dual basis}
that the value of the operator $\Quantumdual_\linkpatt^{(\sigma)}$ 
is independent of the choice of an allowable ordering $\sigma$ of the links,
and that
$\left(\Quantumdual_{\linkpatt}\right)_{\linkpatt\in\LP^{(s)}_{\multii}}$ 
is a basis of the dual space $(\HWsp_\multii^{(s)})^*$.
This concludes the proof.
\end{proof}

\begin{rem}\label{rem: elements of the commutant}
\emph{For fixed $\linkpatt$, by Theorem~\ref{thm: highest weight vector space basis vectors}(b), the maps 
$\Quantumdual_\linkpatt \colon \HWsp_\multii^{(s)} \to \bC$
also define (unnormalized) projectors 
\begin{align*}
\widehat{\Quantumdual}_\linkpatt \; \colon \; \bigotimes_{i=1}^{p}\Wd_{d_{i}} \to \Wd_d, \qquad
\widehat{\Quantumdual}_\linkpatt(F^l.\Puregeom_\tau) := \; & \begin{cases}
 \Quantumdual_\linkpatt(\Puregeom_\tau) \times \Wbas_l^{(d)} ,
 \quad \text{ for any } l \in \{ 0,1,\ldots,s \} \quad &  \text{if } 
 \tau \in \LP^{(s)}_{\multii} \\
0 & \text{otherwise,}
\end{cases}
\end{align*}
from the tensor product~\eqref{eq: order of tensorands}
onto the $s + 1 = d$-dimensional irreducible representation $\Wd_d$ 
of $\Uqsltwo$. 
For a chosen $\upsilon \in \LP^{(s)}_{\multii}$,  
combining $\widehat{\Quantumdual}_\linkpatt$ with the embedding 
$\Wd_d \hookrightarrow \bigotimes_{i=1}^{p}\Wd_{d_{i}}$ given by
$\Wbas_l^{(d)} \mapsto F^l.\Puregeom_{\upsilon}$,
we can define the (unnormalized) projectors
\begin{align*}
\widetilde{\Quantumdual}_\linkpatt^{\upsilon} \; \colon \; \bigotimes_{i=1}^{p}\Wd_{d_{i}} \to \bigotimes_{i=1}^{p}\Wd_{d_{i}}, \qquad
\widetilde{\Quantumdual}_\linkpatt^{\upsilon}(F^l.\Puregeom_\tau) := \; & \begin{cases}
 \Quantumdual_\linkpatt(\Puregeom_\tau) \times F^l.\Puregeom_{\upsilon} ,
 \quad \text{ for any } l \in \{0,1,\ldots,s \} \quad &  \text{if } 
 \tau \in \LP^{(s)}_{\multii} \\
0 & \text{otherwise,}
\end{cases}
\end{align*}
onto the subrepresentations of the tensor product~\eqref{eq: order of tensorands}
isomorphic to $\Wd_d$, generated by $\Puregeom_{\upsilon}$.
This gives rise to 
$\sum_{s \geq 0}(\# \LP_\multii^{(s)})^2$ linearly independent maps
$\widetilde{\Quantumdual}_\linkpatt^{\upsilon}$, with
$\linkpatt,\upsilon \in \LP^{(s)}_{\multii}$,
that belong to the commutant algebra
$\End_{\Uqsltwo}\big( \bigotimes_{i=1}^{p}\Wd_{d_{i}} \big)$.
We discuss this commutant algebra in forthcoming work~\cite{Flores-Peltola:Colored_braid_representations_and_QSW}.
}
\end{rem}

\subsection{\label{subsec: proof of prop limits at same rate}Some details for the proofs of Lemma~\ref{lem: limits at same rate} and Proposition~\ref{prop: strong limits}}

Let $1 \leq j < k \leq p$ and $\linkpatt\in\LP^{(s)}_{\multii}$, 
and let $\tau \in \LP_{\multii_{j,k}}^{(r)}$ be the sub-link pattern of 
$\linkpatt$ with index valences $\multii_{j,k} = (s_j,s_{j+1}\ldots,s_k)$, 
consisting of the lines of $\linkpatt$ attached to the indices $j,j+1,\ldots,k$, 
as in Section~\ref{subsec: further limit properties}, and let 
$\linkpatt \removeLink \tau$ denote the link pattern obtained from $\linkpatt$
by ``removing $\tau$'', that is, removing from $\linkpatt$ the links 
$\link{a\;}{\;b\,}$ with indices $a,b \in \{ j,j+1,\ldots,k \}$, collapsing
the indices $j,j+1,\ldots,k$ of $\linkpatt$ into one point, and relabeling
the indices thus obtained from left to right by $1,2,\ldots$
(see Section~\ref{subsec: further limit properties}).

\begin{lem}\label{lem: details for proof of further limit properties}
Let $\linkpatt \in \LP_\multii^{(s)}$, 
$\tau \in \LP_{\multii_{j,k}}^{(r)}$, 
and $\linkpatt \removeLink \tau$ be as in 
Section~\ref{subsec: further limit properties}. Then we have 
\begin{align*}
\Puregeom_\linkpatt
= \; & \sum_{l=0}^r \sum_{\substack{l_1,\ldots,l_{j-1},\\l_{k+1},\ldots,l_p}}
c_{l_1,\ldots,l_{j-1};l;l_{k+1},\ldots,l_p} \times 
\left(\Wbas_{l_p} \tens \cdots \tens \Wbas_{l_{k+1}}
\tens F^l.\Puregeom_{\tau} \tens \Wbas_{l_{j-1}}
\tens \cdots \tens \Wbas_{l_1}\right) \\
\Puregeom_{\linkpatt \removeLink \tau} 
= \; & \sum_{l=0}^r \sum_{\substack{l_1,\ldots,l_{j-1},\\l_{k+1},\ldots,l_p}}
c_{l_1,\ldots,l_{j-1};l;l_{k+1},\ldots,l_p} \times
\left(\Wbas_{l_p} \tens \cdots \tens \Wbas_{l_{k+1}}
\tens \Wbas_l \tens \Wbas_{l_{j-1}}
\tens \cdots \tens \Wbas_{l_1}\right) ,
\end{align*}
for some constants $c_{l_1,\ldots,l_{j-1};l;l_{k+1},\ldots,l_p} \in \bC$.
\end{lem}
\begin{proof}
By Theorem~\ref{thm: highest weight vector space basis vectors}(b),
the vector $\Puregeom_\linkpatt$ can be written as a linear combination
\[ \Puregeom_\linkpatt = 
\sum_{t \geq 0} \sum_{\upsilon \in \LP^{(t)}_{\multii_{j,k}}} 
\sum_{l=0}^t \sum_{\substack{l_1,\ldots,l_{j-1},\\l_{k+1},\ldots,l_p}}
c^{t,\upsilon}_{l_1,\ldots,l_{j-1};l;l_{k+1},\ldots,l_p} \times 
\left(\Wbas_{l_p} \tens \cdots \tens \Wbas_{l_{k+1}}
\tens F^l.\Puregeom_{\upsilon} \tens \Wbas_{l_{j-1}}
\tens \cdots \tens \Wbas_{l_1}\right), \]
for some constants $c^{t,\upsilon}_{l_1,\ldots,l_{j-1};l;l_{k+1},\ldots,l_p} \in \bC$.
For any $\upsilon \in \LP^{(t)}_{\multii_{j,k}}$, we apply the map 
\[ (\widehat{\Quantumdual}_{\upsilon})_{j,k} = \id^{\tens(p-k)} \tens \widehat{\Quantumdual}_{\upsilon} \tens \id^{\tens(j-1)} \;\colon\; \bigotimes_{i=1}^{p}\Wd_{d_{i}} \; \longrightarrow \; 
(\Wd_{d_p} \tens \cdots \tens \Wd_{d_{k+1}}) \tens \Wd_{t+1} \tens
(\Wd_{d_{j-1}} \tens \cdots \tens \Wd_{d_1}) \]
to both sides of the above expression for $\Puregeom_\linkpatt$.
By the projection properties~\eqref{eq: projection conditions} of 
$\Puregeom_\linkpatt$, the vector
$(\widehat{\Quantumdual}_{\upsilon})_{j,k}(\Puregeom_\linkpatt)$
equals zero unless $\upsilon = \tau$, and if $\upsilon = \tau$, then we have 
$(\widehat{\Quantumdual}_{\tau})_{j,k}(\Puregeom_\linkpatt) = 
\Quantumdual_\tau(\Puregeom_\tau) \times \Puregeom_{\linkpatt \removeLink \tau}$,
by similar arguments as in the proof of 
Proposition~\ref{prop: quantum dual elements}. 
Analogous properties hold for
$(\widetilde{\Quantumdual}_{\upsilon}^{\upsilon})_{j,k}$, which picks the 
component generated by $\Puregeom_\upsilon$ in the tensor positions 
$j,j+1,\ldots,k$. Therefore, we have
\[ \Puregeom_\linkpatt
= \sum_{l=0}^r \sum_{\substack{l_1,\ldots,l_{j-1},\\l_{k+1},\ldots,l_p}}
c_{l_1,\ldots,l_{j-1};l;l_{k+1},\ldots,l_p} \times 
\left(\Wbas_{l_p} \tens \cdots \tens \Wbas_{l_{k+1}}
\tens F^l.\Puregeom_{\tau} \tens \Wbas_{l_{j-1}}
\tens \cdots \tens \Wbas_{l_1}\right), \]
where $c_{l_1,\ldots,l_{j-1};l;l_{k+1},\ldots,l_p} = c^{r,\tau}_{l_1,\ldots,l_{j-1};l;l_{k+1},\ldots,l_p}$,
and in particular,
\[ \Puregeom_{\linkpatt \removeLink \tau} 
= \sum_{l=0}^r \sum_{\substack{l_1,\ldots,l_{j-1},\\l_{k+1},\ldots,l_p}}
c_{l_1,\ldots,l_{j-1};l;l_{k+1},\ldots,l_p} \times
\left(\Wbas_{l_p} \tens \cdots \tens \Wbas_{l_{k+1}}
\tens \Wbas_l \tens \Wbas_{l_{j-1}}
\tens \cdots \tens \Wbas_{l_1}\right). \]
\end{proof}

\bigskip{}

\bibliographystyle{annotate}

\newcommand{\etalchar}[1]{$^{#1}$}

\end{document}